\documentclass[12pt,fleqn]{article}
\usepackage{amsfonts,url}
\usepackage{amssymb}
\usepackage{amsmath}
\usepackage{amstext,verbatim,graphicx,ifthen,multirow,amsthm}
\usepackage{enumerate}
\usepackage{bm}
\usepackage{natbib}
\usepackage[bf]{caption}
\usepackage{booktabs}
\usepackage{float}
\usepackage{setspace}
\usepackage[pdftex,colorlinks=true,unicode,bookmarksnumbered=false, hyperfootnotes=true]{hyperref}
\usepackage{algorithm}
\usepackage[noend]{algpseudocode}
\usepackage{float}
\hypersetup{citecolor = blue}

\usepackage[colorinlistoftodos]{todonotes}

\usepackage{mathtools}
\usepackage{amsmath,amssymb}
\usepackage{multirow}
\usepackage{graphicx}
\usepackage{subcaption}
\graphicspath{ {./graph/} }

\renewcommand{\mathbf}{\boldsymbol}
\renewcommand{\thepage}{}

\newcommand{\E}{\mathbb{E}}
\newcommand{\e}{\mathbf{e}}
\newcommand{\Var}{\mathrm{Var}}

\newcommand{\R}{\mathbb{R}}

\renewcommand{\P}{\mathbb{P}}
\newcommand{\td}{\tilde}
\newcommand{\la}{\left\langle}
\newcommand{\ra}{\right\rangle}
\newcommand{\lb}{\left(}
\newcommand{\rb}{\right)}
\newcommand{\eps}{\epsilon}
\newcommand{\cW}{\mathcal{W}}
\newcommand{\cU}{\mathcal{U}}
\newcommand{\w}{\mathbf{w}}

\renewcommand{\H}{\mathbf{H}}
\newcommand{\Y}{\mathbf{Y}}
\newcommand{\X}{\mathbf{X}}
\newcommand{\W}{\mathbf{W}}
\newcommand{\G}{\mathbf{G}}

\newcommand{\Z}{\mathcal{Z}}
\renewcommand{\L}{\mathbf{L}}

\newcommand{\cD}{\mathcal{D}}
\renewcommand{\S}{\mathbb{S}}

\newcommand{\J}{J}

\newcommand{\I}{\mathcal{I}}

\newcommand{\V}{\mathcal{V}}

\renewcommand{\J}{J}

\newcommand{\bfm}{\mathbf{m}}

\newcommand{\bfp}{\bm{p}}
\newcommand{\bfb}{\bm{b}}
\newcommand{\bfmu}{\bm{\mu}}
\newcommand{\bfpi}{\bm{\pi}}
\newcommand{\bfPi}{\bm{\Pi}}
\newcommand{\bfgamma}{\bm{\gamma}}
\newcommand{\bfpsi}{\bm{\psi}}
\newcommand{\bfphi}{\bm{\phi}}
\newcommand{\bftau}{\bm{\tau}}

\newcommand{\bfeta}{\bm{\eta}}

\newcommand{\bfomega}{\bm{\omega}}
\newcommand{\bftheta}{\bm{\theta}}
\newcommand{\bfkappa}{\bm{\kappa}}
\newcommand{\one}{\textbf{1}}
\newcommand{\zero}{\textbf{0}}
\newcommand{\tran}{\top}

\newcommand{\op}{\mathrm{op}}

\newcommand{\Unif}{\mathrm{Unif}}
\newcommand{\bias}{\mathrm{Bias}_{\tau}}
\newcommand{\EPi}{\E_{\W\sim \bfPi}}
\newcommand{\DATEeq}{DATE equation}
\newcommand{\hbftau}{\hat{\bm{\nu}}}
\newcommand{\htau}{\hat{\nu}}
\newcommand{\nhbftau}{\bm{\nu}}
\newcommand{\nhtau}{\nu}
\newcommand{\bfA}{\bm{A}}
\newcommand{\bfB}{\bm{B}}
\newcommand{\bfC}{\bm{C}}

\newcommand{\sww}{\Gamma_{ww}}
\newcommand{\sw}{\mathbf{\Gamma}_{w}}
\newcommand{\swy}{\Gamma_{wy}}
\newcommand{\sy}{\mathbf{\Gamma}_{y}}
\newcommand{\stheta}{\Gamma_{\theta}}
\newcommand{\swwex}{\mathbf{\Gamma}_{ww, \ex}}
\newcommand{\swex}{\mathbf{\Gamma}_{w, \ex}}
\newcommand{\swyex}{\mathbf{\Gamma}_{wy, \ex}}

\newcommand{\numer}{N_{*}}

\newcommand{\hV}{\hat{\V}}

\newcommand{\hatdenom}{\mathcal{D}}
\newcommand{\hatnumer}{\mathcal{N}}
\newcommand{\ex}{\text{ex}}

\newcommand{\indep}{\perp\!\!\!\perp}

\newcommand{\q}{q}
\renewcommand{\r}{r}
\newcommand{\dcv}{\stackrel{d}{\rightarrow}}
\newcommand{\bfa}{\mathbf{a}}

\newtheorem{assumption}{Assumption}[section]

\newtheorem{theorem}{Theorem}[section]
\newtheorem{lemma}{Lemma}[section]
\newtheorem{proposition}{Proposition}[section]
\newtheorem{corollary}{Corollary}[section]

\theoremstyle{remark}
\newtheorem{remark}{Remark}[section]

\allowdisplaybreaks

\def\monthname{\ifcase\month\or
  January\or February\or March\or April\or May\or June\or July\or
  August\or September\or October\or November\or December\fi}
\numberwithin{equation}{section}

\DeclareMathOperator*{\argmin}{arg\,min}
\DeclareMathOperator*{\diag}{diag}

\DeclareMathOperator*{\Cov}{Cov}

\def\monthname{\ifcase\month\or
January\or February\or March\or April\or May\or June\or
July\or August\or September\or October\or November\or December\fi}
\makeatletter
\def\blfootnote{\gdef\@thefnmark{}\@footnotetext}
\makeatother

\begin{document}

\title{\textbf{Design-Robust Two-Way-Fixed-Effects Regression For Panel
Data
}\blfootnote{{\small Generous support from the Office of Naval Research through ONR grants N00014-17-1-2131 and 
 N00014-19-1-2468
is gratefully acknowledged.}}}
\author{Dmitry  Arkhangelsky \thanks{{\small  Associate Professor, CEMFI, darkhangel@cemfi.es. }} \and Guido W. Imbens\thanks{{\small Professor of
Economics,
Graduate School of Business and Department of Economics, Stanford University, SIEPR, and NBER,
imbens@stanford.edu.}} 
\and Lihua Lei\thanks{{\small Postdoctoral Fellow,
 Department of Statistics, Stanford University, 
lihualei@stanford.edu.}}
\and Xiaoman Luo\thanks{{\small Ph.D. Candidate, 
~Department of Agricultural and Resource Economics,~University of California, Davis,~xmluo@ucdavis.edu}}}
\date{\ifcase\month\or
January\or February\or March\or April\or May\or June\or
July\or August\or September\or October\or November\or December\fi \ \number%
\year\ \ }
\maketitle\thispagestyle{empty}

\begin{abstract}
\singlespacing
\noindent
We propose a new estimator for average causal effects of a binary treatment with panel data in settings with general treatment patterns. Our approach augments the popular two-way-fixed-effects specification with unit-specific weights that arise from a model for the assignment mechanism. We show how to construct these weights in various settings, including the staggered adoption setting, where units opt into the treatment sequentially but permanently. The resulting estimator converges to an average (over units and time) treatment effect under the correct specification of the assignment model, even if the fixed effect model is misspecified. We show that our estimator is more robust than the conventional two-way estimator: it remains consistent if either the assignment mechanism or the two-way regression model is correctly specified. In addition, the proposed estimator performs better than the two-way-fixed-effect estimator if the outcome model and assignment mechanism are locally misspecified. This strong robustness property underlines and quantifies the benefits of modeling the assignment process and motivates using our estimator in practice. We also discuss an extension of our estimator to handle dynamic treatment effects.
\end{abstract}


\begin{center}
\end{center}



\baselineskip=20pt
\setcounter{page}{1}
\renewcommand{\thepage}{\arabic{page}}
\renewcommand{\theequation}{\arabic{section}.\arabic{equation}}


\section{Introduction}\label{sec:intro}
Difference-in-difference (DiD) methods ({\it e.g.}, \citet*{ashenfelter1985using, angrist1999empirical}) are commonly used in empirical economics to establish causal relationships (see \citet*{currie2020technology} for some evidence regarding the usage in the empirical literature). In particular, researchers estimate regression functions of the form
\begin{equation}\label{eq:non_weight}
    Y_{it} = \mu + \alpha_i+ \lambda_t +\beta^\top X_{it} +\tau W_{it} + \epsilon_{it}
\end{equation}
using ordinary least squares (OLS), treating $\alpha_i$ and $\lambda_t$ as fixed parameters -- the fixed effects, leading to the two-way fixed effect (TWFE) estimator. Here $Y_{it}$ is the outcome variable of interest, $W_{it}$ is a binary treatment, $X_{it}$ are observed exogenous characteristics, and $\tau$ is the main object of interest. Practitioners routinely justify regression (\ref{eq:non_weight}) by appealing to ``quasi-experimental'' variation in treatment paths $\W_i = (W_{i1},\dots, W_{iT})$. Formal and informal arguments are invoked to make a case that this variation is not associated with unobserved unit and time-specific components $\epsilon_{it}$. In other words, to motivate  (\ref{eq:non_weight}), researchers reason about the underlying model for $\W_i$. This model, however, does not explicitly enter the estimation process. Moreover, econometric assumptions that justify the OLS estimation apply conditionally on $\W_i$ and do not appeal to randomness in the treatment paths (e.g., \cite{arellano2003panel}). 

In this paper, we develop new methods for estimating causal effects that explicitly incorporates design assumptions on the assignment process without abandoning the transparency and simplicity of the two-way model. We incorporate assumptions about the assignment mechanism by augmenting the specification (\ref{eq:non_weight}) with unit-specific weights $\gamma_i$, leading to
\begin{equation}\label{eq:weight_reg}
    \hat \tau(\gamma) = \argmin_{\tau,\mu, \alpha_i,\lambda_t,\beta}\sum_{it}(Y_{it} - \mu - \alpha_i- \lambda_t - \beta^\top X_{it} -\tau W_{it})^2 \gamma_i.
\end{equation}
We compute the weights $\gamma_i$ using the assignment model for $\W_i$.

We start our analysis by assuming that the assignment process for $\W_i$ is known. In Section \ref{sec:design_based}, we show how to use this knowledge to construct oracle weights $\gamma^{\star}$ and conduct design-based inference. Under the correct specification of the assignment model, our inference procedure is valid regardless of the underlying model for potential outcomes, and in particular, we do not rely on the validity of the equation (\ref{eq:non_weight}). Our results substantially generalize the properties established in \cite{athey2018design}, allowing for an arbitrary assignment process (subject to overlap restrictions).

To construct $\gamma^{\star}$, we need to solve a nonlinear equation that depends on the support of $\W_i$. Practically, this means that the construction and the values of the weights vary across different types of assignment processes. In Appendix \ref{app:DATE_equation} we provide solutions for several prominent examples, including staggered adoption, {\it i.e.}, a situation where units opt into treatment sequentially. Another input we need for $\gamma^{\star}$ is the probability distribution of $\W_i$ (generalized propensity score, \citealp{imbens2000}).

After establishing design-based properties of the oracle estimator $\hat \tau(\gamma^{\star})$ based on knowledge of the assignment process, we turn to the robustness -- the behavior of the estimator in settings where the postulated assignment model can be incorrect. At this point, we use the structure of the regression problem (\ref{eq:weight_reg}) to demonstrate that $\hat \tau(\gamma^{\star})$ has a strong double-robustness property (\citet*{robins1994estimation,kang2007demystifying, bang2005doubly, chernozhukov2018double}): it has a small bias whenever either the assignment or the regression model is approximately correct. We view these results as the primary motivation for using our estimator in practice, where we cannot expect either the TWFE model or the assignment model to be fully correct.

In practice, the assignment model is rarely completely known -- unless  $\W_i$-s are assigned in the controlled experiment (\textit{e.g.}, \citet*{attanasio2012education,broda2014economic, colonnelli2022corruption})-- and has to be estimated. 
We use the insights from the known assignment setting as a building block in Section \ref{sec:dri}, where the assignment process is unknown but can be estimated consistently from the data.
In Section \ref{sec:experiments} we use an empirical example to show how to estimate this distribution for the staggered adoption design using duration models. This approach is connected to \cite{shaikh2019randomization} that uses a duration model to test a sharp null hypothesis that specifies no treatment effects. Our general strategy of explicitly using the assignment model for estimation is directly connected to the recent literature on quasi-experimental designs (\textit{e.g.}, \cite{borusyak2022nonrandom}). Our results on robustness are especially appealing in such contexts because in quasi-experimental settings researchers cannot rule out the misspecification of the assignment model.

Our focus on TWFE regression (\ref{eq:weight_reg}) is motivated by its increased popularity in economics \citep*{currie2020technology}. In applications, this model provides an effective and parsimonious approximation for the baseline outcomes, allowing researchers to capture unobserved confounders and to improve the efficiency of the resulting estimator by reducing noise. At the same time, recent research shows that regression estimators for average treatment effects based on TWFE models might have undesirable properties, particularly  negative weights for unit-time specific treatment effects.
These concerns are particularly salient in settings with heterogeneity in treatment effects and general assignment patterns (e.g., \citet*{de2020two,goodman2018difference,sun2021estimating,callaway2018difference,borusyak2017revisiting}). Our results show that the concerns raised in this literature regarding negative weights lose some of their force under random assignment, or more generally once we properly reweight the observations.  

Our main analysis assumes that the treatment affects only contemporaneous outcomes, ruling out dynamic effects.  We make this choice to crystallize the connection between the TWFE regression model (\ref{eq:weight_reg}) and the assignment process. We do not restrict heterogeneity in contemporaneous treatment effects that can vary over units and periods. To test for, or estimate, dynamic treatment effects, one has to compare units that receive treatment at different times. Such comparisons are justified only if we restrict individual heterogeneity in treatment effects or treat the assignment as random. Consequently, and this is of course a key insight from the causal inference literature in cross-section settings since \cite{rosenbaum1983central}, it is imperative to model both the assignment mechanism and the outcome model. In \citet*{bojinov2020panel} the authors show how to use the assignment process to estimate dynamic treatment effects (see also \cite{blackwell2021adjusting} for the related analysis in large-$T$ setup).  Our results suggest that a fruitful approach may be to construct robust estimators by combining \cite{bojinov2020panel} approach to estimation with conventional dynamic panel regression models using the weighting methods derived in the current paper for the static case. We discuss a particular realization of this in Section \ref{sec:dynamic_effects}.

Our results are related to recent literature on doubly robust estimators with panel data. Conceptually the closest paper to us is \cite{arkhangelsky2019double} that also emphasizes the role of the assignment process in the same setting and shows double robustness.  
In \cite{arkhangelsky2019double} the focus is on a class of estimators defined as a linear function of realized outcomes, with the coefficients in that linear representation chosen to lead to consistent estimators for average treatment effects under either assumptions on the outcome model or on the assignment mechanism. Here, we start with a different class of estimators, restricted to weighted versions of the TWFE estimator in  (\ref{eq:weight_reg}). We also show how to estimate a flexible class of average treatment effects with user-specified weights over units and time. The robustness property in our paper is distinct from the double robustness analyzed recently in the difference-in-difference literature (\textit{e.g.}, \cite{sant2020doubly}): our estimator is robust to arbitrary violations of parallel trends assumptions, as long as the assignment model is correctly specified. At the same time, our estimator is not necessarily semiparametrically efficient in environments where, as in \cite{sant2020doubly}, the conditional parallel trends assumption holds.

We also connect to recent work on causal panel model with experimental data (\textit{e.g.}, \cite{athey2018design,bojinov2020panel,roth2023efficient}). Similar to these papers, we establish properties of regression estimators under design assumptions. Importantly, we consider a general setting without restricting our attention to staggered adoption design. Our contribution to this literature is the characterization of the behavior of $\hat \tau(\gamma)$ for a large class of weighting functions and general designs. By establishing a connection between weighting functions and limiting estimands, we allow users to construct consistent estimators for a pre-specified weighted average treatment effect of interest. 

Finally, the form of our estimator (\ref{eq:weight_reg}) connects it to the Synthetic Difference in Differences (SDID) estimator introduced in \citet*{arkhangelsky2019synthetic}. The difference between these two procedures is in the way they construct the weights $\gamma^{\star}$. The SDID estimator uses pretreatment outcomes to build a synthetic control unit that follows the path of the average treated unit as closely as possible (up to an additive shift). This strategy is infeasible if $W_{it}$ varies over time. However, precisely in situations with enough variation in $\W_i$, we can estimate the assignment process and use it to construct the weights $\gamma^{\star}$. As a result, the two estimators are complementary and can be used in applications with different assignment patterns.

Throughout the paper, we adopt the standard probability notation $O(\cdot), o(\cdot), O_{\P}(\cdot),$ $o_{\P}(\cdot)$. For any vector $v$, denote by $v^\tran$ the transpose of $v$, $\|v\|_{2}$ the $L_2$ norm of $v$, and by $\diag(v)$ the diagonal matrix with the coordinates of $v$ being the diagonal elements. For a pair of vectors $v_1, v_2$, we write $\la v_1, v_2\ra$ for their inner product $v_1^\tran v_2$. Furthermore, let $[m]$ denote the set $\{1, \ldots, m\}$, $I_{m}$ the $m\times m$ identity matrix, and $\one_{m}$ the $m$-dimensional vector with all entries $1$. Finally, the support of a discrete distribution $F$ is the set of elements with positive probabilities under $F$. 

\section{Reshaped IPW Estimator With Known Assignment Mechanisms}\label{sec:design_based}

We consider a setting with $n$ units and each unit is characterized by potential outcomes $\Y_i(1) = (Y_{i1}(1), \ldots, Y_{iT}(1)), \Y_i(0) = (Y_{i1}(0), \ldots, Y_{iT}(0))$ and a set of covariates $\X_i = (X_{i1}, \ldots, X_{iT})$.\footnote{Time-invariant covariates can be handled by letting $X_{i1} = \ldots = X_{iT}$.}  By writing the potential outcomes in this form, we assume away any dynamic effects of past treatments on current outcomes, thus focusing on static models. Analysis of such models is useful both theoretically and practically. First, they constitute a building block for more general environments, which we consider in Section \ref{sec:dynamic_effects}. Second, when the treatment is irreversible, as in staggered adoption designs, we are likely interested in its average (over time) effect on the outcome rather than the transitory dynamics. This makes the static model a reasonable approximation for a more complicated dynamic model. Finally, if we observe the data at a lower frequency than the one that is relevant for dynamics (e.g., days vs. months), then the static model is the only available option. 

Given the realized treatment assignment $W_{it}$, the observed outcomes are defined in the usual way:
\begin{equation}
Y_{it} = Y_{it}(1)W_{it} + Y_{it}(0)(1 - W_{it}).
\end{equation}
Throughout the paper, we treat covariates as fixed and consider $\{(\Y_i(1), \Y_i(0), \W_i): i\in [n]\}$ as a random vector (jointly) drawn from a distribution (conditional on $\{\X_i: i\in [n]\}$). We let $\P$ denote the joint distribution of the entire random vector $\{(\Y_i(1), \Y_i(0), \W_i): i\in [n]\}$ (conditional on $\{\X_i: i\in [n]\}$) and $\E$ denote the expectation over this distribution. We consider the asymptotic regime with $n$ going to infinity and fixed $T\ge 2$. 

This structure nests the conventional sampling-based framework, which is common in panel data analysis, going back to  \citep{chamberlain1984panel}, and which was used to establish statistical results in the recent DiD literature \citep[e.g.][]{abadie2005semiparametric, callaway2018difference}. It also extends the standard fixed effects framework, where the distribution for each unit is characterized by unit-specific parameters, but units themselves are usually assumed independent \citep[e.g.][]{neyman1948consistent,lancaster2000incidental}. Even in the absence of any covariates we do not assume that unit-level observations $(\Y_i(1), \Y_i(0), \W_i)$ are independent or exchangeable, which  brings two practical advantages. First, it allows us to accommodate correlated potential outcomes among units, which is natural in applications involving networks or multilevel structures. Second, it allows the assignments to be correlated across units, which is natural for many commonly used experimental designs. We elaborate on this point in the next section. 

In this section we study a special case where the assignment mechanism is known. This assumption is natural for experimental settings \citep{brown2006stepped, attanasio2012education, broda2014economic, hemming2015stepped,  chandar2019drivers, chandar2019design, colonnelli2022corruption}, but it has also been used to analyze the quasi-experimental settings \citep[e.g.][]{borusyak2023non}. It allows us to derive inferential results under mild assumptions. We will consider the case of unknown designs in Section \ref{sec:dri} at the cost of stronger (yet standard) assumptions.

\subsection{A design-based causal framework}
We assume that, for any $i \in \{1, \ldots, n\}$ and $\w\in \{0, 1\}^T$, 
\begin{equation}\label{eq:design_framework}
\P\lb \W_i = \w \mid \Y_i(1), \Y_i(0)\rb = \bfpi_i(\w),
\end{equation}
where $\bfpi_i$ is a distribution known to the analyst. We call it the generalized propensity score (\cite{imbens2000, athey2018design, bojinov2020panel, bojinov2020design}) -- the marginal probability of the treatment path. 

This structure allows for covariate-adaptive designs, where the probability of $W_{it}$ depends on past covariates. However, we rule out sequentially-adaptive designs where the assignment can depend on past outcomes, even if the randomization protocol is known.\footnote{Even if units are i.i.d. and $\P(W_{it}\mid Y_{i1}, \ldots, Y_{i(t-1)})$ is known, $\P(W_{it}\mid Y_{i1}, \ldots, Y_{iT})$ would depend on the unknown conditional distribution of $(Y_{it}, \ldots, Y_{iT})$ given $W_{it}$} Furthermore, our framework places no restriction on the support of $\W_i$ and substantially generalizes the previous works that focus on simple random sampling for non-staggered difference-in-differences \citep{rambachan2020design} and staggered adoption \citep{athey2018design,roth2023efficient}. 

If the treatment paths $\{\W_i, i \in [n]\}$ are independent across units, then the marginal distributions $\{\bfpi_i(\w), i \in [n]\}$ characterize the joint distribution of $\{\W_i, i \in [n]\}$. However, as discussed above, we allow the assignments to be correlated across units.  In practice, this correlation can range from being very mild, as in the case of completely randomized experiments with a fixed share of treated units \citep{neyman23}, to being sizable, as in cases of cluster-level randomization such as cluster randomized design \citep{abadie2023should} and two-stage randomization. We impose technical restrictions on the dependence across units in Section \ref{subsec:assumptions}. 

\subsection{Causal estimands}\label{subsec:estimand}

We define the unit and time-specific treatment effect as:
\begin{equation}
  \label{eq:tauit_finite_population}
  \tau_{it}\triangleq \E[Y_{it}(1) - Y_{it}(0)].
\end{equation}
Note that $\tau_{it}$ can vary with both $i$ and $t$ since we assume neither identically distributed units nor time-homogeneous treatment effects. For time period $t$, we define the time-specific ATE as:
\begin{equation}
  \label{eq:taut}
  \tau_{t}\triangleq \frac{1}{n}\sum_{i=1}^{n}\tau_{it},
\end{equation}
and consider a broad class of weighted average of time-specific ATE:
\begin{equation}
  \label{eq:taustar}
  \tau^{*}(\xi)\triangleq \sum_{t=1}^{T}\xi_{t}\tau_{t}
\end{equation}
for some user-specified deterministic weights $\xi = (\xi_{1}, \ldots, \xi_{T})^{\tran}$ such that
\begin{equation}
  \label{eq:xi_sum}
  \sum_{t=1}^{T}\xi_{t} = 1, \quad \xi_{t} \ge 0.
\end{equation}
We refer to \eqref{eq:taustar} as a doubly average treatment effect (DATE). For example, the weights $\xi_{t}=1 / T$ yield the usual ATE over units and time periods.
In the difference-in-differences setting with two time periods, $\xi_{t} = \mathbf{1}_{t = 2}$. In a particular application, one might also be interested in an effect with time discounting factor that puts more weight on initial periods, i.e. $\xi_{t} \propto \beta^{t}$ for some $\beta < 1$. 

\begin{remark}
  We can further generalize DATE by allowing for unequal unit weights:
  \begin{equation}\label{eq:generalized_DATE}
    \tau^{*}(\xi; \zeta) = \sum_{i=1}^{n}\sum_{t=1}^{T}(\zeta_{i}\xi_{t})\tau_{it},
  \end{equation}
  where $\zeta = (\zeta_1, \ldots, \zeta_{n})$, $\sum_{i=1}^{n}\zeta_{i} = 1$ and $\zeta_{i}\ge 0$. In particular, $\zeta_i$ can be $i$-specific, e.g., a function of the $i$-th covariates, but cannot depend on outcomes and treatment assignments. Using appropriate propensity-based weights $\zeta$ one can build estimands that target a given subpopulation. 
\end{remark}

\subsection{Technical assumptions}\label{subsec:assumptions}
We allow $\W_i$ to be dependent across units to capture different assignment processes. Such dependence arises in applications, sometimes for technical reasons (e.g., in case of sampling without replacement as in \cite{athey2018design}), and sometimes by the nature of the assignment process (spatial experiments). To quantify this dependence as well as the dependence among the potential outcomes, we follow \cite{renyi1959measures} and define the maximal correlation:
\begin{equation}\label{eq:rhoij}
    \rho_{ij}\triangleq \sup_{f,g}\left\{\text{corr}\left(f(\Y_i(1), \Y_i(0), \W_i),g(\Y_j(1), \Y_j(0), \W_j)\right)\right\},
\end{equation}
where the supremum is taken over all real-valued measurable functions $f, g$.

In the standard design-based framework where potential outcomes are assumed fixed, it reduces to the $\rho$-mixing coefficient between $\W_i$ and $\W_j$. In the main text, we maintain a simplified restriction on $\{\rho_{ij}\}_{ij}$ leaving a more general one to Appendix \ref{app:proofs}. The assumption is stated as follows:
\begin{assumption}\label{as:limit_dep}
There exists $q \in (0, 1]$ such that  as $n$ approaches infinity the following holds:
\begin{equation}\label{eq:weak_dependence}
\frac{1}{n^2}\sum_{i,j=1}^{n}\rho_{ij} = O(n^{-q}).
\end{equation}
\end{assumption}
By definition $\frac{1}{n}\le (1/n^2)\sum_{i,j=1}^{n}\rho_{ij}\le 1$ with lower bound being attained if the observations are independent, and the upper bound being attained if they are perfectly dependent. As a result, one can view $q$ as measuring the strength of the correlation. When $(\Y_i(1), \Y_i(0), \W_i)$ are independent across units, \eqref{eq:weak_dependence} holds with $q = 1$. More generally, when $\{(\Y_i(1), \Y_i(0), \W_i): i\in [n]\}$ have a network dependency with $\rho_{ij} = 0$ if there is no edge between $i$ and $j$, \eqref{eq:weak_dependence} is satisfied if the number of edges is $O(n^{2(1 - q)})$. Note that it imposes no constraint on the maximum degree of the dependency graph. Even if the network is fully connected, it can still hold if the pairwise dependence is weak, e.g., sampling without replacement; see Appendix \ref{subapp:design_based_dependent}. On the other hand, \eqref{eq:weak_dependence} excludes the case where all units are perfectly correlated or equicorrelated with a positive maximal correlation that is bounded away from $0$.

We also impose minimal overlap  restrictions on each $\bfpi_i$:
\begin{assumption}\label{as:overlap}
There exists a universal constant $c > 0$ and a non-stochastic subset $\S^{*}\subset \{0, 1\}^{T}$ with at least two elements and at least one element not in $\{\zero_{T}, \one_{T}\}$, such that
  \begin{equation}
    \label{eq:overlap}
    \bfpi_{i}(\w) > c, \quad \forall \w \in \S^{*}, i\in [n].
  \end{equation} 
\end{assumption}

Our final assumption restricts the second moment of outcomes:
\begin{assumption}\label{as:bounded_outcomes}
There exists $M<\infty$ such that $\max_{i,t,w}\E[Y_{it}^2(w)] <M$.
\end{assumption}
It is presented here only for simplicity. We relax it substantially in Appendix \ref{app:proofs}.

\subsection{Reshaped IPW estimator}

We consider a class of weighted TWFE regression estimators without covariates. We refer to them as reshaped inverse propensity weighted (RIPW) estimators, and formally define them as follows:
\begin{equation}
  \label{eq:weighted_fe}
  \hat{\tau}(\bfPi) \triangleq \argmin_{\tau, \mu, \sum_{i}\alpha_{i} = \sum_{t}\lambda_{t} = 0} \sum_{i=1}^{n}\sum_{t=1}^{T}(Y_{it} - \mu - \alpha_{i} - \lambda_{t} - W_{it}\tau)^{2}\frac{\bfPi(\W_{i})}{\bfpi_{i}(\W_{i})},
\end{equation}
where $\bfPi(\w)$ is a density function on $\{0, 1\}^{T}$, i.e.,
\begin{equation}
  \label{eq:bfPi_density}
  \sum_{\w\in \{0, 1\}^{T}}\bfPi(\w) = 1.
\end{equation}
We refer to the distribution $\bfPi$ as a reshaped distribution, and the weight $\bfPi(\W_{i}) / \bfpi_{i}(\W_{i})$ as a RIP weight. To ensure that the RIPW estimator is well-defined, we require $\bfPi$ to be absolutely continuous with respect to each $\bfpi_{i}$, i.e.
\begin{equation}
  \label{eq:common_support}
  \bfPi(\w) = 0\,\,\text{ if }\,\, \bfpi_{i}(\w) = 0\text{ for some }i\in [n].
\end{equation}
The estimator (\ref{eq:weighted_fe}) is feasible for any such $\bfPi$ because $\bfpi_i$ is assumed to be known.

Adding covariates to the objective function (\ref{eq:weighted_fe}) is relatively straightforward. However, it considerably complicates the notation without contributing substantially to the primary narrative. We will explicitly incorporate covariates in the objective function in Section \ref{sec:dri}. Note that the covariates still play a role in the RIPW estimator through $\bfpi_i$ for covariate-adaptive designs.

The reshaped distribution $\bfPi$ can be interpreted as an experimental design. If $\W_{i}\sim \bfPi$, then $\bfpi_{i} = \bfPi$ and \eqref{eq:weighted_fe} reduces to the standard unweighted TWFE regression. If this is not the case, then  $\bfPi(\W_i) / \bfpi_i(\W_i)$ acts like a likelihood ratio that changes the original design to one provided by $\bfPi$. For cross-sectional data, we would like to shift the distribution to uniform $\{0, 1\}$, making the weights equal to $1 / 2\bfpi_i(\W_i)$ if the fixed effects are not included. This would yield the standard IPW estimator. However, as we alluded to in the introduction, the situation is more complicated with panel data, and shifting towards the uniform design might not deliver consistent estimators for the DATE of interest. We explore this formally in the next section, where we characterize the set of $\bfPi$ that one can use. This interpretation of $\bfPi$ has one caveat: RIP weights only shift the marginal distribution of $\W_i$ to $\bfPi$, but they do not say anything about the joint distribution of  $\{\W_i,i\in[n]\}$ which can remain complicated.

\subsection{\DATEeq ~and consistency of RIPW estimators}

We now derive sufficient conditions under which the RIPW estimator is a consistent estimator for a given DATE of interest. The following theorem presents a precise condition for consistency of $\hat{\tau}(\bfPi)$ for $\tau^{*}(\xi)$: 
\begin{theorem}\label{thm:bias}
Let $\J = I_{T} - \one_{T}\one_{T}^{\tran} / T$ and $\bftau_{i} = (\tau_{i1}, \ldots, \tau_{iT})^{\tran}$; fix $\xi$ that satisfies (\ref{eq:xi_sum}). Under Assumptions \ref{as:limit_dep}-\ref{as:bounded_outcomes}, for any reshaped distribution $\bfPi$ with support $\S^{*}$ defined in Assumption \ref{as:overlap}, as $n$ tends to infinity,
  \[\hat{\tau}(\bfPi) - \tau^{*}(\xi) = O_{\P}\lb \bias(\xi)\rb + o_\P(1),\]
  where 
  \[\bias(\xi) = \la\EPi\left[(\diag(\W) - \xi\W^{\tran})\J(\W - \EPi[\W])\right], \frac{1}{n}\sum_{i=1}^{n}\lb \bftau_{i} - \tau^{*}(\xi)\one_{T}\rb\ra,\]
  and $\la v_1, v_2\ra$ denotes their inner product $v_1^\tran v_2$. 
\end{theorem}
This result has two user-specified parameters: time weights $\xi$, and the reshaped distribution $\bfPi$. They are naturally connected: to guarantee consistency for $\tau^{*}(\xi)$ we can select $\bfPi$ such that the following holds:
\begin{equation}
  \label{eq:DATE_equation}
  \EPi\left[(\diag(\W) - \xi\W^{\tran})\J(\W - \EPi[\W])\right] = 0.
\end{equation}
Alternatively, for a given $\bfPi$, we can look for $\xi$ such that (\ref{eq:DATE_equation}) is satisfied. We call \eqref{eq:DATE_equation} the DATE equation hereafter. For a fixed $\xi$, it is a quadratic system with $\{\bfPi(\w): \w\in \{0, 1\}^{T}\}$ being the variables. Together with the density constraint \eqref{eq:bfPi_density} and the support constraint in Theorem \ref{thm:bias} that $\bfPi(\w) = 0$ for $\w\not\in \S^{*}$, there are $T + 1 + 2^{T} - |\S^{*}|$ equality constraints and $|\S^{*}|$ inequality constraints that impose the positivity of $\bfPi(\w)$ for each $\w\in \S^{*}$. We will show in Appendix \ref{app:DATE_equation} that the \DATEeq ~have closed-form solutions in various examples and provide a generic solver based on nonlinear programming in Appendix \ref{subapp:solver}.

Without further restrictions on $\bftau_i$, we can show that the \DATEeq ~is also a necessary condition for consistency of $\hat{\tau}(\bfPi)$ for $\tau^{*}(\xi)$. To see this assume that  \begin{equation}
  \label{eq:DATE_alt}
  \EPi\left[(\diag(\W) - \xi\W^{\tran})\J(\W - \EPi[\W])\right] = z.
\end{equation}
for some vector $z$ that is not proportional to $\xi$. Because we can vary individual treatment effects without changing the average one, we can find a set $\{\bftau_{i}: i\in [n]\}$ that yields the same DATE but $\la z, (1/n)\sum_{i=1}^{n}\lb \bftau_{i} - \tau^{*}(\xi)\one_{T}\rb\ra \not= 0$, leading to inconsistency. For $z=b \xi$ we get that the inner product of the LHS of (\ref{eq:DATE_alt}) and $\one_{T}$ is $0$ because $\one_{T}^\tran (\diag(\W) - \xi \W^\tran) = W^\tran (1 - \one_T^\tran \xi) = 0$, while that of the right-hand side and $\one_{T}$ is equal to $b$. This implies that $z$ has to be equal to zero, thus proving the necessity of \DATEeq. 

Notably, when the \DATEeq ~has a solution, our estimator is consistent without any restrictions on the potential outcomes, except Assumption \ref{as:bounded_outcomes}. This is in sharp contrast to usual results about TWFE estimators, which typically require the trends to be parallel among units, at least conditionally on observed covariates  \citep[e.g.][]{callaway2018difference, sant2020doubly}. Theorem \ref{thm:bias} shows that if the assignment process is known and the DATE equation has a solution, we can correct the potentially misspecified TWFE regression model by simply reweighting the objective function. We want to stress that this result relies on the knowledge of the assignment process, whereas the analysis based on conditional parallel trends does not require such knowledge.

To further parse the \DATEeq, we discuss two alternative interpretations. First, fix $\xi$ and let $\bfPi$ be the solution of the \DATEeq. Then consider a class of complete randomized experiments where all propensity scores $\bfpi_i$ are identical and are equal to $\bfPi$. Then, by definition, the RIPW estimator with reshaped distribution $\bfPi$ reduces to the standard (unweighted) TWFE estimator. Theorem \ref{thm:bias} guarantees that this estimator converges to $\tau^{*}(\xi)$. Since the \DATEeq is a necessary condition, all experimental designs that do not satisfy this restriction cannot lead to a consistent estimator for $\tau^{*}(\xi)$. As a result, \DATEeq~characterizes all complete randomized experiments under which the unweighted two-way estimator converges to a given estimand. This can be interpreted as a general converse of the results established in \cite{athey2018design}. 

As an alternative interpretation, consider a fixed $\bfPi$ instead. For any such $\bfPi$ the equation \eqref{eq:DATE_equation} can be rewritten as
\begin{equation}
  \label{eq:Pi_xi}
  \lb\EPi[\W^{\tran}\J (\W - \EPi[\W])]\rb\xi = \E[\diag(\W)\J(\W - \EPi[\W])].
\end{equation}
It is easy to see that
\begin{multline}\label{eq:denominator}
  \EPi[\W^{\tran}\J (\W - \EPi[\W])] = \EPi[(\W - \EPi[\W])^{\tran}\J (\W - \EPi[\W])]\\
   = \EPi\left[\left\|\td{\W} - \EPi[\td{\W}]\right\|_{2}^2\right],
\end{multline}
where $\td{\W} = \J\W$. Since the support of $\bfPi$ involves a point $\w\not\in \{\zero_{T}, \one_{T}\}$, for which $\w' \not= 0$ the quantity in \eqref{eq:denominator} is strictly positive. Therefore, \eqref{eq:Pi_xi} implies that
\begin{equation}
  \label{eq:effective_xi}
  \xi = \frac{\EPi[\diag(\W)\J(\W - \EPi[\W])]}{\EPi\left[\left\|\td{\W} - \EPi[\td{\W}]\right\|_{2}^2\right]}.
\end{equation}
By Theorem \ref{thm:bias}, in a randomized experiment with $\bfpi_{i} \triangleq \bfPi$ \citep{athey2018design, roth2023efficient}, the effective estimand of the unweighted TWFE regression is the DATE with weight vector $\xi$. 

\begin{remark}
To illustrate this result, we consider the experiment conducted by Uber in 2017 to test the effect of in-app tipping on labor supply \citep{chandar2019drivers, chandar2019design}. They introduced the in-app tipping feature in a staggered fashion across 209 operational cities in the United States and Canada to avoid bugs in the product. Three cities were randomized to launch this feature on June 20, 2017, followed by 103 cities on July 6, 2017, and the remaining 103 cities on July 17, 2017. We can treat it as a two-period experiment, with June 20 - July 5 being the first period and July 6 - July 16 being the second period. The possible assignments include $\{1, 1\}, \{0, 1\}, \{0, 0\}$ and $\bfpi_i(\{1, 1\}) = 3/209, \bfpi_i(\{0, 1\}) = \bfpi_i(\{0, 0\}) = 103/209$. By \eqref{eq:DATE_equation_T=2} in Appendix \ref{app:DATE_equation}, \eqref{eq:effective_xi} implies that $\xi_1 = 3/106, \xi_2 = 103/106$ for the unweighted TWFE regression, which they applied to estimate the treatment effect. Thus, their analysis is essentially focused on the second period.
\end{remark}

The following result shows that the induced weights are guaranteed to be non-negative for arbitrary design.
\begin{proposition}\label{prop:effective_xi}
Let $\xi$ be defined in \eqref{eq:effective_xi}. Then for any $\bfPi$ on $\{0, 1\}^{T}$, $\xi_t \ge 0$ for all $t$.
\end{proposition}
This result generalizes the conventional cross-sectional logic that says that in randomized experiments, regression estimators are consistent for average effects \citep[e.g.][]{lin2013agnostic}. However, in the case of the TWFE regression, the situation is more nuanced. While the resulting estimand always corresponds to a weighted average effect with non-negative weights, it still depends on the experimental design. As a result, if two analysts were to split a given population into two random subpopulations and conduct two experiments with different designs on each part, the resulting estimands would have been different.

There are two reasons for this unusual behavior. First, in the cross-sectional case, $\W_i$ has two points of support, while in the panel case the support of $\W_i$ ranges from $2$ to $2^T$ points (as long as Assumption \ref{as:overlap} is satisfied). For example, if none of the units is treated in the first period, it is impossible to identify any DATE that puts positive weight on the first period. Second, fixed effects lead to a familiar incidental parameter problem \citep{neyman1948consistent}, albeit in a mild form. To see this, consider $ \bfpi_i= \bfPi \propto 1$, in which case the RIPW estimator corresponds to the conventional TWFE regression. The effective estimand for this regression is equal to the solution of (\ref{eq:effective_xi}) and is different from the effective estimand for the regression without the unit fixed effects.

\begin{remark}\label{rem:generalized_DATE}
  To estimate the generalized DATE defined in \eqref{eq:generalized_DATE}, we only need to mildly adjust the RIPW estimator:
  \begin{equation}
    \label{eq:generalized_RIPW}
  \hat{\tau}(\bfPi; \zeta) \triangleq \argmin_{\tau, \mu, \sum_{i}\alpha_{i} = \sum_{t}\lambda_{t} = 0} \sum_{i=1}^{n}\sum_{t=1}^{T}(Y_{it} - \mu - \alpha_{i} - \lambda_{t} - W_{it}\tau)^{2}\frac{\zeta_i\bfPi(\W_{i})}{\bfpi_{i}(\W_{i})}.
  \end{equation}
  In Appendix \ref{subapp:generalized_DATE} we prove that the adjusted RIPW estimator $\hat{\tau}(\bfPi; \zeta)$ consistently estimates $\tau^{*}(\xi; \zeta)$ under the same set of assumptions as in Theorem \ref{thm:bias}, provided that $n\|\zeta\|_{\infty} = O(1)$, namely that all entries of $\zeta$ are on the same scale. 
\end{remark}

\subsection{Inference on RIPW estimators}
To enable statistical inference of DATE, we first present an asymptotic expansion showing the asymptotic linearity of RIPW estimators.

\begin{theorem}\label{thm:design_linear_expansion}
  Let $\Y_{i}$ be the vector $(Y_{i1}, \ldots, Y_{iT})$. Further let $\Theta_{i} = \bfPi(\W_i) / \bfpi_{i}(\W_i)$, and
\begin{equation*}
  \stheta\triangleq \frac{1}{n}\sum_{i=1}^{n}\Theta_{i}, \quad \sww \triangleq \frac{1}{n}\sum_{i=1}^{n}\Theta_{i}\W_{i}^{\tran}\J\W_{i}, \quad \swy \triangleq \frac{1}{n}\sum_{i=1}^{n}\Theta_{i}\W_{i}^{\tran}\J\Y_i,
\end{equation*}
and
\begin{equation*}
\sw \triangleq \frac{1}{n}\sum_{i=1}^{n}\Theta_{i}\J\W_{i}, \quad \sy \triangleq \frac{1}{n}\sum_{i=1}^{n}\Theta_{i}\J\Y_i.
\end{equation*}
Under the same settings as Theorem \ref{thm:bias},
  \[\hatdenom\cdot \sqrt{n}(\hat{\tau}(\bfPi) - \tau^{*}(\xi)) = \frac{1}{\sqrt{n}}\sum_{i=1}^{n}(\V_{i} - \E[\V_{i}]) + O_\P\lb n^{1/2-2q} \rb,\]
  where $\hatdenom = \sww\stheta - \sw^\tran \sw$, and 
  \begin{align*}
    \V_i &= \Theta_i\Bigg\{\lb\E[\swy] - \tau^{*}(\xi)\E[\sww]\rb - \lb \E[\sy] - \tau^{*}(\xi)\E[\sw]\rb ^\tran\J \W_i \\
    & \qquad\quad + \E[\stheta]\W_i^\tran \J \lb\Y_i - \tau^{*}(\xi)\W_i\rb - \E[\sw]^\tran\J\lb \Y_{i} - \tau^{*}(\xi)\W_i\rb \Bigg\}
  \end{align*}
\end{theorem}

Note that the asymptotic linear expansion holds under a fairly general dependency structure in the treatment assignments. Below, we derive a valid confidence intervals for $\tau^*(\xi)$ when $\{(\Y_i(1), \Y_i(0),\W_i): i\in [n]\}$ are independent. The general case is discussed in Appendix \ref{subapp:design_based_dependent}. If $\{\V_i: i\in [n]\}$ are well-behaved Theorem \ref{thm:design_linear_expansion} implies that
\[\frac{\hatdenom\cdot \sqrt{n}(\hat{\tau}(\bfPi) - \tau^{*}(\xi))}{\sigma_{n}^{*}} \approx N(0, 1), \quad \text{where }\sigma_{n}^{*2} = (1/n)\sum_{i=1}^{n}\Var(\V_i),\]
where $\hatdenom$ is known by design.\footnote{By well-behaved $\V_i$ we mean that they are sufficiently regular for the appropriate version of the Central Limit Theorem to hold. In the simplest case, when data is i.i.d., this reduces to standard moment restrictions.} If $\{\V_i: i\in [n]\}$ were known, a natural estimator for $\sigma_{n}^{*2}$ would be the empirical variance:
\[\hat{\sigma}_{n}^{*2} = \frac{1}{n - 1}\sum_{i=1}^{n}(\V_i - \bar{\V})^2, \quad \text{where }\bar{\V} = \frac{1}{n}\sum_{i=1}^{n}\V_i.\]
We should not expect the difference between $\hat{\sigma}_{n}^{*}$ and $\sigma_{n}^{*}$ to converge to zero since $\E[\V_i]$ in general varies over $i$. Nonetheless, $\hat{\sigma}_{n}^{*}$ is an asymptotically conservative estimate of $\sigma_{n}^{*}$ since
\begin{equation}
  \label{eq:conservative_variance}
  \E[\hat{\sigma}_{n}^{*2}] \approx \frac{1}{n}\sum_{i=1}^{n}\E\left[\lb \V_i - \frac{1}{n}\sum_{i=1}^{n}\E[\V_i]\rb^2\right]\approx \sigma_{n}^{*2} + \underbrace{\frac{1}{n-1}\sum_{i=1}^{n}\lb \E[\V_i] - \frac{1}{n}\sum_{i=1}^{n}\E[\V_i]\rb^2}_{\text{empirical variance of }\E[\V_i]},
\end{equation}
where the second term measures the heterogeneity of $\E[\V_i]$ and is always non-negative, implying that $\hat{\sigma}_{n}^{*2}$ is a conservative estimator for $\sigma_{n}^{*2}$. This is unsurprising because even in the cross-section case, the asymptotic design-based variance is only partially identifiable due to the unknown correlation structure between two potential outcomes; see, e.g., Neyman's variance formula \citep{neyman23, rubin74}.

In general, $\V_i$ is unknown due to $\tau^{*}(\xi)$ and the expectation terms. Nonetheless, we can estimate $\V_i$ by replacing each expectation with the corresponding plug-in estimate, i.e.
\begin{align}
    \hV_i &= \Theta_i\Bigg\{\lb\swy - \hat{\tau}\sww\rb - \lb \sy - \hat{\tau}\sw\rb ^\tran\J \W_i \nonumber\\
    & \qquad\quad + \stheta\W_i^\tran \J \lb\Y_i - \hat{\tau}\W_i\rb - \sw^\tran\J\lb \Y_{i} - \hat{\tau}\W_i\rb \Bigg\},\label{eq:hatVi}
\end{align}
and use them to compute the variance:
\begin{equation}
  \label{eq:var_estimate}
  \hat{\sigma}^2 = \frac{1}{n - 1}\sum_{i=1}^{n}(\hV_i - \bar{\hV})^2, \quad \text{where }\bar{\hV} = \frac{1}{n}\sum_{i=1}^{n}\hV_i.
\end{equation}
This yields a Wald-type confidence interval for $\tau^{*}(\xi)$ as
\begin{equation}
  \label{eq:confidence_interval}
  \hat{C}_{1 - \alpha} = [\hat{\tau}(\bfPi) - z_{1 - \alpha / 2}\hat{\sigma} / (\sqrt{n}\hatdenom), \hat{\tau}(\bfPi) + z_{1 - \alpha / 2}\hat{\sigma} / (\sqrt{n}\hatdenom)],
\end{equation}
where $z_{\eta}$ is the $\eta$-th quantile of the standard normal distribution. Properties of this confidence interval are established in the next theorem.
\begin{theorem}\label{thm:design_inference}
  Assume that $\{(\Y_i(1), \Y_i(0), \W_i): i \in [n]\}$ are independent with
  \begin{equation}\label{eq:variance_lower}
  \frac{1}{n}\sum_{i=1}^{n}\Var(\V_i)\ge v_0, \quad \text{for some constant } v_0 > 0.
  \end{equation}
    Then under Assumptions \ref{as:overlap} and \ref{as:bounded_outcomes}, for any $\alpha\in (0, 1)$, 
  \[\liminf_{n\rightarrow \infty}\P\lb \tau^{*}(\xi)\in \hat{C}_{1 - \alpha}\rb\ge 1 - \alpha.\]
\end{theorem}
In Appendix \ref{subapp:design_based_dependent}, we discuss a generic result for general dependent assignments (Theorem \ref{thm:design_based_generic_inference}), which covers completely randomized experiments, blocked and matched pair experiments, two-stage randomized experiments, and so on. We present a detailed result (Theorem \ref{thm:design_based_without_replacement}) for completely randomized experiments where potential outcomes are fixed and $\W_i$'s are sampled without replacement from a user-specified subset of $\{0, 1\}^T$.\footnote{Specifically, given any support $\S^{*}$ and a pre-specified vector $\{n_{\w}: \w\in \S^*\}$ with $\sum_{\w\in \S^*}n_{\w} = n$, the experimenter sample assignments $(\W_1, \ldots, \W_n)$ with probability $\prod_{\w\in \S^{*}}n_{\w} ! / n!$.} This substantially generalizes the setting of \cite{athey2018design} and \cite{roth2023efficient}, where the assignments are sampled without replacement from the set of $T+1$ staggered assignments. At the same time, compared to \cite{roth2023efficient}, we cannot provide efficiency guarantees for our estimator.

\subsection{Discussion}\label{subsec:disc_nonng_w}

Theorem \ref{thm:bias} and Proposition \ref{prop:effective_xi} might appear counter-intuitive given well-understood problems of TWFE estimators (e.g., \cite{de2020two,goodman2018difference,sun2021estimating}). To put our result in context, we emphasize two important features of the setup. First, we restrict attention to static models, and second, we use the randomness that is coming from $\W_i$. Both of these restrictions play a key role in Theorem \ref{thm:bias}.  The absence of dynamic effects implies that we can meaningfully average units with different histories of past treatments. A version of this assumption is inescapable if we want the method to work for general designs where controlling for past history is practically infeasible. As we explain below, the randomness of assignments helps to resolve the issue that TWFE estimators put negative weights on some individual treatment effects.

In \cite{de2020two,goodman2018difference,sun2021estimating} the authors show that treated units are averaged with potentially negative weights, but these results
are conditional on the assignments $\W = (\W_1, \ldots, \W_n)$ being fixed. Let $\xi_{it}(\gamma; \W)$ be these weights for the general weighted least squares estimator $\hat{\tau}(\gamma)$ defined in \eqref{eq:weight_reg} such that
\[\E[\hat{\tau}(\gamma)\mid \W] = \sum_{i=1}^{n}\sum_{t=1}^{T}\xi_{it}(\gamma; \W)\tau_{it},\]
where we now explicitly allow the weights to depend on $\W$. When the assignments are treated as random, the large sample limit of $\hat{\tau}(\gamma)$ is 
\[\E[\hat{\tau}(\gamma)] = \sum_{i=1}^{n}\sum_{t=1}^{T}\xi_{it}(\gamma)\tau_{it},\]
where $\xi_{it}(\gamma) = \E_{\W}[\xi_{it}(\gamma; \W)]$. While $\{(i, t): \xi_{it}(\gamma; \W) < 0\}$ is non-empty almost surely for every realization of $\W$, it is still possible that all $\xi_{it}(\gamma)$ are positive due to the averaging over $\W$. For illustration, we consider a simulation study with $n = 100, T = 4$ and other details specified in Section \ref{subsec:synthetic}. We consider the conditional and unconditional weights induced by the unweighted and RIP-weighted TWFE estimator in Figure \ref{fig:unweighted_weights} and Figure \ref{fig:ripw_weights}, respectively. We plot the histograms of $\{(nT)\cdot\xi_{it}(\gamma; \W): i\in [n], t \in [T]\}$ for three realizations of $\W$ and the histogram of $\{(nT)\cdot \xi_{it}(\gamma): i\in [n], t\in [T]\}$, approximately by averaging over a million realizations of $\W$, where the multiplicative factor $nT$ is chosen to normalize the weights into a more interpretable scale. Clearly, despite the large fraction of negative weights in each realization, their averages do not have any negatives. Therefore, the criticism on TWFE estimators does not apply in this case. Indeed, it never applies to the RIPW estimator by Proposition \ref{prop:effective_xi}. In this study, all weights are designed to be $1/nT > 0$ when $\bfPi$ is a solution of the \DATEeq ~with $\xi = \one_{T}/T$, as shown in Figure \ref{fig:ripw_weights}(b), regardless of the data generating process.

\begin{figure}
    \centering
    \begin{subfigure}{0.49\textwidth}
        \centering
        \includegraphics[width = 0.8\textwidth]{{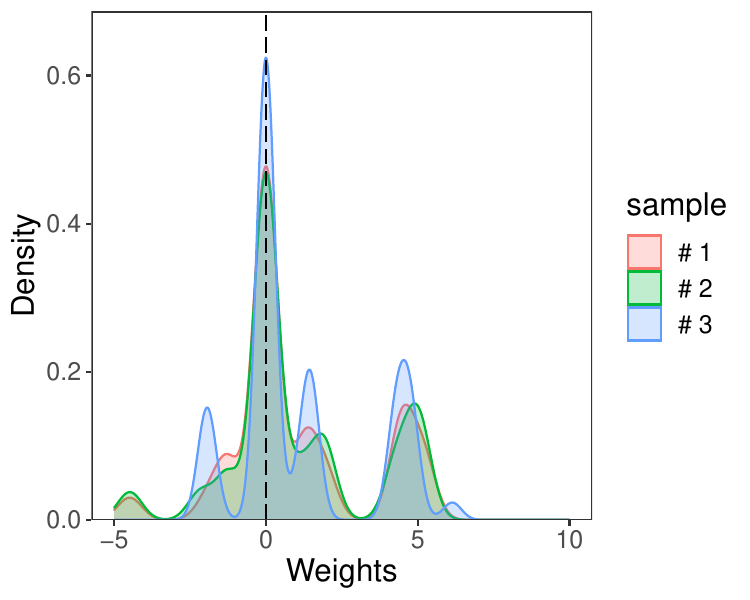}}
        \caption{Histograms of $(nT)\cdot\xi_{it}(\gamma; \W)$'s}
    \end{subfigure}
    \begin{subfigure}{0.49\textwidth}
        \centering
        \includegraphics[width = 0.8\textwidth]{{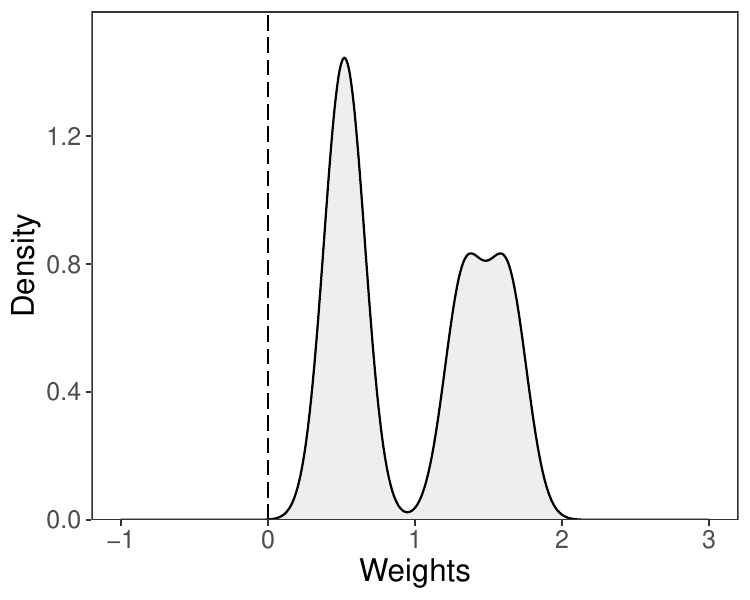}}
        \caption{Histogram of $(nT)\cdot\xi_{it}(\gamma)$'s}
    \end{subfigure}
    \caption{Effect weights for the unweighted TWFE estimator.}
    \label{fig:unweighted_weights}
\end{figure}

The discussion above demonstrates that while for each cell $(i,t)$, a particular realization of weights can be negative, this fact is not systematic. If we use the RIPW estimator designed for the equally weighted DATE, then all cells will receive the same weight on average. An alternative description of the same phenomenon is that once correctly weighted, the realized treatment paths $\W_i$ are independent of potential outcomes. This independence implies that there cannot be systematic differences in treatment effects among units with distinct assignment paths, and thus negative weights do not create complications for the interpretation of the estimates. As we illustrate in Section \ref{sec:dynamic_effects}, this interpretation remains valid even when certain dynamic effects are present. 

\begin{remark}
One might ask if Figure \ref{fig:unweighted_weights} (b) presents a general feature of unweighted TWFE estimators with random assignments. For completely randomized experiments where $\bfpi_i \equiv \bfPi$, the standard TWFE estimator is equivalent to the RIPW estimator with reshaped distribution $\bfPi$. By Proposition \ref{prop:effective_xi}, all weights are guaranteed to be non-negative. When $\bfpi_i$ varies across units, the weights are not guaranteed to be non-negative. Consider the extreme case where $\bfpi_i$ assigns $1-\eps$ mass on one assignment pass and $\eps$ mass on all others. As $\eps \rightarrow 0$, this approaches the case of fixed treatment assignments, for which the unconditional weights are almost the same as the conditional weights, which always include negative ones. 
\end{remark}

\begin{figure}
    \centering
    \begin{subfigure}{0.49\textwidth}
    \centering
        \includegraphics[width = 0.8\textwidth]{{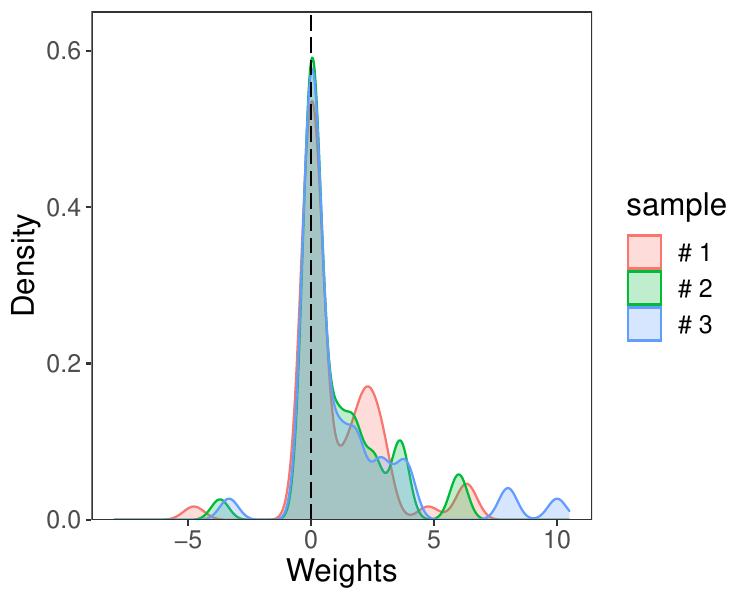}}
        \caption{Histograms of $(nT)\cdot\xi_{it}(\gamma; \W)$'s}
    \end{subfigure}
    \begin{subfigure}{0.49\textwidth}
    \centering
        \includegraphics[width = 0.8\textwidth]{{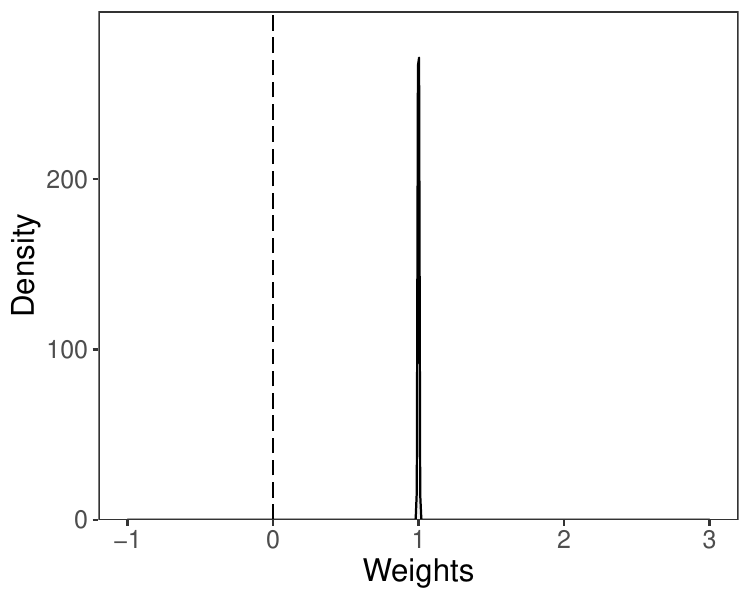}}
        \caption{Histogram of $(nT)\cdot\xi_{it}(\gamma)$'s}
    \end{subfigure}
    \caption{Effect weights for our RIPW estimator.}
    \label{fig:ripw_weights}
\end{figure}

\section{Reshaped IPW Estimator With Unknown Assignment Mechanisms}\label{sec:dri}

In this section, we move to non-experimental settings where the assignment mechanism is not controlled by the researcher and is unknown. We assume that researchers constructed unit-level estimates $\{\hat{\bfpi}_i, i \in [n]\}$. In addition, we assume that the researchers have access to a set of estimates $\{(\hat{\bfmu}_{i}(0), \hat{\bfmu}_{i}(1)): i\in [n]\}$ of $\{(\E[\Y_{i}(0)], \E[\Y_{i}(1)]): i\in [n]\}$. Further, let $\hat{m}_{it}$ be the double-centered version of $\hat{\mu}_{it}(0)$ and $\htau_{it}$ be a shifted version of $\hat{\mu}_{it}(1) - \hat{\mu}_{it}(0)$:
 \begin{align}
&\hat{m}_{it} \triangleq  \hat{\mu}_{it}(0) - \frac{1}{n}\sum_{i=1}^{n}\hat{\mu}_{it}(0) - \frac{1}{T}\sum_{t=1}^{T}\hat{\mu}_{it}(0) + \frac{1}{nT}\sum_{i=1}^{n}\sum_{t=1}^{T}\hat{\mu}_{it}(0), \label{eq:hatm} \\
&\htau_{it}  \triangleq (\hat{\mu}_{it}(1) - \hat{\mu}_{it}(0)) - \sum_{t=1}^{T}\frac{\xi_{t}}{n}\sum_{i=1}^{n}(\hat{\mu}_{it}(1) - \hat{\mu}_{it}(0)).\label{eq:hattau}
\end{align}
For notational convenience, we write $\hat{\bfm}_i$ for the vector $(\hat{m}_{i1}, \ldots, \hat{m}_{iT})$ and $\hbftau_i$ for the vector $(\htau_{i1}, \ldots, \htau_{iT})$. Given a set of estimates $\{(\hat{\bfpi}_i, \hat{\bfm}_{i}, \hbftau_i): i\in [n]\}$, we define the RIPW estimator as
\begin{equation}
  \label{eq:two_stage}
  \hat{\tau}(\bfPi) \triangleq \argmin_{\tau, \mu, \sum_{i}\alpha_{i} = \sum_{t}\lambda_{t} = 0} \sum_{i=1}^{n}\sum_{t=1}^{T}((Y_{it} - \hat{m}_{it} - \htau_{it}W_{it}) - \mu - \alpha_{i} - \lambda_{t} - W_{it}\tau)^{2}\frac{\bfPi(\W_{i})}{\hat{\bfpi}_{i}(\W_{i})}.
\end{equation}
The above estimator generalizes \eqref{eq:weighted_fe} by allowing for regression adjustment. Throughout the rest of the paper, we will abuse the notation by denoting it as $\hat{\tau}(\bfPi)$. This two-stage formulation replaces the regression with covariates by regression on the modified outcome $(Y_{it} - \hat{m}_{it} - \htau_{it}W_{it})$ without covariates, yielding a simplified structure which allows us to use previously established results. In the rest of this section, we discuss 
formal properties they need to satisfy to guarantee consistency and asymptotic normality of $\hat{\tau}(\bfPi)$.


In the previous section, we assumed that the researcher controlled the assignment process, which led to the restriction (\ref{eq:design_framework}). In observational studies, the assignment process is unknown, and we must substitute this restriction with a different assumption.
Throughout this section, we impose a high-level restriction on the relationship between unit-specific potential outcomes and assignment paths.
\begin{assumption}\label{as:latent_ign}\sc{(unit-specific mean ignorability)}
\begin{equation}\label{eq:latent_ign}
   \E[(\Y_i(1), \Y_i(0))\mid \W_i] = \E[(\Y_i(1), \Y_i(0))], \quad i = 1,\ldots, n.
\end{equation}
\end{assumption}
Recall that we do not assume that $(\Y_i(1), \Y_i(0),\W_i)$ are identically distributed across units. As a result, Assumption \ref{as:latent_ign} imposes $n$ separate restrictions, one for each unit.  It follows the tradition of the part of the panel data literature that treats unit-specific unobservables as fixed parameters \citep[e.g.][]{lancaster2000incidental,hahn2004jackknife}, rather than random variables as in \citep{chamberlain1984panel}. It is trivially satisfied in an extreme case where $(\Y_i(1), \Y_i(0))$ has a degenerate distribution for each $i \in [n]$, which corresponds to the finite population analysis \citep[e.g.][]{abadie2020sampling}. In applications where $(\Y_i(1), \Y_i(0))$ is random, this assumption imposes a strict exogeneity restriction. It describes the average behavior of the outcomes conditional on the whole treatment path and does not allow the current treatment to depend on past outcomes. To illustrate this connection, consider the classical linear TWFE model where 
\begin{equation}\label{eq:classical_TWFE}
Y_{it} = \mu + \alpha_i + \lambda_t + X_{it}^\tran \beta + \tau W_{it} + \eps_{it}, \quad \text{where }\sum_{i=1}^{n}\alpha_i = \sum_{t=1}^{T}\lambda_t = 0.
\end{equation}
Assumption \ref{as:latent_ign} is equivalent to $\E[\eps_{it}\mid \W_i] = 0$ for $t = 1, \ldots, T$, which is a strict exogeneity restriction. In contrast, if $\{\eps_{it}: i\in [n], t\in [T]\}$ only satisfies contemporenous restrictions $\E[\eps_{it}\mid W_{it}] =   0$, Assumption \ref{as:latent_ign} does not necessarily hold. We want to note that in the DiD literature it is common to impose restrictions only on $\Y_i(0)$, while Assumption \ref{as:latent_ign} restricts both potential outcomes. This is necessary given our focus on the ATE, defined in Section \ref{subsec:estimand}.

Assumption \ref{as:latent_ign} is also related to the recent cross-sectional literature on quasi-experimental designs \citep[e.g.][]{borusyak2023non}. A typical restriction in that literature is that while the distribution of the treatment of interest varies over units in a complicated way it still can be estimated and then used to construct counterfactuals. For this approach to be valid, one needs to impose a version of Assumption \ref{as:latent_ign}. In the panel data literature, this type of quasi-experimental variation was also exploited. For example, \cite{wojtaszek2022sensitivity} studied the effect of military bonuses on charitable giving and found that the timing of receiving the bonus is (nearly) as-if random. Depending on the choice of outcome variable, the bonus can be viewed as a staggered or one-off treatment with a uniform generalized propensity score.

To construct estimators $\{(\hat{\bfpi}_i, \hat{\bfm}_{i}, \hbftau_i): i\in [n]\}$ we use the observed covariates $\{\X_i, i \in [n]\}$. Our assumptions implicitly restrict the set of feasible covariates. In particular to respect Assumption \ref{as:latent_ign}, we do not allow any parts of the observed outcomes $\Y_i$ to be used as covariates. The situation is more delicate for $\W_i$, and we allow functions of $\W_i$ to be part of $\X_i$ as long as Assumption \ref{as:overlap} holds. We elaborate on this in the next two sections. 

\subsection{Assignment model estimation}

In strictly exogenous panel models, the distribution of $\W_i$ is commonly left unspecified and the analysis is based on the outcome model alone. In particular, the distribution of $\W_i$ can be degenerate for each $i \in [n]$, which is another extreme case where Assumption \ref{as:latent_ign} trivially holds. However, researchers often informally appeal to random or quasi-random variation in $\W_i$ as a source of identification, even though they continue using outcome-based methods, such as the TWFE regression. We interpret these informal statements as statistical restrictions on $\{\bfpi_i, i \in [n]\}$ that go beyond Assumption \ref{as:latent_ign}.

Precisely because the arguments used in the applied work are often informal, we cannot offer and analyze a general methodology of how to use them to construct $\{\hat{\bfpi}_i, i \in [n]\}$. Instead, we discuss several strategies that are potentially relevant for a large class of applications. Our goal is to demonstrate how to utilize the information used to construct the outcome-based estimators and thus is readily available. In practice, researchers can have other sources of information that we do not incorporate in our analysis. After this discussion, we continue our formal analysis under high-level assumptions on $\{\hat{\bfpi}_i, i \in [n]\}$.

We use $\{(\W_i, \X_i): i\in [n]\}$  to estimate $\bfpi_i$. At first glance, it might appear to be challenging to estimate the distribution of the whole vector. Nevertheless, treatment paths often have restricted support with a size much smaller than $2^T$, such as staggered adoption and/or special structures that reduce the complexity of the distribution, such as the Markov structure. We present a few examples below for illustration.

In the staggered adoption designs $\W_i$ is equivalent to an adoption time $A_i \in \{1, \ldots, T, \infty\}$, where $A_i = \infty$ for never-treated units and $A_i = t$ for units initially treated at time $t$. Then $A_i$ can be viewed as an event or during outcome, and one can apply any survival or duration model, such as the Cox proportional hazard model and accelerated failure time model, to estimate its distribution which yields $\bfpi_i$ by taking the difference between the consecutive points; see Section \ref{subsec:covid} for an empirical illustration that uses this strategy and additional discussion.  For transient treatments that occur at most once during the study period, $\W_i$ can be expressed by the adoption time $A_i \in \{1, \ldots, T, \infty\}$ as above. The propensity score $\bfpi_i$ can then be estimated via a discrete choice model. 

For general designs where the treatment can be alternated on and off, $\bfpi_i$ can be reparametrized as a sequence of conditional distributions $\P(W_{it}\mid W_{i(t-1)}, \ldots, W_{i1}, \X_{i})$ and estimated by a Markov model. In particular, \cite{arkhangelsky2019double} show that if $\X_i$ incorporates appropriate sufficient statistics, then conditioning on $\X_i$ eliminates the ex-ante present unobserved heterogeneity from the distribution of $\W_i$. \cite{aguirregabiria2021sufficient} show that these assumptions are satisfied by a large class of models that are widely used in economic applications, including structural models with forward-looking agents as well as myopic (backward-looking) dynamic logit models. They also provide explicit characterizations for sufficient statistics in such models. These results can be directly applied in our setting.

Given an estimate $\hat{\bfpi}_i$, we say that it estimates the assignment model well if $\hat{\bfpi}_{i}$ is close to $\bfpi_{i}$ in $L^2$ distance. Specifically, for each unit $i$ we define the accuracy of $\hat{\bfpi}_{i}$ as 
\begin{equation}
  \label{eq:deltapii}
  \delta_{\pi i} \triangleq  \sqrt{\E[(\hat{\bfpi}_{i}(\W_i) - \bfpi_{i}(\W_i))^2]}.
\end{equation}
Here, the expectation is taken over both $\W_i$ and $\hat{\bfpi}_i$ (conditional on $\{\X_i: i\in [n]\}$). In the setting of Section \ref{sec:design_based}, $\delta_{\pi i} = 0$ because $\hat{\bfpi}_i = \bfpi_{i}$.

\subsection{Outcome model estimation}
In this section, we discuss the construction of the terms $\{(\hat{\bfm}_{i}, \hbftau_i): i\in [n]\}$ which we use to build the estimator (\ref{eq:two_stage}). We start with unit specific quantities $(\hat{\bfmu}_{i}(0), \hat{\bfmu}_i(1))$, which we view as estimators for $(\E[\Y_i(0)], \E[\Y_i(1)])$. There are many ways of constructing such estimators, and our results require only high-level restrictions on these objects. For example, one can consider a generalization of the linear TWFE model \eqref{eq:classical_TWFE}: 
\begin{equation}\label{eq:interacted_TWFE}
\E[Y_{it}(w)] = \mu + \alpha_i + \lambda_t + X_{it}^\tran\beta + (\tau + X_{it}^\tran\phi)w, \quad \text{where }\sum_{i=1}^{n}\alpha_i = \sum_{t=1}^{T}\lambda_t = 0.
\end{equation}
Then $(\hat{\mu}_{it}(0), \hat{\mu}_{it}(1))$ can be chosen as 
\begin{equation}\label{eq:interacted_TWFE_estimate}
\hat{\mu}_{it}(w) = \hat{\mu} + \hat{\alpha}_i + \hat{\lambda}_t + X_{it}^\tran \hat{\beta} + (\hat{\tau}  +X_{it}^\tran \hat{\phi})w,
\end{equation}
where the parameters are estimated by regressing $Y_{it}$ on $X_{it}, W_{it}$, the covariate-treatment interaction $X_{it}W_{it}$, and a set of fixed effects. When we estimate $\hat{\mu}_{it}(w)$ for a new unit whose unit fixed effect is not estimated, we can simply set $\hat{\mu}_{it}(w) = \hat{\mu} + \hat{\lambda}_t + X_{it}^\tran \hat{\beta} + (\hat{\tau}  +X_{it}^\tran \hat{\phi})w$.

In the cross-sectional case, an estimate $(\hat{\bfmu}_{i}(0), \hat{\bfmu}_i(1))$ is considered an accurate estimate of $(\E[\Y_{i}(0)], \E[\Y_i(1)])$ if $\{\|\hat{\bfmu}_{i}(0) - \E[\Y_{i}(0)]\|_{2} + \|\hat{\bfmu}_i(1) - \E[\Y_i(1)]\|_{2}: i\in [n]\}$ is small on average \citep[e.g.][]{robins1994estimation, kang2007demystifying}. Constructing such estimators for panel models with fixed effects and a finite number of periods is impossible. Thus, the standard approach of measuring accuracy does not apply in our setting, and we need to consider alternative measures.

We start by defining the estimands $(m_{it}, \nu_{it})$ that $(\hat{m}_{it}, \hat{\nu}_{it})$ attempt to estimate:
 \begin{align}
& m_{it} = \E[Y_{it}(0)] - \frac{1}{n}\sum_{i=1}^{n}\E[Y_{it}(0)] - \frac{1}{T}\sum_{t=1}^{T}\E[Y_{it}(0)] + \frac{1}{nT}\sum_{i=1}^{n}\sum_{t=1}^{T}\E[Y_{it}(0)] \label{eq:m} \\
&\nhtau_{it}  \triangleq \tau_{it} - \tau^{*}(\xi).\label{eq:nu}
\end{align}
To measure the degree of mispecification of the outcome model we introduce the following quantity: 
\begin{equation}\label{eq:deltayi}
\delta_{yi} \triangleq \sqrt{\E[\|\hat{\bfm}_{i} - \bfm_{i}\|_{2}^2] + \E[\|\hbftau_{i} - \nhbftau_i\|_{2}^2]}.
\end{equation}
The first term captures the estimation accuracy of $\bfm_i$, and the second term captures the estimation accuracy of $\bftau_i$. By definition, $\delta_{yi}$ is invariant if we replace $\E[Y_{it}(0)]$ by $\E[Y_{it}(0)] + \mu' + \alpha_i' + \lambda_t'$ and $\tau_{it}$ by $\tau_{it} + \tau'$ for any $\mu', \tau', \{\alpha_i': i\in [n]\}$, and $\{\lambda_t': t\in [T]\}$. Thus, requiring $\delta_{yi}$ to be small is strictly less stringent than requiring the standard measure of outcome model accuracy for cross-sectional data to be small. 

In the simplest TWFE model \eqref{eq:classical_TWFE} without covariates, $\delta_{yi} = 0$ if we choose $\hat{\mu}_{it}(0) = \hat{\mu}_{it}(1) = 0$. For the more general TWFE model \eqref{eq:interacted_TWFE}, regardless whether unit $i$ is used for fitting the TWFE regression,
\[\delta_{yi} = \sqrt{\E\left[\sum_{t=1}^{T}\{(X_{it} - \bar{X}_{i\cdot} - \bar{X}_{\cdot t} + \bar{X}_{\cdot \cdot})^\tran (\hat{\beta} - \beta)\}^2 + \{(X_{it} - \sum_{t=1}^{T}\xi_{t}\bar{X}_{\cdot t})^\tran (\hat{\phi} - \phi)\}^2\right]},\]
Standard assumptions \citep[e.g.][]{arellano2003panel, wooldridge2010econometric} guarantee that $(\hat{\beta}, \hat{\phi})$ are consistent for $(\beta, \phi)$ even with a finite number of periods. We can further generalize the model by replacing $X_{it}^\tran\beta$ and $X_{it}^\tran \phi$ with nonlinear functions $g(X_{it})$ and $\tau(X_{it})$ and estimate them by nonparametric TWFE regressions \citep{boneva2015semiparametric}.

The requirement that $\delta_{yi} \approx 0$, at least on average, puts restrictions on the treatment effects. These requirements, however, can be redundant, depending on the structure of $\X_i$. For example, \cite{wooldridge2021two} shows that the problems with heterogeneous treatment effects can be solved, under conditional parallel trends and linearity, by including a sufficiently rich set of controls, which includes functions of $\W_i$. In the staggered adoption case, one needs to include interactions with all the adoption dates. Unfortunately, including such interactions into $\X_i$ violates the overlap assumption \ref{as:overlap}. 


\subsection{Consistency of RIPW estimators}

In this and the next subsection, we consider a simplified case where the estimates $\{(\hat{\bfpi}_i, \hat{\bfm}_{i}, \hbftau_i): i\in [n]\}$ are independent of the data (e.g., obtained from external data). While this is not always possible in practice, the theory of consistency and asymptotic normality can be stated without much mathematical complication. Moreover, these results are building blocks for the theory of cross-fitting estimator described at length in Appendix \ref{app:cross-fitting}. To ease implementation, we provide a self-contained description of the (derandomized) cross-fitting RIPW estimator in Algorithm \ref{algo:RIPW_derandomized_cf} at the end of the next subsection.
\begin{assumption}\label{as:overlap_dr}
There exists $c > 0$ such that, for the same $\S^{*}$ defined in Assumption \ref{as:overlap},
  \begin{equation*}
    \hat{\bfpi}_{i}(\w) \ge c,\quad  \forall \w \in \S^{*}, i\in [n], \quad \text{almost surely}.
  \end{equation*}
\end{assumption}

\begin{assumption}\label{as:bounded_outcomes_dr}
There exists $M < \infty$ such that $\max_{i, t, w}\E[\hat{m}_{it}^2 + \htau_{it}^2] \le M$.
\end{assumption}

Theorem \ref{thm:bias} implies that the RIPW estimator with $\bfPi$ being a solution of the \DATEeq, if any, is a consistent estimator of DATE without any outcome model when $\hat{\bfpi}_i = \bfpi_{i}$ is known. On the other hand, when the outcome model is correctly specified, $Y_{it} - \hat{m}_{it} - \htau_{it}W_{it}\approx Y_{it} - m_{it} - (\tau_{it} - \tau^{*}(\xi))W_{it}$ 
is a linear model with two-way fixed effects and a single predictor $W_{it}$ and $\hat{\tau}$ is approximately a weighted least squares estimator which is consistent under mild conditions on the weights \citep[e.g.,][]{wooldridge2010econometric}. This shows a weak double robustness property that $\hat{\tau}(\bfPi)$ is consistent if either the outcome model or the assignment model is exactly correct.

For cross-sectional data, the augmented IPW estimator enjoys a strong double robustness property, which states that the asymptotic bias is the product of estimation errors of the outcome and assignment models \citep[e.g.,][]{robins1994estimation, kang2007demystifying, chernozhukov2017double, chernozhukov2018double}. Clearly, this implies the weak double robustness. It further implies the estimator has higher asymptotic precision than estimators based on merely the outcome or assignment modeling when both models are estimated well. The next result provides a sufficient condition for strong double robustness of $\hat{\tau}(\bfPi)$ when the estimated treatment and outcome models are independent of the data.

\begin{theorem}\label{thm:doubly_robust_expansion_maintext}
 Assume that $\{(\hat{\bfpi}_i, \hat{\bfm}_{i}, \hbftau_i): i\in [n]\}$ are independent of the data. Under Assumptions \ref{as:limit_dep}-\ref{as:bounded_outcomes} and \ref{as:latent_ign} - \ref{as:bounded_outcomes_dr}, conditional on the estimates,

 \[\hat{\tau}(\bfPi) = \tau^{*}(\xi) + O_\P(\bar{\delta}_{\pi}\bar{\delta}_{y}), \quad \text{where }\bar{\delta}_{\pi} = \sqrt{\frac{1}{n}\sum_{i=1}^{n}\delta_{\pi i}^2}, \quad \bar{\delta}_{y} = \sqrt{\frac{1}{n}\sum_{i=1}^{n}\delta_{y i}^2}.\]
 In particular, $\hat{\tau}(\bfPi)$ is a consistent estimator of $\tau^{*}(\xi)$ if 
  $\bar{\delta}_{\pi}\bar{\delta}_{y} = o(1)$.
\end{theorem}
Assumptions \ref{as:bounded_outcomes} and \ref{as:bounded_outcomes_dr} guarantee that $\bar{\delta}_y$ is bounded. Thus, the RIPW estimator is consistent whenever $\bfpi_i$ is consistently estimated without any requirement on the rate of convergence. On the other hand, under the TWFE model  \eqref{eq:interacted_TWFE} or nonparametric TWFE models discussed in the last subsection, $\bar{\delta}_{y} = o(1)$ and the estimator is consistent even if the assignment model is globally misspecified.

\subsection{Inference with independent model estimates}

Similar to Theorem \ref{thm:design_linear_expansion}, we can derive an asymptotic linear expansion for $\hatdenom\cdot \sqrt{n}(\hat{\tau}(\bfPi) - \tau^{*}(\xi))$. 

\begin{theorem}\label{thm:dr_linear_expansion}
  Assume that $\{(\hat{\bfpi}_i, \hat{\bfm}_i, \hbftau_i): i \in [n]\}$ are independent of the data. Let $\stheta, \sww, \sw$, and $\hatdenom$ be defined as in Theorem \ref{thm:design_linear_expansion} with $\hat{\bfpi}_i$ in the definition of $\Theta_i$. Redefine $\swy, \sy$, and $\V_i$ by replacing $(\Y_i(0), \Y_i(1))$ with $(\td{\Y}_{i}(0), \td{\Y}_{i}(1)) = (\Y_{i}(0) - \hat{\bfm}_{i}, \Y_{i}(1) - \hat{\bfm}_{i} - \hbftau_i)$. Under Assumptions \ref{as:limit_dep} - \ref{as:bounded_outcomes} and \ref{as:latent_ign} - \ref{as:bounded_outcomes_dr},
  \begin{equation}\label{eq:expansion_independent_estimates}
      \hatdenom\cdot \sqrt{n}(\hat{\tau}(\bfPi) - \tau^{*}(\xi)) = \frac{1}{\sqrt{n}}\sum_{i=1}^{n}(\V_{i} - \E[\V_{i}]) + O_\P(n^{1/2-2q} + \bar{\delta}_{\pi}\bar{\delta}_{y}).
      \end{equation}
 In particular, the last term is $o_\P(1/\sqrt{n})$ if $q > 1/2$ and $\bar{\delta}_{\pi}\bar{\delta}_{y} = o(1/\sqrt{n})$.
\end{theorem}

Similar to Section \ref{sec:design_based}, we can estimate $\hat{\V}_i$ and the asymptotic variance via \eqref{eq:var_estimate} and construct the Wald-type confidence interval as \eqref{eq:confidence_interval} when units are independent. This is a special case of Theorem \ref{thm:design_based_generic_inference} in Appendix \ref{subapp:design_based_dependent} for general dependent designs.
\begin{theorem}\label{thm:dr_coverage}
  Assume that $\{(\Y_i(1), \Y_i(0), \W_i): i \in [n]\}$ are independent. Under the same settings as in Theorem \ref{thm:dr_linear_expansion},
    \[\liminf_{n\rightarrow \infty}\P\lb \tau^{*}(\xi)\in \hat{C}_{1 - \alpha}\rb\ge 1 - \alpha,\]
    if, further, \eqref{eq:variance_lower} holds. 
\end{theorem}

Under Assumption \ref{as:latent_ign}, Theorem \ref{thm:dr_linear_expansion} and Theorem \ref{thm:dr_coverage} strictly generalize Theorem \ref{thm:design_linear_expansion} and Theorem \ref{thm:design_inference} -- when $\bfpi_i$ is known, $\bar{\delta}_{\pi} = 0$ and hence $\bar{\delta}_{\pi}\bar{\delta}_{y} = 0 = o(1/\sqrt{n})$ regardless of the accuracy of the outcome model estimates. When $\bfpi_i$ is unknown, $\bar{\delta}_{\pi}$ and $\bar{\delta}_{y}$ are typically no less than $O(1/\sqrt{n})$ without external data. As a result, both models should be consistently estimated to achieve $\bar{\delta}_{\pi}\bar{\delta}_{y} = o(1/\sqrt{n})$ though the estimates can have a slower convergence rate than $O(1/\sqrt{n})$. For example, it would be satisfied if $\bar{\delta}_{\pi}, \bar{\delta}_{y} = o(n^{-1/4})$. We emphasize that this rate requirement is standard for inference with cross-sectional data \citep{chernozhukov2017double, chernozhukov2018double}. Under this rate condition, by virtue of the asymptotic linear expansion in Theorem \ref{thm:dr_linear_expansion}, the researcher can safely ignore the variability of the model estimates and use them in the variance calculation as if they are the truth. 

Even when condition $\bar{\delta}_{\pi}\bar{\delta}_{y} = o(1/\sqrt{n})$ is violated, the asymptotically valid inference is still possible at the cost of more involved variance estimation. 
Doubly robust inference in this regime is generally hard \citep[e.g.][]{benkeser2017doubly}. We consider the setting where parametric models are used to fit the generalized propensity score and regression adjustment. This setting has been studied in the literature for cross-sectional data \citep[e.g.][]{cao2009improving}. Our formal results are deferred in Appendix \ref{subapp:global_misspecification} due to the mathematical complication.
Roughly speaking, if the estimators $\{(\hat{\bfpi}_i, \hat{\bfm}_i, \hbftau_i): i \in [n]\}$ come from a smooth parametric model, then one can use their asymptotic expansion (around their limits, which do not necessarily correspond to the true parameters) to compute the asymptotic variance. Similar to the case discussed Section \ref{sec:design_based} and this section, we can obtain an asymptotically conservative variance estimator without knowing which model is misspecified apriori.

In practice, it is uncommon to obtain estimates of $(\hat{\bfpi}_i, \hat{\bfm}_i, \hbftau_i)$ that are independent of the data, except in the design-based inference where $\hat{\bfpi}_i = \bfpi_i$ and $\hat{\bfm}_i = \hbftau_i = \zero_{T}$, or when external data is available. Usually, these parameters need to be estimated from the data. The resulting dependence invalidates the assumptions of Theorem \ref{thm:dr_linear_expansion} and \ref{thm:dr_coverage}. However, as we show in Appendix \ref{app:cross-fitting} similar results hold if we use a particular version of cross-fitting. Note that this implies that $\hat{\bfpi}_i, \hat{\bfm}_i, \hbftau_i$ cannot contain unit-specific fixed effects. Moreover, we propose a simple approach to mitigate the randomness introduced by sample splitting. We describe the estimator in Algorithm \ref{algo:RIPW_derandomized_cf}. More details can be found in Appendix \ref{app:cross-fitting}. We implemented this method in an \texttt{R} package \texttt{ripw} that is available at \url{https://github.com/lihualei71/ripw}.

\begin{algorithm}[h]
\caption{RIPW estimator with derandomized cross-fitting}
\label{algo:RIPW_derandomized_cf}
\begin{algorithmic}
\State \textbf{Input: }data $\{(\X_i, \W_i, \Y_i): i\in [n]\}$, number of folds $K$, number of data splits $B$,\\
 \qquad \quad \,\, reshaped distribution $\bfPi$
\State 
\Procedure{}{}
\For{$b = 1, \ldots, B$}
\State Randomly split $[n]$ into $K$ folds $\mathcal{I}_1, \ldots, \mathcal{I}_K$ with $|\mathcal{I}_j|\in \{\lfloor n/K\rfloor, \lceil n/K\rceil\}$
    \For{$k = 1, \ldots, K$}
        \State Fit the assignment model $\hat{\bfpi}(\w; \mathbf{x})$ using data in $\cup_{j\neq k}\mathcal{I}_j$
        \State Fit the outcome model $ (\hat{\bfm}_t(\mathbf{x}), \hbftau_t(\mathbf{x}))$ using data in $\cup_{j\neq k}\mathcal{I}_j$
        \For{$i \in \mathcal{I}_k$}
            \State $\hat{\bfpi}_i(\w) \gets \hat{\bfpi}(\w; \X_i)$
            \State $(\hat{\bfm}_{it}, \hbftau_{it}) \gets (\hat{\bfm}_t(\X_i), \hbftau_t(\X_i))$ for each $t\in [T]$
        \EndFor 
    \EndFor
    \State Compute $\hat{\tau}^{(b)}(\bfPi)$ via \eqref{eq:two_stage}
    \State Compute $\hatdenom^{(b)}$ defined in  Theorem \ref{thm:dr_linear_expansion}
    \State Compute $\{\hat{\V}_i^{(b)}: i \in [n]\}$ based on \eqref{eq:hatVi}
\EndFor
\State $\hat{\tau}(\bfPi)\gets \sum_{b=1}^{B}\mathcal{D}^{(b)}\hat{\tau}^{(b)}(\bfPi)/\sum_{b=1}^{B}\mathcal{D}^{(b)}$
\State $\bar{\hat{\V}}_i \gets \sum_{b=1}^{B}\hat{\V}_i^{(b)}/\sum_{b=1}^{B}\mathcal{D}^{(b)}$ for each $i\in [n]$
\State $\hat{\sigma}^2\gets$ sample variance of $\{\bar{\hat{\V}}_i: i\in [n]\}$
\State $\hat{C}_{1 - \alpha} \gets [\hat{\tau}(\bfPi) - z_{1 - \alpha / 2}\hat{\sigma} / \sqrt{n}, \hat{\tau}(\bfPi) + z_{1 - \alpha / 2}\hat{\sigma} / \sqrt{n}],$
\EndProcedure
\State
\State \textbf{Output: } the derandomized cross-fitting estimator $\hat{\tau}(\bfPi)$ and confidence interval $\hat{C}_{1 - \alpha}$
\end{algorithmic}
\end{algorithm}

\section{Design-robust event study specifications}\label{sec:dynamic_effects}

A key limitation of our analysis in previous sections is the focus on static models. This is important both theoretically and practically. Theoretically, some policies of interest are transient in nature, e.g., a large infrastructure investment,  but policymakers expect them to have a lasting impact, which requires a dynamic model. Practically, a large part of applied work in economics uses regression models that explicitly incorporate lags of treatment variables. 

We consider a relatively simple class of linear potential outcome modes to address these concerns. For every $i$ and $t$, we specify the potential outcomes as a function of the current treatment $w$ and its $p$ lags:
\begin{equation}\label{as:dynamic_effect}
   \E[Y_{it}(w_0,w_{-1},\dots, w_{-p})] = \mu_{it} + \sum_{l=0}^p\tau_{i,-l}w_{-l}.
\end{equation}
As in Section \ref{sec:dri}, the expectation is conditional on covariates, and we do not require the units to be independent or identically distributed. This model does not restrict the baseline outcomes but puts structure on the dynamic effects of the treatment. First, the effect of the treatment is present only for $p$ periods after it is implemented. Second, the effect is linear, i.e., the causal effect of being treated one period ago, $w_{-1}$, does not depend on whether the unit was treated two periods ago $w_{-2}$. Finally, the effects are homogenous over time, meaning that $\tau_{i,l}$ do not depend on calendar time $t$. These restrictions are important: the first eliminates the possibility of long-term effects, while the other two eliminate state dependence. Still, we think this model is flexible enough to be useful for a large class of empirical applications. 
 
Interestingly, if the treatment timing is fixed and common across units, and $p$ is large enough, then \eqref{as:dynamic_effect} is a parametrization of all realizable potential outcomes and thus does not impose any testable restrictions. To see this, let $q+1$ denote the adoption time and set $p = T-q+1$. Then each unit $i$ has $T + (T - q)$ potential outcomes $\{Y_{it}(\zero_{T}): t\in [T]\}$ and $\{Y_{it}(\zero_{q}, \one_{T-q}): t \in \{q+1, \ldots, T\}\}$. It is easy to see that \eqref{as:dynamic_effect} holds with 
$\mu_{it} = \E[Y_{it}(\zero_{T})]$, $\tau_{i, 0} = \E[Y_{i(q+1)}(\zero_{q}, \one_{T-q})] - \E[Y_{i(q+1)}(\zero_{T})]$, and 
\[\tau_{i, -\ell} = \E[Y_{i(q+\ell+1)}(\zero_{q}, \one_{T-q})] - \E[Y_{i(q+\ell+1)}(\zero_{T})] - (\E[Y_{i(q+\ell)}(\zero_{q}, \one_{T-q})] - \E[Y_{i(q+\ell)}(\zero_{T})]), \quad \ell \in [p].\]
Similar logic extends to staggered adoption designs as long as we treat the assignment as fixed. However, it breaks if we assume that the adoption time is randomly assigned. In this case, we can test the static model from Section \ref{sec:design_based} and the dynamic model (\ref{as:dynamic_effect}) by comparing outcomes across units that were previously treated at different periods. This emphasizes the importance of the assignment model for the analysis of dynamic effects.  

In this case, it is natural to consider the RIPW estimator coupled with an event-study regression model, i.e.,
\begin{align}
&(\hat{\tau}_{0}, \hat{\tau}_{-1}, \ldots, \hat{\tau}_{-p})\nonumber\\
& = \argmin_{\tau_0, \tau_{-1}, \ldots, \tau_{-p}, \mu, \sum_{i}\alpha_i = \sum_{t}\lambda_t = 0}\sum_{i=1}^{n}\sum_{t=1}^{T}\lb Y_{it} - \mu - \alpha_i - \lambda_t - W_{it}\tau_{0} - \sum_{j=1}^{p}W_{i(t-j)}\tau_{-j}\rb^2\frac{\bfPi(\W_i)}{\bfpi_i(\W_i)}\label{eq:RIPW_event}
\end{align}
where $W_{it}$ is defined as $0$ whenever $t \le 0$. Our next result describes the probability limit of $(\hat{\tau}_{0}, \hat{\tau}_{-1}, \ldots, \hat{\tau}_{-p})$. The proof is presented in Appendix \ref{subapp:dynamic}.

\begin{theorem}\label{thm:dynamic_consistency}
Assume that $Y_{it}(w_{0}, w_{-1}, \ldots, w_{-p})$ satisfies \eqref{as:dynamic_effect} and the generalized propensity score $\bfpi_i(\w) \triangleq \mathbb{P}(\W_i = \w\mid \{Y_{it}(\td{\w}): t\in [T], \td{\w}\in \{0, 1\}^{T}\})$ is known. Further assume that $\EPi[(\W_{\ex} - \EPi[\W_{\ex}])\J (\W_{\ex} - \EPi[\W_{\ex}])]$ is positive definite, where 
\[\W_{\ex} = (\W, \W_{-1}, \ldots, \W_{-p})\in \{0, 1\}^{T\times (p+1)}, \quad \W_{-k} = (0, \ldots, 0, W_1, \ldots, W_{T-k})^\tran,\]
and $\W = (W_1, \ldots, W_{T})$ denote a generic random vector drawn from the distribution $\bfPi$. Then, under Assumptions \ref{as:limit_dep}-\ref{as:bounded_outcomes} (with $Y_{it}(w)$ replaced by $Y_{it}(w_0, w_{-1}, \ldots, w_{-p})$), 
\[\hat{\tau}_{-k} = \frac{1}{n}\sum_{i=1}^{n}\tau_{i, -k} + o_\P(1), \quad k = 0, 1, \ldots, p.\]
\end{theorem}

This result justifies using the RIPW estimator in a large class of applications. If $\bfpi_i$-s are unknown, then one can estimate them using one of the strategies discussed in the previous section. Similarly, one can introduce covariates in this model in the same way as before. Also, applied researchers often consider leads in addition to lags in their regressions, especially in the context of staggered adoption designs. To incorporate this practice into our framework, one simply needs to shift the treatment path $\W_{i}$ appropriately. The resulting estimators for the leads can then be used to test for the validity of the underlying model.

We do not establish analogs of Theorems \ref{thm:doubly_robust_expansion_maintext} - \ref{thm:dr_coverage} for this estimator, but we expect them to hold under appropriate technical conditions. In particular, under \eqref{as:dynamic_effect}, if the TWFE model holds for the baseline potential outcomes such that $\mu_{it} = \mu + \alpha_{i} + \lambda_{t}$ and $\tau_{i, -\ell} = \tau_{-\ell}$, then \eqref{eq:RIPW_event} is consistent for $(\tau_0, \tau_{-1}, \ldots, \tau_{-p})$ since it is a weighted least squares estimator for a correctly specified linear model.
Compared to our analysis in previous sections, the reshaping distribution $\bfPi$ does not play a major role in these results. The reason for this behavior is that the model for treatment effects is time-homogeneous. If we relax this assumption and allow for time-varying dynamic effects $\tau_{i,-l,t}$, then the distribution $\bfPi$ becomes important again. The corresponding DATE equation for this problem is more complicated than the one presented in Section \ref{sec:design_based}, and its analysis is beyond the scope of this paper.

\section{Numerical Studies}\label{sec:experiments}

In this section, we investigate the properties of our estimator in simulations and show how to apply it to real datasets. The \texttt{R} programs to replicate all results in this section is available at \url{https://github.com/xiaomanluo/ripwPaper}.

\subsection{Synthetic data}\label{subsec:synthetic}
 To highlight the central role of the reshaping function in eliminating the bias, we focus on inference with known assignment mechanisms. Put another way, in such settings, the bias of the unweighted or IPW estimators is purely driven by the wrong reshaping function rather than other sources of variability. We consider the DATE with $\xi = \one_{T} / T$ for simplicity. We also design a simulation study with unknown assignment mechanisms and present the results in Appendix \ref{app:AIPW}, which involves all $2$-by-$2$ settings with correct/incorrect assignment/outcome model and a detailed comparison between the RIPW estimator and several other competing estimators. 

We consider a short panel with $T = 4$ and sample size $n = 1000$. We generate a single time-invariant covariate $X_{it} = X_{i}$ with $P(X_i = 1) = 0.7$ and $P(X_i = 2) = 0.3$ and a single time-invariant unobserved confounder $U_{it} = U_{i}$ with $U_i\sim \mathrm{Unif}(\{1, \ldots, 10\})$. Within each experiment, the covariates and unobserved confounders are only generated once and then fixed to ensure a fixed design. For treatment assignments, we consider a staggered adoption design, i.e., $\W_i\in \cW^{\mathrm{sta}}$. We assume that $\W_i$ is less likely to be treated when $X_i = 1$. In particular, 
\[\big(\bfpi_i(\w_{(0)}), \bfpi_i(\w_{(1)}), \bfpi_i(\w_{(2)}), \bfpi_i(\w_{(3)}), \bfpi_i(\w_{(4)})\big) = \left\{
    \begin{array}{ll} 
      (0.8, 0.05, 0.05, 0.05, 0.05) & (X_i = 1)\\
      (0.1, 0.1, 0.2, 0.3, 0.3) & (X_i = 2)
    \end{array}
\right..\]

The potential outcome $Y_{it}(0)$ and the treatment effect $\tau_{it}$ are generated as follows:
\[Y_{it}(0) = \mu + \alpha_i + \lambda_t + m_{it} + \epsilon_{it}, \quad m_{it} = \sigma_{m}X_i\beta_{t}, \quad \tau_{it} = \sigma_{\tau} a_{i}b_{t},\]
where $\mu = 0$, $\beta_t = t - 1$, $\alpha_i = 0.5U_i$, $\lambda_t\stackrel{i.i.d.}{\sim} \mathcal{N}(0,1)$, $b_t \stackrel{i.i.d.}{\sim} \mathcal{N}(0,1)$, and $\epsilon_{it} \stackrel{i.i.d.}{\sim} N(0, 1)$. For $a_i$, we consider two settings: we either set $a_i = 1$ thus making $\tau_{it}$ unit-invariant; or  $a_{i}\stackrel{i.i.d.}{\sim}\mathrm{Unif}([0, 1])$, in which case $\tau_{it}$ varies over units and periods. As with the covariates $X_i$, the time fixed effects $\lambda_t$ and factors $a_{i}, b_t$ are generated once for each setting and then fixed over runs. In contrast, $\epsilon_{it}$ will be resampled in every run as the stochastic errors. Note that both $m_{it}$ and $\tau_{it}$ are generated from rank-one factor models. 

The parameters $\sigma_{m}$ and $\sigma_{\tau}$ measures two types of deviations from the TWFE model: $\sigma_{m}$ measures the violation of parallel trend because we will not adjust for $X_i$ in the design-based inference, and $\sigma_{\tau}$ measures the violation of constant treatment effects. We consider two settings: we either set $\sigma_{m} = 1, \sigma_{\tau} = 0$ --- a model without parallel trends, but constant treatment effects; alternatively, we set  $\sigma_{m} = 0, \sigma_{\tau} = 1$ ---  a TWFE model with heterogeneous effects, but parallel trends. In the first setting $\tau_{it} = 0$ regardless of the model for $a_i$,  thus, we have $3$ different scenarios in total.

We consider three estimators: the unweighted TWFE estimator, the IPW estimator, and the RIPW estimator with $\bfPi$ given by \eqref{eq:midpoint}. For each of the three experiments, we resample $W_{it}$'s and $\eps_{it}$'s, while keeping other quantities fixed, for $1000$ times and collect the estimates and the confidence intervals. Figure \ref{fig:bias} presents the boxplots of the bias $\hat{\tau}(\bfPi) - \tau^{*}(\xi)$. In all settings, the unweighted estimator is clearly biased, demonstrating that both the parallel trend and treatment effect homogeneity are indispensible for classical TWFE regression. In contrast, the IPW estimator is biased when the treatment effects are heterogeneous, but unbiased otherwise even if the parallel trend assumption is violated. This is by no means a coincidence; in this case, $\bftau_i = \tau^{*}(\xi)\one_{T}$ for all $i$ and, by Theorem \ref{thm:bias}, the asymptotic bias $\Delta_{\tau}(\xi) = 0$ for RIPW estimators with any reshaped function including the IPW estimator. Finally, as implied by our theory, the RIPW estimator is unbiased in all settings. Moreover, the coverage of confidence intervals for the RIPW estimator is $94.6\%, 95.2\%$, and $94.6\%$ in these three settings, respectively, confirming the inferential validity stated in Theorem \ref{thm:design_inference}.

\begin{figure}
  \centering
  \includegraphics[width = 0.32\textwidth]{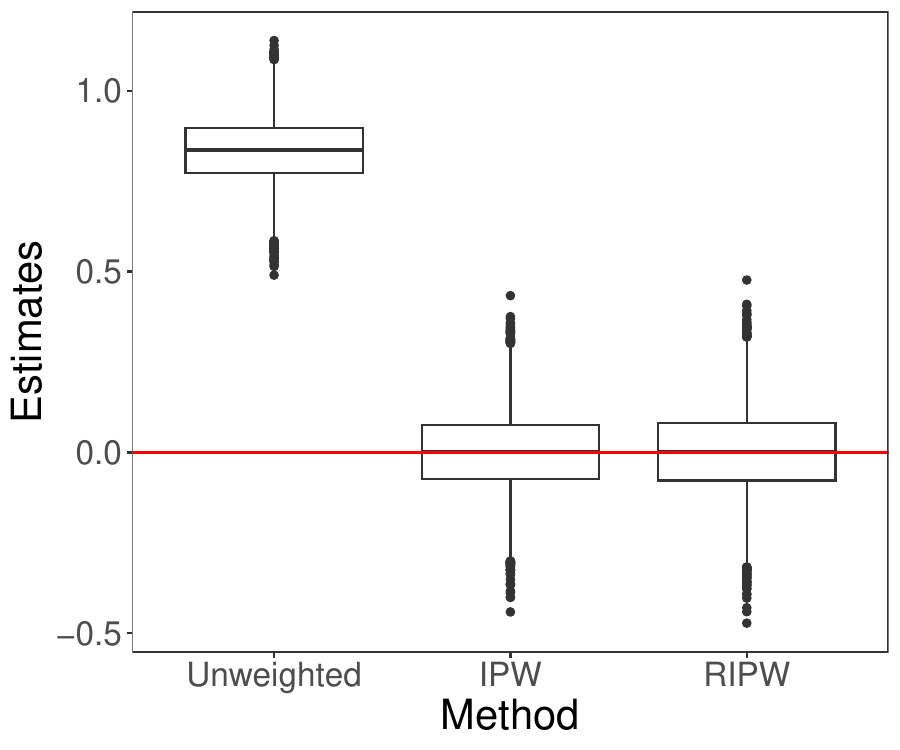}
  \includegraphics[width = 0.32\textwidth]{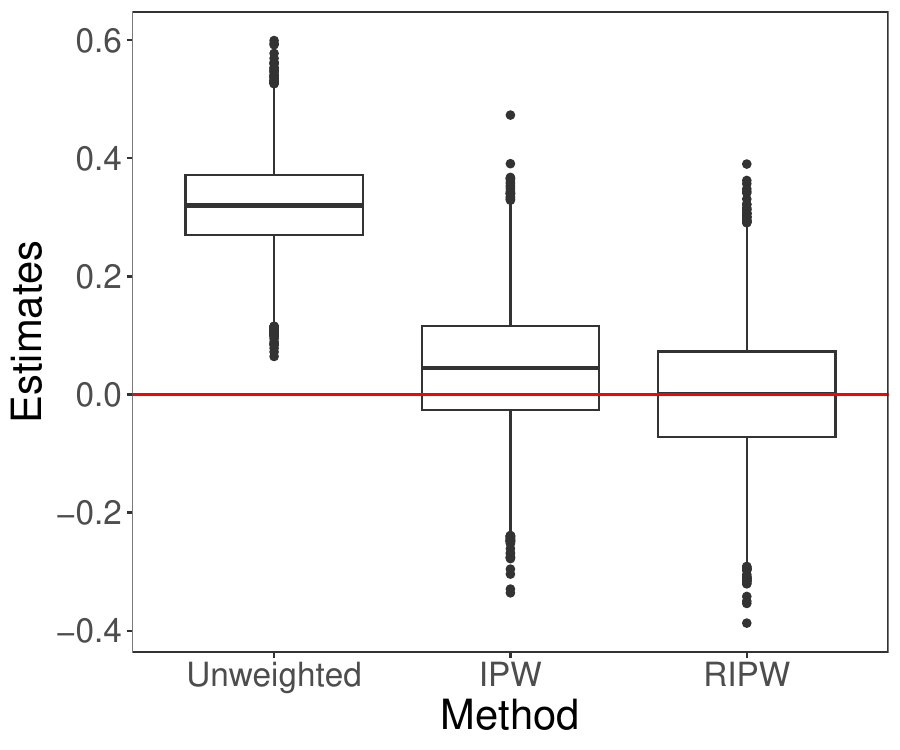}
  \includegraphics[width = 0.32\textwidth]{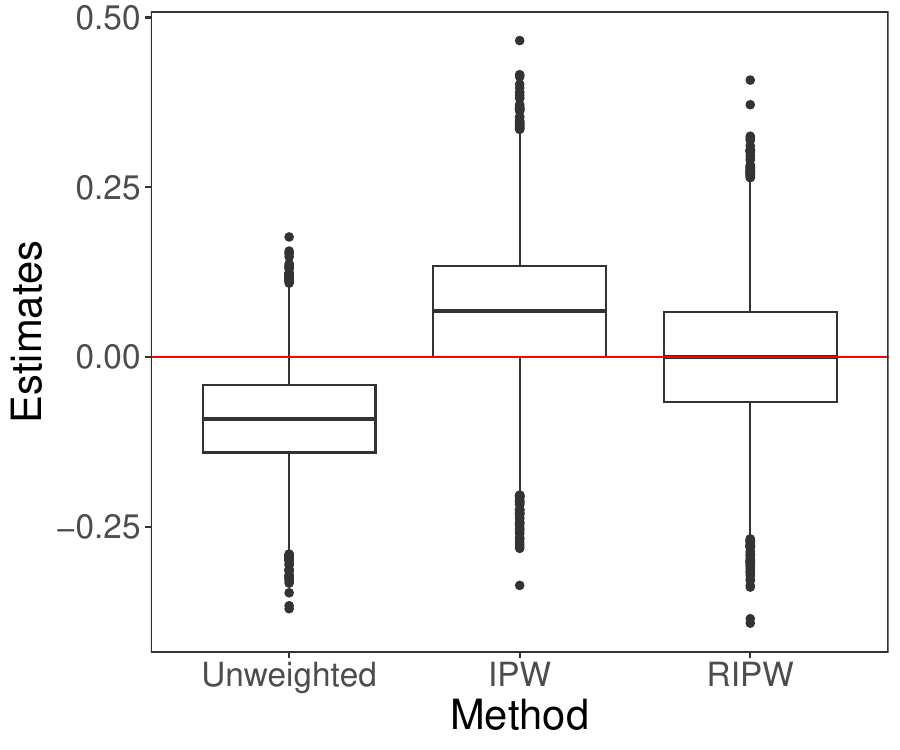}
  \caption{Boxplots of bias across $10,000$ replicates for the unweighted, IPW, and RIPW estimators under (left) violation of parallel trend ($\sigma_{m} = 1, \sigma_{\tau} = 0$), (middle) heterogeneous treatment effect with limited heterogeneity ($\sigma_{m} = 0, \sigma_{\tau} = 1, a_{i} = 1$), and (right) heterogeneous treatment effect with full heterogeneity ($\sigma_{m} = 0, \sigma_{\tau} = 1, a_{i}\sim \Unif([0, 1])$).}\label{fig:bias}
\end{figure}

\subsection{Analysis of OpenTable data in the early COVID-19 pandemic}\label{subsec:covid}

On February 29th, 2020, Washington declared a state of emergency in response to the COVID-19 pandemic. A state of emergency is a situation in which a government is empowered to perform actions or impose policies that it would normally not be permitted to undertake.
It alerts citizens to change their behaviors and urges government agencies to implement emergency plans. As the pandemic has swept across the country, more states declared a state of emergency in response to the COVID-19 outbreak.

The state of emergency restricts various human activities. It would be valuable for governments and policymakers to get a sense of the short-term effect of this urgent action. Since mid-February 2020, OpenTable has been releasing daily data of year-over-year seated diners for a sample of restaurants on the OpenTable network through online reservations, phone reservations, and walk-ins. This provides an opportunity to study how the state of emergency affects the restaurant industry in a short time. The data covers 36 states in the United States, which we will focus our analysis on. 
Policy evaluation in the pandemic is extremely challenging due to the complex confounding and endogeneity issues \citep[e.g.,][]{chetty2020did, chinazzi2020effect, goodman2020using, holtz2020interdependence, kraemer2020effect, abouk2021immediate}. Fortunately, compared to the policies later in the pandemic, the state of emergency suffered from less confounding since it was the first policy that affected the vast majority of the public in the US. On the other hand, the restaurant industry is responding to the policy swiftly because the restaurants are forced to limit and change operations, thereby eliminating some confounders that cannot take effect in a few days.

Despite being more approachable, the problem remains challenging due to the effect heterogeneity and the difficulty of building a reliable model for the dine-in rates in a short time window. In contrast, the declaration time of the state of emergency is arguably less complex to model because it is mainly driven by the progress of the pandemic and the authority's attitude towards the pandemic.

We demonstrate our RIPW estimator on this data. The summary statistics and data sources can be found in Appendix \ref{app:experiment}. The outcome variable is the daily state-level year-over-year percentage change in seated diners provided by OpenTable. The treatment variable is the indicator of whether the state of emergency has been declared. We also include the state-level accumulated confirmed cases to measure the progress of the pandemic, the vote share of Democrats based on the 2016 presidential election data to measure the political attitude towards COVID-19, 
 and the number of hospital beds as a proxy for the amount of regular medical resources. For demonstration purposes, we restrict the analysis to February 29th -- March 13th, the first 14 days since the first declaration by Washington. As of March 13th, 34 out of 36 states have declared a state of emergency; thus, the declaration times are right-censored. The treatment paths are plotted in Figure \ref{fig:treatment_paths}.

 \begin{figure}
 \centering
 \includegraphics[width = 0.7\textwidth]{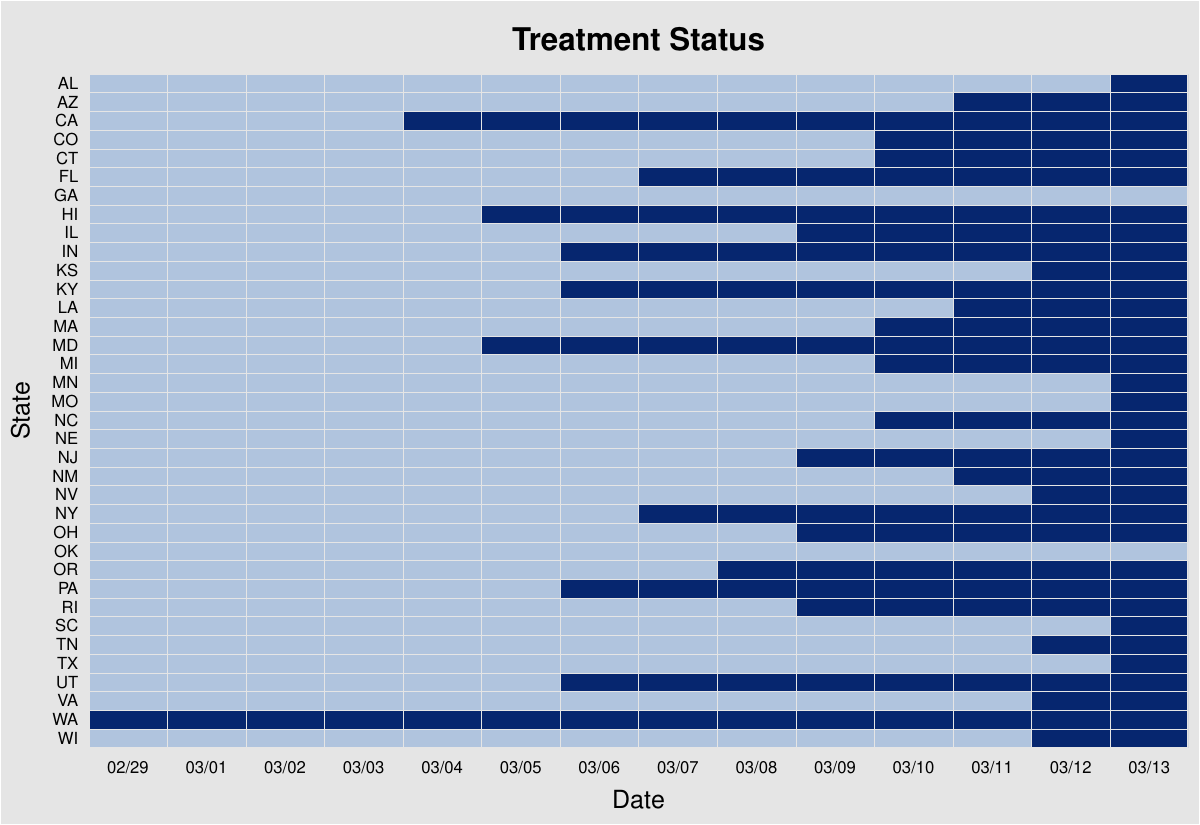}
 \caption{Treatment paths of each state. The darker color marks the treated days. }\label{fig:treatment_paths}
 \end{figure}

For the treatment model, we fit a Cox proportional hazard model on the declaration date to derive an estimate of the generalized propensity scores. Specifically, letting $T_i$ be declaration time of state $i$, a Cox proportional hazard model with time-varying covariates $X_{it}$ assumes that
\[h_i(t\mid X_{it}) = h_0(t)\exp\{X_{it}^\tran \beta\}\]
where $h_i(t\mid \cdot)$ denotes the hazard function for state $i$, and $h_0(t)$ denotes a nonparametric baseline hazard function. The estimates $\hat{h}_0$ and $\hat{\beta}$ yield an estimate $\hat{F}_i(t)$ of the survival function $\P(T_i \ge t)$ for state $i$, differencing which yields an estimate of the generalized propensity score \footnote{For discrete event times, an alternative is the discrete Cox model introduced in Section 6 of \cite{cox1972regression}. Here, we stick with the standard Cox model for simplicity.}
\[\hat{\bfpi}_i(\W_i) = \left\{
    \begin{array}{ll}
      \hat{F}_i(T_i) - \hat{F}_i(T_i + 1) & (\text{State }i \text{ declared state of emergency no later than 03/13})\\
      1 - \hat{F}_i(03/13) & (\text{otherwise})
    \end{array}
\right..\]
Here, we include as the time-varying covariates the logarithms of the accumulated confirmed cases and as the time-invariant covariates the logarithms of the number of hospital beds and the vote share. Note that fixed effects cannot be added into the Cox model because each state has only one outcome. To address unobserved heterogeneity, we include region fixed effects (Northeast, North Central, South, and West). While we will cross-fit the Cox model for the RIPW estimator, we fit the model on the entire data to illustrate the effect of covariates on the adoption time. Table \ref{tab:cox} summarizes the exponentiated parameter estimates along with their standard errors with and without region fixed effects. It also reports the p-value of the joint significance test for the null hypothesis that all coefficients are zero. While most of the coefficients are not significant individually, they are jointly significant, suggesting that the generalized propensity score is non-constant. 


\begin{table}
\centering
\begin{tabular}{lcc}
\toprule
\toprule
 & w/o Region FE & w/ Region FE \\
\midrule
 $\log(\text{confirmed cases})$ & 0.225 & 0.166\\
 & (0.257) & (0.245) \\
 vote share & 0.071$^{**}$ & 0.050 \\
 & (0.029) & (0.036) \\
 $\log(\text{beds})$ & -0.162 & 0.193\\
 & (0.282) & (0.342)\\
 region (South) & & -0.884\\
 & & (0.810)\\
 region (North Central) & & -0.389 \\
 & & (0.731)\\
 region (West) & & 0.396 \\
 & & (0.613) \\ 
\bottomrule
Logrank test p-value & 0.002$^{***}$ & 0.006$^{***}$ \\
\bottomrule
\end{tabular}
\caption{Parameter estimates and standard errors of parameter estimates (in parentheses) for the Cox regression with and without region fixed effects. The bottom row reports the p-value of the joint significant test.}\label{tab:cox}
\end{table}

\begin{figure}[h]
\centering
 \includegraphics[width = 0.6\textwidth]{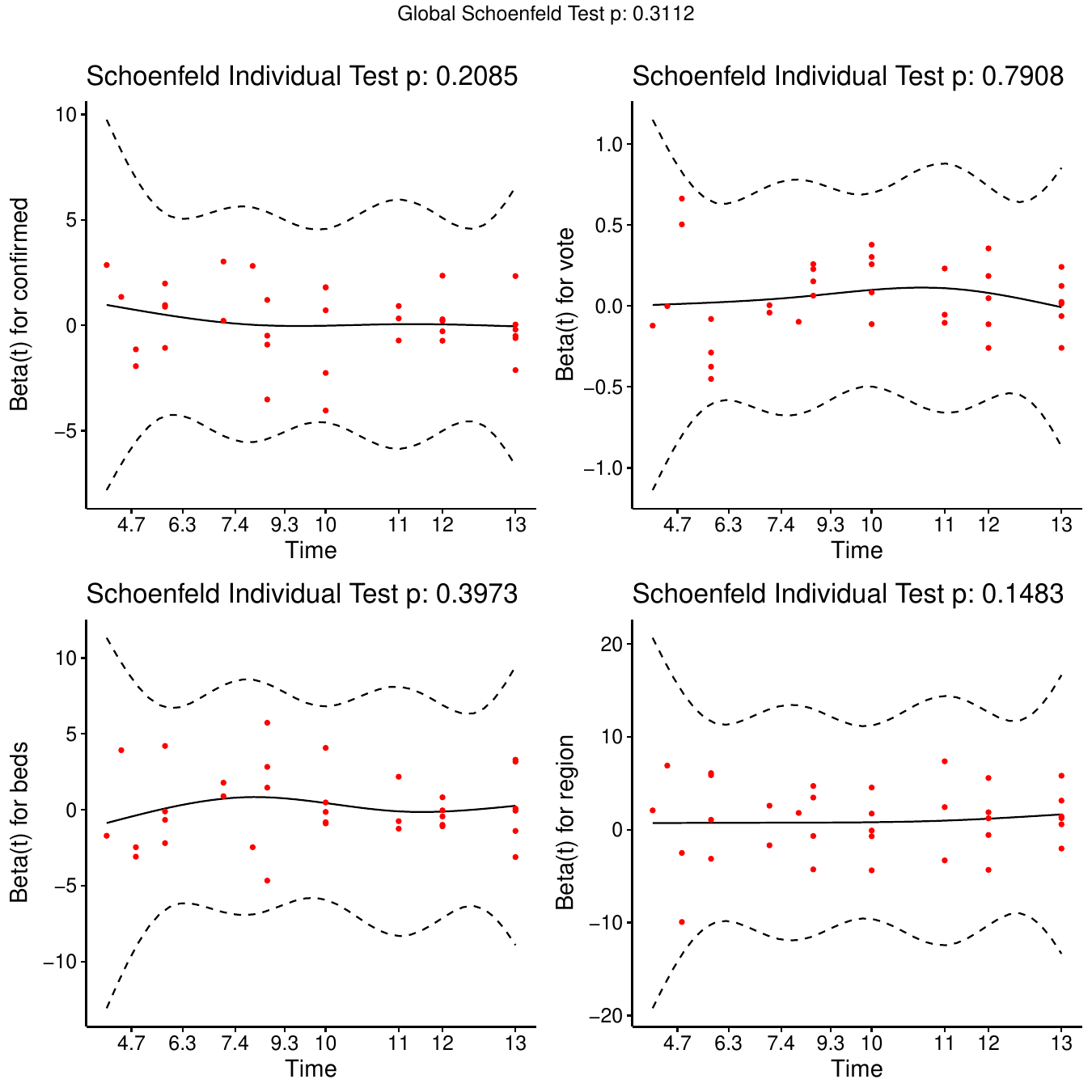}
 \caption{Diagnostics for the Cox proportional hazard model on adoption times.}\label{fig:schoenfeld}
 \end{figure}
 
The proportional hazard assumption imposed by the Cox model is often controversial. Here, we apply the standard statistical tests based on Schoenfeld residuals \citep{schoenfeld1980chi} as a specification test for the Cox model. Figure \ref{fig:schoenfeld} presents the p-values yielded by Schoenfeld's test. Clearly, none of them show evidence against the proportional hazard assumption. The p-value of Schoenfeld's test is $0.311$, suggesting no evidence against the specification.

For the outcome model, we fit an interacted TWFE regression in the form of \eqref{eq:interacted_TWFE} with the same set of covariates. Since unit fixed effects are included, no time-invariant covariate can be added to the main effects due to perfect collinearity. Thus, we add log confirmed cases, treatment, and the interactions between treatment and all variables, including region fixed effects, into the TWFE regression. Table \ref{tab:outcome_model} summarizes the results. The first row gives the treatment effect estimates by the TWFE regressions, though these estimates are irrelevant in the regression adjustment for our RIPW estimator, which only depends on the other rows. Table \ref{tab:cox} reports the p-value of the joint significance test for the null hypothesis that all coefficients other than the two-way fixed effects are zero. Again, the null hypothesis that $m_{it} = 0$ for all $(i, t)$ is rejected in both settings.

\begin{table}[H]
\centering
\begin{tabular}{lcc}
\toprule
\toprule
 & w/o Region FE $ \times$ treat & w/ Region FE $ \times$ treat\\
\midrule
 treat & -0.641 & -0.619 \\
 & (1.640) & (1.640)\\
 $\log(\text{confirmed cases})$ & -3.022$^{**}$ & -2.896$^{**}$\\
 & (1.230) & (1.232) \\
  $\log(\text{confirmed cases})\times \text{treat}$ & 0.580 & 0.662\\
 & (2.466) & (2.491) \\
 $\text{vote share}\times \text{treat}$ & -0.251$^{**}$ & -0.217 \\
 & (0.115) & (0.134) \\
 $\log(\text{beds})\times \text{treat}$ & -0.925 & -0.125\\
 & (1.288) & (1.440)\\
  $\text{region (South)}\times \text{treat}$ & & 4.813 \\
 & & (3.477)\\
  $\text{region (North Central)}\times \text{treat}$ & & 2.681\\
 & & (3.617)\\
  $\text{region (West)}\times \text{treat}$ & & 7.316$^{**}$\\
 & & (3.219)\\
\bottomrule
F-test p-value & 0.002$^{***}$ & 0.001$^{***}$ \\
\bottomrule
\end{tabular}
\caption{Parameter estimates and standard errors (in parentheses) for the unweighted TWFE regression with and without region fixed effects. The bottom row reports the p-value of the joint significant test of all coefficients other than unit- and time-fixed effects.}\label{tab:outcome_model}
\end{table}

Finally, we compute the RIPW estimator for equally-weighted DATE with the reshaped distribution \eqref{eq:midpoint} in Appendix \ref{app:DATE_equation} for staggered adoption and $10$-fold cross-fitting that is described in Algorithm \ref{algo:RIPW_derandomized_cf} and discussed at length in Appendix \ref{app:cross-fitting}. Since this problem has a small sample size, the estimate exhibits large variation across different data splits. We thus apply the de-randomization procedure discussed in Appendix \ref{subapp:derandomization} with $10,000$ splits (i.e., $B = 10,000$ in Algorithm \ref{algo:RIPW_derandomized_cf}). Our de-randomized cross-fitted RIPW estimate is reported in Table \ref{tab:results_summary}, together with the estimates obtained using the TWFE regressions reported in Table \ref{tab:outcome_model}. 
It is significant at the $10\%$ level and the magnitude is larger than that given by the unweighted TWFE regressions shown in Table \ref{tab:outcome_model}. Recall that the joint F-test p-value for the assignment model presents strong evidence of selection and hence the difference between the RIPW estimator and the unweighted TWFE regression are likely due to the bias of the latter. 

\begin{table}
\centering
\begin{tabular}{lccc}
\toprule
\toprule
 & TWFE (w/o Region FE) & TWFE (w/ Region FE) & RIPW \\
\midrule
Estimate & -0.641 & -0.619 & -3.403$^{*}$\\
 & (1.640) & (1.640) & (2.047)\\
90\% CI & [-3.34, 2.06] & [-3.32, 2.08] & [-6.77, -0.04]\\
95\% CI & [-3.86, 2.57] & [-3.83, 2.60] & [-7.42, 0.61]\\
\bottomrule
\bottomrule
\end{tabular}
\caption{Treatment effect estimates, standard errors  (in parentheses) and confidence intervals.}\label{tab:results_summary}
\end{table}


\section{Conclusion}
We demonstrate both theoretically and empirically that the unit-specific reweighting of the OLS objective function improves the robustness of the resulting treatment effects estimator in applications with panel data.  The proposed weights are constructed using the assignment process (either known or estimated) and thus appropriate in situations with substantial cross-sectional variation in the treatment paths. Practically, our results allow applied researchers to exploit domain knowledge about outcomes and assignments, thus resulting in a more balanced approach to identification and estimation.

\newpage
\bibliographystyle{plainnat}
\bibliography{dr_panel}

\newpage
\appendix
\newcommand{\bV}{\mathbf{\mathcal{V}}}
\newcommand{\bA}{\mathbf{A}}  
\renewcommand{\theequation}{\thesection.\arabic{equation}}

\section{Statistical Properties of RIPW Estimators}\label{app:proofs}
\subsection{Setup and preliminaries}\label{subapp:notation}
We will consider the more general setting in Section \ref{sec:dri} since it nests the setting in Section \ref{sec:design_based}. For ease of reference, we state the framework here, along with the list of assumptions some of which are weaker than those stated in the main text.

Suppose the $i$-th unit is characterized by potential outcomes $\Y_i(1) = (Y_{i1}(1), \ldots, Y_{iT}(1)), \Y_i(0) = (Y_{i1}(0), \ldots, Y_{iT}(0))$, the treatment path $\W_i = (W_{i1}, \ldots, W_{iT})$, and a set of covariates $\X_i = (X_{i1}, \ldots, X_{iT})$. The vector of time-varying treatment effects for unit $i$ is denoted by $\bftau_{i} = (\tau_{i1}, \ldots, \tau_{iT}) = \Y_i(1) - \Y_i(0)$. We treat covariates as fixed and consider $\{(\Y_i(1), \Y_i(0), \W_i): i\in [n]\}$ as a random vector (jointly) drawn from a distribution (conditional on $\{\X_i: i\in [n]\}$). We let $\P$ denote the joint distribution of the entire random vector $\{(\Y_i(1), \Y_i(0), \W_i): i\in [n]\}$ (conditional on $\{\X_i: i\in [n]\}$) and $\E$ denote the expectation over this distribution. 

The assignment model is characterized by the generalized propensity score defined as 
\[\bfpi_i(\w) = \P(\W_i = \w).\] 
The outcome model is characterized by $\{(\bfm_{i}, \nhbftau_{i}): i\in [n]\}$ where $\bfm_{i} = (m_{i1}, \ldots, m_{iT}), \nhbftau_{i} = (\nhtau_{i1}, \ldots, \nhtau_{iT})$,
\begin{align*}
m_{it} &= \E[Y_{it}(0)] - \frac{1}{n}\sum_{i=1}^{n}\E[Y_{it}(0)] - \frac{1}{T}\sum_{t=1}^{T}\E[Y_{it}(0)] + \frac{1}{nT}\sum_{i=1}^{n}\sum_{t=1}^{T}\E[Y_{it}(0)],\\
\tau_{it} & = \E[Y_{it}(1)] - \E[Y_{it}(0)],\\
\nhtau_{it} & = \tau_{it} - \tau^{*}(\xi).
\end{align*}
Let $\{(\hat{\bfpi}_i, \hat{\bfmu}_i(0), \hat{\bfmu}_i(1)): i\in [n]\}$ be an estimate of $\{(\bfpi_i, \bfmu_i(0), \bfmu_i(1))\}$. Further let $\hat{\bfm}_i = (\hat{m}_{i1}, \ldots, \hat{m}_{iT})$ and $\hbftau_{i} = (\htau_{i1}, \ldots, \htau_{iT})$, where 
 \begin{align*}
\hat{m}_{it} &\triangleq  \hat{\mu}_{it}(0) - \frac{1}{n}\sum_{i=1}^{n}\hat{\mu}_{it}(0) - \frac{1}{T}\sum_{t=1}^{T}\hat{\mu}_{it}(0) + \frac{1}{nT}\sum_{i=1}^{n}\sum_{t=1}^{T}\hat{\mu}_{it}(0),   \\
\hat{\tau}_{it} & \triangleq \hat{\mu}_{it}(1) - \hat{\mu}_{it}(0),\\
\htau_{it} & \triangleq \hat{\tau}_{it} - \sum_{t=1}^{T}\frac{\xi_{t}}{n}\sum_{i=1}^{n}\hat{\tau}_{it}.
\end{align*}
The results in Section \ref{sec:design_based} are given by the special case where $\hat{\bfpi}_i = \bfpi_i, \hat{\bfm}_i = \hbftau_i = \zero_T$.

Define the modified potential outcomes as 
\[\td{\Y}_i(0) = \Y_i(0) - \hat{\bfm}_i, \quad \td{\Y}_i(1) = \Y_i(1) - \hat{\bfm}_i - \hbftau_i,\]
and the modified treatment effects as 
\begin{equation}\label{eq:tdbftau_i}
\td{\bftau}_i = \E[\td{\Y}_i(1) - \td{\Y}_i(0)] = \bftau_i - \E[\hbftau_i].
\end{equation}
By definition, 
\begin{equation}\label{eq:tau*_tdtau}
\tau^{*}(\xi) = \sum_{t=1}^{T}\frac{\xi_t}{n}\sum_{i=1}^{n}\tau_{it} = \sum_{t=1}^{T}\frac{\xi_t}{n}\sum_{i=1}^{n}\td{\tau}_{it}.
\end{equation}
Then the modified observed outcome is $\td{\Y}_i = (\td{Y}_{i1}, \ldots, \td{Y}_{iT})$ where
\[\td{Y}_{it} = \td{Y}_{it}(1)W_{it} + \td{Y}_{it}(0)(1 - W_{it}) = Y_{it} - \hat{m}_{it} - \htau_{it}W_{it}.\]
With a reshaped distribution $\bfPi$ on $\{0, 1\}^{T}$, the RIPW estimator is defined as 
\[  \hat{\tau}(\bfPi) \triangleq \argmin_{\tau, \mu, \sum_{i}\alpha_{i} = \sum_{t}\lambda_{t} = 0} \sum_{i=1}^{n}\sum_{t=1}^{T}(\td{Y}_{it} - \mu - \alpha_{i} - \lambda_{t} - W_{it}\tau)^{2}\frac{\bfPi(\W_{i})}{\hat{\bfpi}_{i}(\W_{i})}.\]
We will suppress $\xi$ from $\tau^{*}(\xi)$ and $\bfPi$ from $\hat{\tau}(\bfPi)$ throughout the section.

Since $\hat{\tau}(\bfPi)$ remains invariant if we replace $Y_{it}$ by $Y_{it} - \mu' - \alpha_i' - \lambda_t'$, we assume that 
\begin{equation}\label{eq:Yi0_WLOG}
\Y_i(0) = \bfm_i \Longleftrightarrow \td{\Y}_i(0) = \bfm_i - \hat{\bfm}_i,
\end{equation}
by setting $\mu' = (1/nT)\sum_{i=1}^{n}\sum_{t=1}^{T}\E[Y_{it}(0)], \alpha_i' = (1/T)\sum_{t=1}^{T}(\E[Y_{it}(0)] - \mu')$, and \\$\lambda_t'= (1/n)\sum_{i=1}^{n}(\E[Y_{it}(0)] - \mu')$.

The accuracy of the assignment model and the outcome model for unit $i$ are defined as 
\[\delta_{\pi i} = \sqrt{\E[|\hat{\bfpi}_i(\W_i) - \bfpi_i(\W_i)|^2]}, \quad \delta_{y i} = \sqrt{\E[\|\hat{\bfm}_i - \bfm_i\|_{2}^2 + \|\hbftau_i - \nhbftau_i\|_{2}^2]}.\]
In the proofs, we need the conditional version of these measures 
\[\Delta_{\pi i} = \sqrt{\E\left[|\hat{\bfpi}_i(\W_i) - \bfpi_i(\W_i)|^2\mid \hat{\bfpi}_i\right]^2}, \quad \Delta_{y i} = \sqrt{\E[\|\hat{\bfm}_i - \bfm_i\|_{2}^{2}\mid \hat{\bfm}_i] + \E[\|\hbftau_i - \nhbftau_i\|_{2}^2\mid \hbftau_i]}.\]
We then define the unconditional and conditional average accuracy measures
\[\bar{\delta}_{\pi} = \sqrt{\frac{1}{n}\sum_{i=1}^{n}\delta_{\pi i}^2}, \quad \bar{\delta}_{y} = \sqrt{\frac{1}{n}\sum_{i=1}^{n}\delta_{y i}^2},\]
and
\[\bar{\Delta}_{\pi} = \sqrt{\frac{1}{n}\sum_{i=1}^{n}\Delta_{\pi i}^2}, \quad \bar{\Delta}_{y} = \sqrt{\frac{1}{n}\sum_{i=1}^{n}\Delta_{y i}^2}.\]
By law of iterated expectations, 
\[\E[\bar{\Delta}_{\pi}^2] = \bar{\delta}_{\pi}^2, \quad \E[\bar{\Delta}_{y}^2] = \bar{\delta}_{y}^2.\]
By Markov inequality, 
\begin{equation}\label{eq:Delta_delta}
\bar{\Delta}_{\pi} = O_\P(\bar{\delta}_{\pi}), \quad \bar{\Delta}_{y} = O_\P(\bar{\delta}_{y}).
\end{equation}
Therefore, if we can prove the result only assuming $\bar{\Delta}_{\pi}\bar{\Delta}_{y} = o(1)$ conditional on $(\hat{\bfpi}_i, \hat{\bfm}_i, \hbftau_i)_{i=1}^{n}$, we can prove it assuming that $\bar{\delta}_{\pi}\bar{\delta}_{y} = o(1)$ as in Section \ref{sec:dri}.

To be self-contained, we list all quantities involved in the ~\DATEeq and the asymptotically linear expansion of the RIPW estimator. Let $J = I_{T} - \one_{T}\one_{T}^\tran / T$,
\[\Theta_i = \bfPi(\W_i) / \hat{\bfpi}_i(\W_i), \]
\begin{equation*}
  \stheta\triangleq \frac{1}{n}\sum_{i=1}^{n}\Theta_{i}, \quad \sww \triangleq \frac{1}{n}\sum_{i=1}^{n}\Theta_{i}\W_{i}^{\tran}\J\W_{i}, \quad \swy \triangleq \frac{1}{n}\sum_{i=1}^{n}\Theta_{i}\W_{i}^{\tran}\J\td{\Y}_i,
\end{equation*}
\begin{equation*}
\sw \triangleq \frac{1}{n}\sum_{i=1}^{n}\Theta_{i}\J\W_{i}, \quad \sy \triangleq \frac{1}{n}\sum_{i=1}^{n}\Theta_{i}\J\td{\Y}_i,
\end{equation*}
and
\begin{align*}
  \V_i &= \Theta_i\Bigg\{\lb\E[\swy] - \tau^{*}\E[\sww]\rb - \lb \E[\sy] - \tau^{*}\E[\sw]\rb ^\tran\J \W_i \\
       & \,\, + \E[\stheta]\W_i^\tran \J \lb\td{\Y}_i - \tau^{*}\W_i\rb - \E[\sw]^\tran\J\lb \td{\Y}_{i} - \tau^{*}\W_i\rb \Bigg\}.
\end{align*}
This coincides with the definition in Theorem \ref{thm:design_linear_expansion} when $\hat{\bfpi}_i = \bfpi_i$ and $\hat{\bfm}_i = \hbftau_i = \mathbf{0}_{T}$.

Finally, we state the core assumptions, some of which are repeated and combined for ease of reference and the rest of which are weakened. We start by restating the unit-specific mean ignorability assumption.
\begin{assumption}\label{as:latent_ign_app} For each $i \in [n]$,
\begin{equation}
   \E[(\Y_i(1), \Y_i(0))\mid \W_i] = \E[(\Y_i(1), \Y_i(0))].
\end{equation}
\end{assumption}
Next, we combine the overlap condition for the true propensity scores (Assumption \ref{as:overlap}) 
 and that for the estimated propensity scores (Assumption \ref{as:overlap_dr}) with the constant $c$ replaced by $c_{\pi}$ to be more informative in the proofs.
\begin{assumption}\label{as:overlap_general}
  There exists a universal constant $c > 0$ and a non-stochastic subset $\S^{*}\subset \{0, 1\}^{T}$ with at least two elements and at least one element not in $\{\zero_{T}, \one_{T}\}$, such that
  \begin{equation}
    \label{eq:empirical_overlap}
    \hat{\bfpi}_{i}(\w) > c_{\pi}, \bfpi_{i}(\w) > c_{\pi}, \quad \forall \w \in \S^{*}, i\in [n], \quad \text{almost surely}.
  \end{equation} 
\end{assumption}
Lastly, we state the following assumption that unifies and weakens Assumptions \ref{as:limit_dep}, \ref{as:bounded_outcomes}, and \ref{as:bounded_outcomes_dr}.

\begin{assumption}\label{as:regularity_general}
There exists $\q \in (0, 1]$,
  \[\frac{1}{n^2}\sum_{i=1}^{n}\rho_{i}\left\{\E\|\td{\Y}_i(1)\|_{2}^2 + \E\|\td{\Y}_i(0)\|_{2}^2 + 1\right\} = O(n^{-\q}),\]
  and
  \[\frac{1}{n}\sum_{i=1}^{n}\bigg\{ \E\|\td{\Y}_i(1)\|_{2}^2 + \E\|\td{\Y}_i(0)\|_{2}^2\bigg\} = O(1).\]  
\end{assumption}
We close this section by a basic property of the maximal correlation.
\begin{lemma}\label{lem:rho_mixing}
Let $\Z_i = (\Y_i(1), \Y_i(0), \X_i)$ and $f_i$ be any deterministic function on the domain of $\Z_i$. Then
\[\Var\left[\sum_{i=1}^{n}f_i(\Z_{i})\right] \le \frac{1}{2}\sum_{i=1}^{n}\Var[f_i(\Z_{i})]\rho_{i}.\] 
\end{lemma}
\begin{proof}
  By definition of $\rho_{ij}$,
\[\Cov\lb f_i(\Z_{i}), f_j(\Z_{j})\rb\le \rho_{ij}\sqrt{\Var[f_i(\Z_{i})]\Var[f_j(\Z_{j})]}\le \frac{\rho_{ij}}{2}\left\{\Var[f_i(\Z_{i})] + \Var[f_j(\Z_{j})]\right\}.\]
Thus, 
\begin{align*}
  &\Var\left[\sum_{i=1}^{n}f_i(\Z_{i})\right]  = \sum_{i,j=1}^{n}\Cov(f_i(\Z_{i}), f_j(\Z_{j}))\\
& \le \sum_{i,j=1}^{n}\frac{\rho_{ij}}{2}\left\{\Var[f_i(\Z_{i})] + \Var[f_j(\Z_{j})]\right\} = \sum_{i=1}^{n}\Var[f_i(\Z_{i})]\rho_{i}.
\end{align*}
\end{proof}

\subsection{A non-stochastic formula of RIPW estimators}
\begin{theorem}\label{theorem:estimator}
With the same notation as Theorem \ref{thm:design_linear_expansion}, $\hat{\tau} = \hatnumer / \hatdenom$, where
\begin{equation}
  \label{eq:tauhat}
  \hatnumer = \swy\stheta - \sw^{\tran}\sy, \quad \hatdenom = \sww\stheta - \sw^{\tran}\sw.
\end{equation}
\end{theorem}
\begin{proof}
  Let $\bfgamma = (\lambda_{1}, \ldots, \lambda_{t})$ be any vector with $\bfgamma^\tran \one_{T} = 0$. First we derive the optimum $\hat{\mu}(\bfgamma, \tau), \hat{\alpha}_{i}(\bfgamma, \tau)$ given any values of $\bfgamma$ and $\tau$. Recall that
  \[(\hat{\mu}(\bfgamma, \tau), \hat{\alpha}_i(\bfgamma, \tau)) = \argmin_{\sum_{i}\alpha_{i} = 0} \sum_{i=1}^{n}\lb\sum_{t=1}^{T}(\td{Y}_{it} - \mu - \alpha_{i} - \lambda_{t} - W_{it}\tau)^{2}\rb\Theta_{i}.\]
  Since the weight $\Theta_{i}$ only depends on $i$, it is easy to see that
  \[\hat{\mu}(\bfgamma, \tau) + \hat{\alpha}_{i}(\bfgamma, \tau) = \frac{1}{T}\sum_{t=1}^{T}(\td{Y}_{it} - \lambda_{t} - W_{it}\tau), \quad \hat{\mu}(\bfgamma, \tau) = \frac{1}{nT}\sum_{i=1}^{n}\sum_{t=1}^{T}(\td{Y}_{it} - \lambda_{t} - W_{it}\tau).\]
  As a result,
  \begin{align*}
    & \sum_{t=1}^{T}(\td{Y}_{it} - \hat{\mu}(\bfgamma, \mu) - \hat{\alpha}_{i}(\bfgamma, \mu) - \lambda_{t} - W_{it}\tau)^{2}\\
    = & \bigg\|\lb\td{\Y}_{i} - \bfgamma - \W_i \tau\rb - \frac{\one_{T}\one_{T}^\tran}{T}\lb\td{\Y}_{i} - \bfgamma - \W_i \tau\rb\bigg\|_{2}^2\\
    = & \bigg\|\J\lb\td{\Y}_{i} - \bfgamma - \W_i \tau\rb\bigg\|_{2}^2.
  \end{align*}
  This yields a profile loss function for $\bfgamma$ and $\tau$:
  \[(\hat{\bfgamma}, \hat{\tau}) = \argmin_{\bfgamma^\tran\one_{T} = 0} \sum_{i=1}^{n}\bigg\|\J\lb\td{\Y}_{i} - \bfgamma - \W_i \tau\rb\bigg\|_{2}^2\Theta_{i} = \argmin_{\bfgamma^\tran\one_{T} = 0} \sum_{i=1}^{n}\bigg\|\J\lb\td{\Y}_{i} - \W_i \tau\rb - \bfgamma\bigg\|_{2}^2\Theta_{i},\]
  where the last equality uses the fact that $\J \bfgamma = \bfgamma$. Given $\tau$, the optimizer $\hat{\bfgamma}(\tau)$ is simply the weighted average of $\{\J(\td{\Y}_i - \W_i\tau)\}_{i=1}^{n}$ in absence of the constraint $\bfgamma^\tran \one_T = 0$, i.e.
  \[\hat{\bfgamma}(\tau) = \frac{\sum_{i=1}^{n}\Theta_i\J(\td{\Y}_i - \W_i\tau)}{\sum_{i=1}^{n}\Theta_{i}} = \frac{\sy}{\stheta} - \frac{\sw}{\stheta}\tau.\]
  Noting that $\hat{\bfgamma}(\tau)^\tran\one_T = 0$ since $\J \one_T = 0$, $\hat{\bfgamma}(\tau)$ is also the minimizer of the constrained problem, i.e.
  \[\hat{\bfgamma}(\tau) = \argmin_{\bfgamma^\tran\one_{T} = 0} \sum_{i=1}^{n}\bigg\|\J\lb\td{\Y}_{i} - \W_i \tau\rb - \bfgamma\bigg\|_{2}^2\Theta_{i}.\]
  Plugging in $\hat{\bfgamma}(\tau)$ yields a profile loss function for $\tau$
  \[\hat{\tau} = \argmin \sum_{i=1}^{n}\bigg\|\J\lb\td{\Y}_{i} - \W_i \tau\rb - \hat{\bfgamma}(\tau)\bigg\|_{2}^2\Theta_{i}\triangleq L(\tau).\]
  A direct calculation shows that
  \begin{align*}
    \frac{L'(\tau)}{2n}= &\frac{1}{n}\sum_{i=1}^{n}\Theta_{i}\lb -J\W_i + \frac{\sw}{\stheta}\rb^\tran\lb \J\lb\td{\Y}_{i} - \W_i \tau\rb - \frac{\sy}{\stheta} + \frac{\sw}{\stheta} \tau\rb\\
    = & \frac{1}{n}\left\{\sum_{i=1}^{n}\Theta_{i}\lb J\W_i - \frac{\sw}{\stheta}\rb^\tran \lb J\W_i - \frac{\sw}{\stheta}\rb\right\}\tau - \frac{1}{n}\left\{\sum_{i=1}^{n}\Theta_{i}\lb J\W_i - \frac{\sw}{\stheta}\rb^\tran \lb J\td{\Y}_i - \frac{\sy}{\stheta}\rb\right\}\\
    = & \left\{\sww - \frac{\sw^\tran\sw}{\stheta}\right\}\tau - \left\{\swy - \frac{\sw^\tran \sy}{\stheta}\right\}
  \end{align*}
  Since $L(\tau)$ is a convex quadratic function of $\tau$, the first-order condition is sufficient and necessary to determine the optimality. The proof is then completed by solving $L'(\hat{\tau}) = 0$.
\end{proof}

\subsection{Statistical properties of RIPW estimators with  deterministic $(\hat{\bfpi}_i, \hat{\bfm}_{i}, \hbftau_i)$}

\subsubsection{Asymptotic linear expansion of RIPW estimators}\label{subsubapp:asymptotic_linear}
As a warm-up, we assume that $(\hat{\bfpi}_i, \hat{\bfm}_{i}, \hbftau_i)_{i=1}^{n}$ are deterministic. This, for example, includes the pure design-based inference where $\hat{\bfpi}_i = \bfpi_i$ and $\hat{\bfm}_{i} = \hbftau_i = 0$. In this case, the measures of accuracy can be simplified as
\begin{equation}
  \label{eq:delta_deterministic}
  \Delta_{\pi i} = \sqrt{\E\left[\hat{\bfpi}_i(\W_i) - \bfpi_i(\W_i)\right]^2}, \quad \Delta_{y i} = \sqrt{\|\hat{\bfm}_i - \bfm_i\|_{2}^{2} + \|\hbftau_i - \nhbftau_i\|_{2}^2}.
\end{equation}
As a result, $(\Delta_{\pi i}, \Delta_{y i})$ are deterministic.

We start by a lemma showing that $\stheta, \swy, \sww, \sw, \sy$ concentrate around their means. For notational convenience, we let $\Var(Z)$ denote $\E\|Z - \E[Z]\|_{2}^2$ for a random vector $Z$.
\begin{lemma}\label{lem:mean_var}
  Under Assumptions \ref{as:overlap_general} and \ref{as:regularity_general},
  \[|\E[\stheta]|  +|\E[\swy]| + |\E[\sww]| + \|\E[\sw]\|_{2} + \|\E[\sy]\|_{2} = O(1),\]
  and
  \[\Var(\stheta) + \Var(\swy) + \Var(\sww) + \Var(\sw) + \Var(\sy) = O(n^{-\q}).\]
  As a consequence,
  \[\big|\stheta - \E[\stheta]\big| + \big|\swy - \E[\swy]\big| + \big|\sww - \E[\sww]\big| + \big\|\sw - \E[\sw]\big\|_{2} + \big\|\sy - \E[\sy]\big\|_{2} = O_{\P}\lb n^{-\q/2}\rb.\]
\end{lemma}
\begin{proof}
  By Assumption \ref{as:overlap_general}, $\Theta_i\le 1 / c_{\pi}$ almost surely. Moreover, $\|\W_i\|_{2}\le \sqrt{T}$ since $W_{it}\in \{0, 1\}$. Thus,
  \[\|\sw\|_{2} \le \frac{\sqrt{T}}{c_{\pi}}, \quad |\sww|\le \frac{T}{c_{\pi}}, \quad |\stheta|\le \frac{1}{c_{\pi}} \Longrightarrow \E\|\sw\|_{2} + \E|\sww| + \E|\stheta| = O(1).\]
Next, we derive bounds for $(\E[\swy])^{2}$ and $\|\E[\sy]\|_{2}^{2}$ separately. For $(\E[\swy])^2$,
  \begin{align*}
    (\E[\swy])^2
    & \le \lb\frac{1}{n}\sum_{i=1}^{n}\E[\Theta_{i}\W_i^\tran \J \td{\Y}_{i}]\rb^2\le \frac{1}{n}\sum_{i=1}^{n}\E[\Theta_{i}\W_i^\tran \J \td{\Y}_{i}]^{2}\\
    & \le \frac{1}{nc_{\pi}^{2}}\sum_{i=1}^{n}\E[\W_i^\tran \J \td{\Y}_{i}]^{2}\le \frac{T}{nc_{\pi}^{2}}\sum_{i=1}^{n}\E\|\td{\Y}_{i}\|_{2}^{2}\\
    & \le \frac{T}{nc_{\pi}^{2}}\sum_{i=1}^{n}\left\{\E\|\td{\Y}_{i}(0)\|_{2}^{2} + \E\|\td{\Y}_{i}(1)\|_{2}^{2}\right\}\\
    & = O(1),
  \end{align*}
  where the last step follows from the Assumption \ref{as:regularity_general}. For $\|\E[\sy]\|_{2}^{2}$,
  \begin{align*}
    \|\E[\sy]\|_{2}^{2} & \le \frac{1}{n}\sum_{i=1}^{n}(\E[\Theta_{i}\J\td{\Y}_{i}])^2\le \frac{1}{nc_{\pi}^2}\sum_{i=1}^{n}\|\td{\Y}_{i}\|_{2}^{2}\\
                        & \le \frac{1}{nc_{\pi}^2}\sum_{i=1}^{n}\left\{\E\|\td{\Y}_{i}(0)\|_{2}^{2} + \E\|\td{\Y}_{i}(1)\|_{2}^{2}\right\}\\
    & = O(1),
  \end{align*}
  where the last step follows from the Assumption \ref{as:regularity_general}. Putting the pieces together, the bound on the sum of expectations is proved.

  ~\\
  Next, we turn to the bound on the variances. By Lemma \ref{lem:rho_mixing},
  \[\Var(\stheta)\le \frac{1}{n^2}\sum_{i=1}^{n}\Var(\Theta_{i})\rho_i \le \frac{1}{n^2c_{\pi}^2}\sum_{i=1}^{n}\rho_i.\]
  The Assumption \ref{as:overlap_general} implies that
  \[\frac{1}{n^2}\sum_{i=1}^{n}\rho_{i} = O(n^{-\q}).\]
  Therefore, $\Var(\stheta) = O(n^{-q})$. For $\sww$,
  \begin{align*}
    \Var(\sww)
    & \le \frac{1}{n^2}\sum_{i=1}^{n}\Var(\Theta_i \W_i^{\tran}\J \W_i)\rho_i\le \frac{1}{n^2}\sum_{i=1}^{n}\E(\Theta_i \W_i^{\tran}\J \W_i)^2\rho_i\\
    & \le \frac{1}{n^2c_{\pi}^2}\sum_{i=1}^{n}\E\|\W_i\|_{2}^{2}\rho_i \le \frac{T}{n^2c_{\pi}^2}\sum_{i=1}^{n}\rho_i = O(n^{-\q}),
  \end{align*}
  where the last equality uses the fact that $\|\W_i\|_{2}\le \sqrt{T}$. For $\swy$, 
  \begin{align*}
    \Var(\swy)
    & \le \frac{1}{n^2}\sum_{i=1}^{n}\Var(\Theta_i \W_i^{\tran}\J \td{\Y}_i)\rho_i\le \frac{1}{n^2}\sum_{i=1}^{n}\E(\Theta_i \W_i^{\tran}\J \td{\Y}_i)^2\rho_i\\
    & \stackrel{(i)}{\le} \frac{1}{n^2c_{\pi}^2}\sum_{i=1}^{n}\E\left[\|\W_i\|_{2}^{2}\cdot \|\td{\Y}_i\|_{2}^2\right]\rho_i\\
    & \stackrel{(ii)}{\le}\frac{T}{n^2c_{\pi}^2}\sum_{i=1}^{n}\lb\E\|\td{\Y}_i(1)\|_{2}^2 + \E\|\td{\Y}_i(0)\|_{2}^2\rb\rho_i\\
    & \stackrel{(iii)}{=} O(n^{-\q}),
  \end{align*}
  where (i) follows from the Cauchy-Schwarz inequality and that $\|\J\|_{\mathrm{op}} = 1$, (ii) is obtained from the fact that $\|\W_i\|_{2}^2\le T$ and $\td{\Y}_{i}\in \{\td{\Y}_i(1), \td{\Y}_i(0)\}$, and (iii) follows from the Assumption \ref{as:regularity_general}.

  ~\\
  For $\sw$, recall that $\Var(\sw)$ is the sum of the variance of each coordinate of $\sw$. By Lemma \ref{lem:rho_mixing},
  \begin{align*}
    \Var(\sw)
    & \le \frac{1}{n^2}\sum_{i=1}^{n}\Var(\Theta_{i}\J \W_{i})\rho_{i}\le \frac{1}{n^2}\sum_{i=1}^{n}\E\|\Theta_{i}\J \W_{i}\|_{2}^2\rho_{i}\\
    & \le \frac{1}{n^2c_{\pi}^2}\sum_{i=1}^{n}\E\|\W_i\|_{2}^2\rho_{i} \le \frac{T}{n^2c_{\pi}^2}\sum_{i=1}^{n}\rho_{i} = O(n^{-\q}).
  \end{align*}
  For $\sy$, analogues to inequalities (i) - (iii) for $\swy$, we obtain that
  \begin{align*}
    \Var(\sy)
    & \le \frac{1}{n^2}\sum_{i=1}^{n}\Var(\Theta_{i}\J \td{\Y}_{i})\rho_{i}\le \frac{1}{n^2}\sum_{i=1}^{n}\E\|\Theta_{i}\J \td{\Y}_{i}\|_{2}^2\rho_{i}\\
    & \le \frac{1}{n^2c_{\pi}^2}\sum_{i=1}^{n}\lb \E[\|\td{\Y}_i(1)\|_{2}^2] + \E[\|\td{\Y}_i(0)\|_{2}^2]\rb\rho_{i} = O(n^{-\q}),
  \end{align*}
  where the last step follows from the Assumption \ref{as:regularity_general}.

  ~\\
  Finally, by Markov's inequality,
  \begin{align*}
    &\big|\stheta - \E[\stheta]\big| + \big|\swy - \E[\swy]\big| + \big|\sww - \E[\sww]\big| + \big\|\sw - \E[\sw]\big\|_{2} + \big\|\sy - \E[\sy]\big\|_{2}\\
    =& O_\P\lb\sqrt{\Var(\stheta) + \Var(\swy) + \Var(\sww) + \Var(\sw) + \Var(\sy)}\rb = O_\P(n^{-\q/2}).
  \end{align*}
\end{proof}

The following lemma shows that the denominator of $\hat{\tau}$ is bounded away from $0$.
\begin{lemma}\label{lem:denom}
  Under Assumptions \ref{as:regularity_general}, regardless of the dependence between $(\hat{\bfpi}_i, \hat{\bfm}_{i}, \hbftau_i)$ and the data,
  \[\hatdenom\ge c_{\hatdenom}^{2}\lb\frac{1}{n}\sum_{i=1}^{n}I(\W_i = \w_1)\rb\lb\frac{1}{n}\sum_{i=1}^{n}I(\W_i = \w_2)\rb,\]
  for some constant $c_{\hatdenom}$ that only depends on $\bfPi$. As a result, $\hatdenom\ge 0$ almost surely. If Assumption \ref{as:overlap_general} also holds, \footnote{A more rigorous version of the second statement is $\max\{c_{\hatdenom}^2 c_{\pi}^2 - \hatdenom, 0\} = o_\P(1)$}
  \[\E[\hatdenom] \ge c_{\hatdenom}^2 \lb c_{\pi}^2 - \frac{1}{n^2}\sum_{i=1}^{n}\rho_{i}\rb, \quad \hatdenom\ge c_{\hatdenom}^2 c_{\pi}^2 - o_\P(1).\]
\end{lemma}
\begin{proof}
  By definition,
  \begin{align*}
    \hatdenom
    &= \lb\frac{1}{n}\sum_{i=1}^{n}\Theta_{i}\td{\W}_i^\tran\J\td{\W}_i\rb\lb\frac{1}{n}\sum_{i=1}^{n}\Theta_{i}\rb - \left\|\frac{1}{n}\sum_{i=1}^{n}\Theta_{i}\J\td{\W}_i\right\|_{2}^{2}\\
    & = \frac{1}{n^2}\sum_{i,j=1}^{n}\Theta_{i}\Theta_{j}\lb\td{\W}_{i}^\tran\J\td{\W}_{i} + \td{\W}_{j}^\tran\J \td{\W}_{j} - 2\td{\W}_{i}^\tran\J\td{\W}_{j}\rb\\
    & = \frac{1}{n^2}\sum_{i,j=1}^{n}\Theta_{i}\Theta_{j}\|\J(\W_{i} - \W_{j})\|_{2}^{2}.
  \end{align*}
  Let $\w_1, \w_2$ be two distinct elements from $\S^{*}$ with $\w_1 \not\in \{\zero_{T}, \one_{T}\}$ and
  \begin{equation}
    \label{eq:choice_w1w2}
    \frac{1}{n}\sum_{i=1}^{n}\bfpi_{i}(\w_k) > c_{\pi}, \quad k\in\{1, 2\}.
  \end{equation}
This is enabled by Assumption \ref{as:overlap_general}. Note that $\J(\w_1 - \w_2) = 0$ iff $\w_1 - \w_2 = a\one_{T}$ for some $a\in \R$, which is impossible since $\w_1\not\in \{\zero_{T}, \one_{T}\}$ and all entries of $\w_1$ and $\w_2$ are binary. In addition, since $\bfPi$ has support $\S^{*}$, $\bfPi(\w_1), \bfPi(\w_2) > 0$. Let
\[c_{\hatdenom} = \min\{\bfPi(\w_1), \bfPi(\w_2)\}\|\J(\w_{1} - \w_{2})\|_{2} > 0.\]
Then
  \begin{align}
    \hatdenom
    &\ge \frac{c_{\hatdenom}^2}{n^2}\sum_{i,j=1}^{n}\frac{1}{\hat{\bfpi}_i(\W_i)\hat{\bfpi}_j(\W_j)}I(\W_i = \w_1, \W_j = \w_2)\nonumber\\
    &\ge \frac{c_{\hatdenom}^2}{n^2}\sum_{i,j=1}^{n}I(\W_i = \w_1, \W_j = \w_2)\nonumber\\
    &= c_{\hatdenom}^{2}\lb\frac{1}{n}\sum_{i=1}^{n}I(\W_i = \w_1)\rb\lb\frac{1}{n}\sum_{i=1}^{n}I(\W_i = \w_2)\rb,\nonumber
  \end{align}
  where the second inequality follows from the fact that $\hat{\bfpi}_i(\w)\le 1$. By \eqref{eq:choice_w1w2},
  \begin{equation*}
    \E\left[\frac{1}{n}\sum_{i=1}^{n}I(\W_i = \w_k)\right] = \frac{1}{n}\sum_{i=1}^{n}\bfpi_{i}(\w_k) > c_{\pi}, \quad k\in\{1, 2\}.
  \end{equation*}
  Furthermore, by Lemma \ref{lem:rho_mixing}, 
  \begin{align}
    &\Bigg|\Cov\left[\frac{1}{n}\sum_{i=1}^{n}I(\W_i = \w_1), \frac{1}{n}\sum_{i=1}^{n}I(\W_i = \w_2) \right]\Bigg|\nonumber\\
    & = \frac{1}{n^2}\bigg|\sum_{i,j=1}^{n}\Cov(I(\W_i = \w_1), I(\W_j = \w_2))\bigg|\nonumber\\
    & \le \frac{1}{n^2}\sum_{i,j=1}^{n}\big|\Cov(I(\W_i = \w_1), I(\W_j = \w_2))\big|\nonumber\\
    & \le \frac{1}{n^2}\sum_{i,j=1}^{n}\rho_{ij}\sqrt{\Var(I(\W_i = \w_1))\Var(I(\W_j = \w_2))}\nonumber\\
    & \le \frac{1}{n^2}\sum_{i,j=1}^{n}\rho_{ij} = \frac{1}{n^2}\sum_{i=1}^{n}\rho_{i}.\nonumber
  \end{align}
 Putting pieces together, we obtain that
  \begin{align}
    \E[\hatdenom]
    &\ge c_{\hatdenom}^2 \E \left[\lb\frac{1}{n}\sum_{i=1}^{n}I(\W_i = \w_1)\rb\lb\frac{1}{n}\sum_{i=1}^{n}I(\W_i = \w_2)\rb\right]\nonumber\\
    & = c_{\hatdenom}^2\Bigg\{\E \left[\frac{1}{n}\sum_{i=1}^{n}I(\W_i = \w_1)\right]\E\left[\frac{1}{n}\sum_{i=1}^{n}I(\W_i = \w_2)\right]\nonumber\\
    & \quad\quad + \Cov\left[\frac{1}{n}\sum_{i=1}^{n}I(\W_i = \w_1), \frac{1}{n}\sum_{i=1}^{n}I(\W_i = \w_2)\right]\Bigg\}\nonumber\\
    & \ge c_{\hatdenom}^2 \lb c_{\pi}^2 - \frac{1}{n^2}\sum_{i=1}^{n}\rho_i\rb.\nonumber
  \end{align}
  On the other hand, by Lemma \ref{lem:rho_mixing}, for $k\in \{1, 2\}$, 
  \[\Var\lb\frac{1}{n}\sum_{i=1}^{n}I(\W_i = \w_k)\rb\le \frac{1}{n^2}\sum_{i=1}^{n}\rho_i = O(n^{-\q}) = o(1).\]
  By Markov's inequality, for $k \in \{1, 2\}$,
  \[\frac{1}{n}\sum_{i=1}^{n}I(\W_i = \w_k) = \frac{1}{n}\sum_{i=1}^{n}\P(\W_i = \w_1) - o_\P\lb 1\rb \ge c_{\pi} - o_\P(1).\]
  Therefore,
  \[\hatdenom \ge c_{\hatdenom}^2 (c_{\pi} - o_\P(1))(c_{\pi} - o_\P(1))\ge c_{\hatdenom}^2 c_{\pi}^2 - o_\P(1).\]
\end{proof}

Based on Lemma \ref{lem:mean_var} and \ref{lem:denom}, we can derive an asymptotic linear expansion for the RIPW estimator.
\begin{theorem}\label{thm:asymptotic_expansion_simple}
  Under Assumptions \ref{as:overlap_general} and \ref{as:regularity_general},
  \[\hatdenom(\hat{\tau} - \tau^{*}) = \numer + \frac{1}{n}\sum_{i=1}^{n}(\V_i - \E[\V_i]) + O_\P\lb n^{-\q}\rb,\]
  where
  \[\numer = \frac{1}{2n}\sum_{i=1}^{n}\E[\V_i] = \E[\swy]\E[\stheta] - \E[\sw]^\tran\E[\sy] - \tau^{*}\lb\E[\sww]\E[\stheta] - \E[\sw]^\tran\E[\sw]\rb.\]
  Furthermore,
  \[\hat{\tau} - \tau^{*} = O_\P(|\numer|) + O_\P(n^{-\q / 2}).\]
\end{theorem}
\begin{proof}
  Note that
  \[\hatdenom(\hat{\tau} - \tau^{*}) = \hatnumer - \tau^{*}\hatdenom.\]
  By Lemma \ref{lem:mean_var},
  \begin{align*}
    &\big|(\swy - \E[\swy])(\stheta - \E[\stheta])\big| + \big|(\sw - \E[\sw])^\tran(\sy - \E[\sy])\big|\\
    \le & \frac{1}{2}\bigg\{(\swy - \E[\swy])^2 + (\stheta - \E[\stheta])^2 + \|\sw - \E[\sw]\|_{2}^{2} + \|\sy - \E[\sy]\|_{2}^{2}\bigg\}\\
    = & O_\P\lb \Var(\swy) + \Var(\stheta) + \Var(\sw) + \Var(\sy)\rb = O_\P(n^{-\q}).
  \end{align*}
  Let
  \[\V_{i1} = \Theta_i\left\{\E[\swy] - \E[\sy]^\tran\J \W_i  + \E[\stheta]\W_i^\tran \J \td{\Y}_i - \E[\sw]^\tran\J \td{\Y}_{i} \right\}.\]
  Then, 
  \[\hatnumer = \E[\swy]\E[\stheta] - \E[\sw]^\tran\E[\sy] + \frac{1}{n}\sum_{i=1}^{n}(\V_{i1} - \E[\V_{i1}]) + O_\P(n^{-\q}),\]
  Similarly,
  \[\hatdenom = \E[\sww]\E[\stheta] - \E[\sw]^\tran\E[\sw] + \frac{1}{n}\sum_{i=1}^{n}(\V_{i2} - \E[\V_{i2}]) + O_\P(n^{-\q}),\]
  where
  \[\V_{i2} = \Theta_i\left\{\E[\sww] - \E[\sw]^\tran\J \W_i  + \E[\stheta]\W_i^\tran \J \W_i - \E[\sw]^\tran\J \W_{i} \right\}.\]
  Since $\V_i = \V_{i1} - \tau^{*}\V_{i2}$, 
  \[\hatdenom(\hat{\tau} - \tau^{*}) = \hatnumer - \tau^{*}\hatdenom = \numer + \frac{1}{n}\sum_{i=1}^{n}(\V_i - \E[\V_i]) + O_\P(n^{-\q}).\]
  This proves the first statement.

  ~\\
  Next, we prove the second statement on $\hat{\tau} - \tau^{*}$. By Lemma \ref{lem:denom}, $1 / \hatdenom = O_\P(1)$. It is left to show that
  \[\frac{1}{n}\sum_{i=1}^{n}(\V_i - \E[\V_i]) = O_\P(n^{-\q / 2}).\]
  Applying the inequality that $\Var(Z_1 + Z_2) = 2\Var(Z_1) + 2\Var(Z_2) - \Var(Z_1 - Z_2)\le 2(\Var(Z_1) + \Var(Z_2))$, we obtain that
  \begin{align*}
    &\frac{1}{4}\Var(\V_{i1})\\
    \le & \Var\lb \Theta_{i}\E[\swy]\rb + \Var\lb \Theta_{i}\E[\stheta]\W_i^\tran \J \td{\Y}_i\rb + \Var\lb\Theta_{i}\E[\sw]^\tran\J\td{\Y}_{i} \rb + \Var\lb \Theta_{i}\E[\sy]^\tran\J \W_i\rb\\
    \le & \E\lb \Theta_{i}\E[\swy]\rb^2 + \E\lb \Theta_{i}\E[\stheta]\W_i^\tran \J \td{\Y}_i\rb^2 + \E\lb\Theta_{i}\E[\sw]^\tran\J\td{\Y}_{i} \rb^2 + \E\lb \Theta_{i}\E[\sy]^\tran\J \W_i\rb^2\\
    \stackrel{(i)}{\le} & \frac{1}{c_{\pi}^2}\left\{ (\E[\swy])^2 + (\E[\stheta])^2\E(\W_i^\tran \J \td{\Y}_i)^2 + \E\lb\E[\sw]^\tran\J\td{\Y}_{i} \rb^2 + \E\lb \E[\sy]^\tran\J \W_i\rb^2\right\}\\
    \stackrel{(ii)}{\le} & \frac{1}{c_{\pi}^2}\left\{ (\E[\swy])^2 + (\E[\stheta])^2\E\|\W_i\|_{2}^{2}\E\|\td{\Y}_i\|_{2}^2 + \|\E[\sw]\|_{2}^{2}\E\|\td{\Y}_{i} \|_{2}^{2} + \|\E[\sy]\|_{2}^{2}\E\|\W_i\|_{2}^2\right\}\\
    \stackrel{(iii)}{\le} & \frac{1}{c_{\pi}^2}\left\{ (\E[\swy])^2 + T(\E[\stheta])^2\E\|\td{\Y}_i\|_{2}^2 + \|\E[\sw]\|_{2}^{2}\E\|\td{\Y}_{i} \|_{2}^{2} + T\|\E[\sy]\|_{2}^{2}\right\},
  \end{align*}
  where (i) follows from the Assumption \ref{as:overlap_general} that $\Theta_{i}\le 1 / c_{\pi}$ almost surely, (ii) follows from the Cauchy-Schwarz inequality and the fact that $\|\J\|_{\mathrm{op}} = 1$, and (iii) follows from the fact that $\|\W_i\|_{2}^{2}\le T$. By Lemma \ref{lem:mean_var}, we obtain that for all $i\in [n]$,
  \begin{equation}
    \label{eq:var_Vi}
    \Var(\V_{i1})\le C_1\lb 1 + \E\|\td{\Y}_{i}\|_{2}^{2}\rb\le C_1\lb 1 + \E\|\td{\Y}_{i}(0)\|_{2}^{2} + \E\|\td{\Y}_{i}(1)\|_{2}^{2}\rb,
  \end{equation}
  for some constant $C_1$ that only depends on $c_{\pi}$ and $T$. Similarly, we have that $\Var(\V_{i2})\le C_2$ for some constant $C_2$ that only depends on $c_{\pi}$ and $T$. By Assumption \ref{as:regularity_general}, 
  \[\tau^{*} = \sum_{t=1}^{T}\xi_{t}\left\{\frac{1}{n}\sum_{i=1}^{n}\lb\E[\td{Y}_{it}(1)] - \E[\td{Y}_{it}(0)]\rb\right\} = O(1).\]
  Therefore, 
  \[\Var(\V_i)\le 2\Var(\V_{i1}) + 2(\tau^{*})^2\Var(\V_{i2})\le C\lb 1 + \E\|\td{\Y}_{i}(0)\|_{2}^{2} + \E\|\td{\Y}_{i}(1)\|_{2}^{2}\rb.\]
  for some constant $C$ that only depends on $c_{\pi}$ and $T$. Since $\V_i$ is a function of $(\Y_i(1), \Y_i(0), \X_i)$, by Lemma \ref{lem:rho_mixing} and Assumption \ref{as:regularity_general},
  \[\Var\lb\frac{1}{n}\sum_{i=1}^{n}\V_{i}\rb\le \frac{1}{n^2}\sum_{i=1}^{n}\Var(\V_{i})\rho_{i} = O(n^{-q}).\]  
  By Chebyshev's inequality,
  \[\frac{1}{n}\sum_{i=1}^{n}(\V_{i} - \E[\V_i]) = O_\P(n^{-\q/2}).\]
  The proof is then completed.
\end{proof}

\subsubsection{\DATEeq ~and consistency}\label{subsubapp:consistency}
Theorem \ref{thm:asymptotic_expansion_simple} shows that the asymptotic limit of $\hatdenom(\hat{\tau} - \tau^*)$ is $\numer$. For consistency, it remains to prove that $\numer = o(1)$. We start by proving that the asymptotic bias is zero when either the treatment or the outcome model is perfectly estimated.
\begin{lemma}\label{lem:numer_weak_deterministic}
  Under Assumptions \ref{as:latent_ign_app}, \ref{as:overlap_general}, and \ref{as:regularity_general}, $\numer = 0$, if either (1) $\Delta_{y i} = 0$ for all $i \in [n]$, or (2) $\Delta_{\pi i} = 0$ for all $i \in [n]$, and $\bfPi$ satisfies the \DATEeq ~\eqref{eq:DATE_equation}. 
\end{lemma}
\begin{proof}
  Without loss of generality, we assume that $\tau^{*} = 0$; otherwise, we replace $Y_{it}(1)$ by $Y_{it}(1) - \tau^{*}$ and the resulting $\hat{\tau}$ becomes $\hat{\tau} - \tau^{*}$. Then
  \begin{align*}
    \numer & = \E[\swy]\E[\stheta] - \E[\sw]^\tran\E[\sy].
  \end{align*}
  It remains to prove that $\numer = 0$. Since $(\hat{\bfpi}_i, \hat{\bfm}_{i}, \hbftau_i)$ are deterministic, by Assumption \ref{as:latent_ign_app} and \eqref{eq:tdbftau_i}, 
  \begin{align*}
    \E[\swy]
    &= \frac{1}{n}\sum_{i=1}^{n}\E[\Theta_i \W_i^\tran\J\td{\Y}_i] = \frac{1}{n}\sum_{i=1}^{n}\E[\Theta_i \W_i^\tran\J\{\td{\Y}_i(0) + \diag(\W_i)(\td{\Y}_i(1) - \td{\Y}_i(0))\}]\\
    & = \frac{1}{n}\sum_{i=1}^{n}\E[\Theta_i \J\W_i]^\tran\E[\td{\Y}_i(0)] + \frac{1}{n}\sum_{i=1}^{n}\E[\Theta_i \W_i^\tran\J\diag(\W_i)]\td{\bftau}_i.
  \end{align*}
  Similarly,
  \begin{align*}
    \E[\sy]
    & = \frac{1}{n}\sum_{i=1}^{n}\E[\Theta_i]\J\E[\td{\Y}_i(0)] + \frac{1}{n}\sum_{i=1}^{n}\E[\Theta_i \diag(\W_i)]\td{\bftau}_i.
  \end{align*}
  As a result,
  \begin{align}
    \numer
    & = \frac{1}{n}\sum_{i=1}^{n}\left\{\E[\Theta_i\J\W_i]\E[\stheta] - \E[\Theta_i]\E[\sw]\right\}^\tran\E[\td{\Y}_i(0)]\nonumber\\
    & \quad + \frac{1}{n}\sum_{i=1}^{n}\left\{\E[\Theta_i \W_i^\tran\J\diag(\W_i)]\E[\stheta] - \E[\sw]^\tran\E[\Theta_i\diag(\W_i)]\right\}\td{\bftau}_i.\label{eq:numer_bias}
  \end{align}

  If $\Delta_{yi} = 0$, $\hat{\bfm}_i = \bfm_i$ and $\hbftau_i = \nhbftau_i$. Since we have assumed $\tau^* = 0$, $\td{\bftau}_i = \zero_{T}$. By \eqref{eq:Yi0_WLOG}, $\E[\td{\Y}_i(0)] = \zero_{T}$. It is then obvious from \eqref{eq:numer_bias} that $\numer = 0$.

  If $\Delta_{\pi i} = 0$, $\hat{\bfpi}_i = \bfpi_i$ and thus for any function $f(\cdot)$, 
  \begin{equation}\label{eq:change_of_measure}
  \E[\Theta_i f(\W_i)] = \sum_{\w\in \{0, 1\}^{T}} \frac{\bfPi(\w)}{\bfpi_i(\w)}f(\w)\bfpi_i(\w) = \E_{\W\sim \bfPi}[f(\W)].
  \end{equation}
  As a result,
  \[\E[\Theta_i \J\W_i] = \EPi[\J\W] = \E[\sw], \quad \E[\Theta_i] = 1 = \E[\stheta],\]
  and 
  \[\E[\Theta_i \W_i^\tran\J\diag(\W_i)] = \EPi[\W\J\diag(\W)], \quad \E[\Theta_i\diag(\W_i)] = \EPi[\diag(\W)].\]
  Then
  \[\E[\Theta_i\J\W_i]\E[\stheta] - \E[\Theta_i]\E[\sw] = \EPi[\J \W] - \EPi[\J \W] = 0,\]
  and by \DATEeq,
  \begin{align*}
    &\E[\Theta_i \W_i^\tran\J\diag(\W_i)]\E[\stheta] - \E[\sw]^\tran\E[\Theta_i\diag(\W_i)]\\
    &= \EPi[(\W - \EPi[\W])^\tran \J \diag(\W)]\\
    & = \EPi[(\W - \EPi[\W])^\tran \J \W]\xi^\tran.
  \end{align*}
  By \eqref{eq:numer_bias} and \eqref{eq:tau*_tdtau},
  \begin{align*}
    \numer
    & = \frac{1}{n}\sum_{i=1}^{n}\EPi[(\W - \EPi[\W])^\tran \J \W]\xi^\tran\td{\bftau}_i\\
    & = \EPi[(\W - \EPi[\W])^\tran \J \W]\lb\frac{1}{n}\sum_{i=1}^{n}\xi^\tran\td{\bftau}_i\rb\\
    & = \EPi[(\W - \EPi[\W])^\tran \J \W]\tau^* = 0.
  \end{align*}
\end{proof}

Next, we prove a general bound for the asymptotic bias $\numer$ as a function of $(\Delta_{yi}, \Delta_{\pi i})_{i=1}^{n}$.
\begin{theorem}\label{thm:numer_strong_deterministic}
  Let $\bfPi$ be an solution of the \DATEeq ~\eqref{eq:DATE_equation}. Under Assumptions \ref{as:latent_ign_app}, \ref{as:overlap_general}, and \ref{as:regularity_general},
  \[|\numer| = O\lb\bar{\Delta}_\pi\bar{\Delta}_y\rb.\]
\end{theorem}
\begin{proof}
 As in the proof of Lemma \ref{lem:numer_weak_deterministic}, we assume that $\tau^{*} = 0$. Let
 \[\Theta_i^* = \frac{\bfPi(\W_i)}{\bfpi_i(\W_i)}, \quad \td{\Y}_i^{*} = \Y_i - \bfm_i - \diag(\W_i)\nhbftau_i.\]
 Further, let $\stheta^*$ and $\sw^*$ be the counterpart of $\stheta$ and $\sw$ with $(\Theta_i, \td{\Y}_i)$ replaced by $(\Theta_i^*, \td{\Y}_i^{*})$. For any function $f: \{0, 1\}^T \mapsto \R$ such that $\E[f^2(\W_i)]\le C_1$ for some constant $C_1 > 0$, by Cauchy-Schwarz inequality,
 \begin{align}
   &\E[\Theta_i f(\W_i) - \Theta_i^*f(\W_i)]
     = \E[(\Theta_i - \Theta_i^*)f(\W_i)]\le \sqrt{C_1}\sqrt{\E\lb\Theta_i - \Theta_i^*\rb^2}\nonumber\\
     & = \sqrt{C_1}\sqrt{\E\left[\frac{\bfPi(\W_i)^2}{\hat{\bfpi}_i(\W_i)^2\bfpi_i(\W_i)^2}\lb\hat{\bfpi}_i(\W_i) - \bfpi_i(\W_i)\rb^2\right]}\le \frac{\sqrt{C_1}}{c_{\pi}^2}\Delta_{\pi i}\label{eq:Thetai-Thetai*}.
 \end{align}
 Thus, there exists a constant $C_2$ that only depends on $c_{\pi}$ and $T$ such that
 \begin{align*}
   &|\E[\Theta_i] - \E[\Theta_i^{*}]| + \|\E[\Theta_i \J \W_i] - \E[\Theta_i^* \J \W_i]\|_{2} + \|\E[\Theta_i \W_i^\tran \J \diag(\W_i)] - \E[\Theta_i^* \W_i^\tran \J \diag(\W_i)]\|_{2}\\
   & \quad + \|\E[\Theta_i \diag(\W_i)] - \E[\Theta_i^* \diag(\W_i)]\|_{\op}\le C_2\Delta_{\pi i}.
 \end{align*}
 By triangle inequality and Cauchy-Schwarz inequality, we also have
 \[|\E[\stheta] - \E[\stheta^*]| + \|\E[\sw] - \E[\sw^*]\|_{2}\le \frac{C_2}{n}\sum_{i=1}^{n}\Delta_{\pi i} \le C_2\bar{\Delta}_{\pi}.\]
 On the other hand, by Lemma \ref{lem:mean_var}, there exists a constant $C_3$ that only depends on $c_{\pi}$ and $T$,
 \[|\E[\stheta]| + \|\E[\sw]\|_{2}\le C_3.\]
 Without loss of generality, we assume that
 \[C_3 \ge 1 + \sqrt{T} \ge 1 + \|\EPi[\J \W]\|_{2} = \E[\Theta_i^*] + \|\E[\Theta_i^*\J\W_i]\|_{2}.\]
 Putting pieces together,
 \begin{align*}
   &\bigg|\E[\Theta_i\J\W_i]\E[\stheta] - \E[\Theta_i]\E[\sw] - \lb \E[\Theta_i^*\J\W_i]\E[\stheta^*] - \E[\Theta_i^*]\E[\sw^*]\rb\bigg|\\
   & \le \big|\E[\Theta_i\J\W_i] - \E[\Theta_i^*\J\W_i]\big|\cdot \E[\stheta] + \big|\E[\Theta_i] - \E[\Theta_i^*]\big|\cdot \|\E[\sw]\|_{2} \\
   & \quad + \big|\E[\stheta] - \E[\stheta^*]\big|\cdot \|\E[\Theta_i^*\J\W_i]\|_{2} + \big\|\E[\sw] - \E[\sw^*]\big\|\cdot \E[\Theta_i^*]\\
   & \le 2C_3C_2(\Delta_{\pi i} + \bar{\Delta}_{\pi}).
 \end{align*}
 Similarly,
 \begin{align*}
   &\bigg|\E[\Theta_i \W_i^\tran\J\diag(\W_i)]\E[\stheta] - \E[\sw]^\tran\E[\Theta_i\diag(\W_i)]\\
   & - \lb \E[\Theta_i^* \W_i^\tran\J\diag(\W_i)]\E[\stheta^*] - \E[\sw^*]^\tran\E[\Theta_i^*\diag(\W_i)]\rb\bigg|\\
   & \le 2C_3C_2(\Delta_{\pi i} + \bar{\Delta}_{\pi}).
 \end{align*}
 Let
 \begin{align*}
   \numer'
   & = \frac{1}{n}\sum_{i=1}^{n}\left\{\E[\Theta_i^*\J\W_i]\E[\stheta^*] - \E[\Theta_i^*]\E[\sw^*]\right\}^\tran\E[\td{\Y}_i(0)]\nonumber\\
   & \quad + \frac{1}{n}\sum_{i=1}^{n}\left\{\E[\Theta_i^* \W_i^\tran\J\diag(\W_i)]\E[\stheta^*] - \E[\sw^*]^\tran\E[\Theta_i^*\diag(\W_i)]\right\}\td{\bftau}_i.
 \end{align*}
 Using the same arguments as in the proof of Lemma \ref{lem:numer_weak_deterministic},
 \[\E[\Theta_i^*\J\W_i]\E[\stheta^*] - \E[\Theta_i^*]\E[\sw^*] = 0,\]
 and
 \[\E[\Theta_i^* \W_i^\tran\J\diag(\W_i)]\E[\stheta^*] - \E[\sw^*]^\tran\E[\Theta_i^*\diag(\W_i)] = \EPi[(\W - \EPi[\W])^\tran \J \W]\xi^\tran.\]
 Then
 \[\numer' = \frac{1}{n}\sum_{i=1}^{n}\EPi[(\W - \EPi[\W])^\tran \J \W]\xi^\tran \td{\bftau}_i = \EPi[(\W - \EPi[\W])^\tran \J \W]\tau^* = 0.\]
 This entails that
 \begin{align*}
   |\numer| & = |\numer - \numer'|\le \frac{2C_3C_2}{n}\sum_{i=1}^{n}(\Delta_{\pi i} + \bar{\Delta}_{\pi})(\|\E[\td{\Y}_i(0)]\|_{2} + \|\td{\bftau}_i\|_{2}).
 \end{align*}
 By \eqref{eq:tdbftau_i}, \eqref{eq:tau*_tdtau}, and \eqref{eq:Yi0_WLOG}, 
 \[\|\E[\td{\Y}_i(0)]\|_{2} + \|\td{\bftau}_i\|_{2} = \|\hat{\bfm}_i - \bfm_i\|_{2} + \|\hbftau_i - \nhbftau_i\|_{2}\le 2\Delta_{yi}.\]
 Since $(1/n)\sum_{i=1}^{n}\Delta_{yi}\le \sqrt{(1/n)\sum_{i=1}^{n}\Delta_{yi}^2}$, 
 \begin{align*}
   |\numer|&\le \frac{4C_3 C_2}{n}\sum_{i=1}^{n}(\Delta_{\pi i} + \bar{\Delta}_{\pi})\Delta_{y i} = 4C_3C_2 \bar{\Delta}_{\pi}\bar{\Delta}_{y}.
 \end{align*}
 The proof is then completed.
\end{proof}

\subsubsection{Asymptotic inference under independence}
Theorem \ref{thm:asymptotic_expansion_simple} and Theorem \ref{thm:numer_strong_deterministic} imply the following properties of RIPW estimators. 

\begin{theorem}\label{thm:doubly_robust_expansion}
  Let $\bfPi$ be an solution of the \DATEeq ~\eqref{eq:DATE_equation}. Under Assumptions \ref{as:latent_ign_app}, \ref{as:overlap_general}, and \ref{as:regularity_general},
  \[\hat{\tau} - \tau^{*} = o_\P(1), \quad \text{if }\,\, \bar{\Delta}_{\pi}\bar{\Delta}_{y} = o(1).\]
  If, further, $q > 1/2$ in Assumption \ref{as:regularity_general} and $\bar{\Delta}_{\pi}\bar{\Delta}_{y} = o(1 / \sqrt{n})$,
  \[\hatdenom\cdot \sqrt{n}(\hat{\tau} - \tau^{*}) = \frac{1}{\sqrt{n}}\sum_{i=1}^{n}(\V_i - \E[\V_i]) + o_\P(1).\]
\end{theorem}

Recalling \eqref{eq:delta_deterministic} that $(\Delta_{\pi i}, \Delta_{yi})$ are deterministic, $\bar{\Delta}_{\pi}\bar{\Delta}_{y} = \E[\bar{\Delta}_{\pi}\bar{\Delta}_{y}]$. Since Assumptions \ref{as:overlap_general} and \ref{as:regularity_general} generalize Assumptions \ref{as:limit_dep}-\ref{as:bounded_outcomes}, Theorem \ref{thm:doubly_robust_expansion} implies Theorem \ref{thm:bias} and \ref{thm:design_linear_expansion}. Similarly, Theorem \ref{thm:doubly_robust_expansion_maintext} is implied by Theorem \ref{thm:doubly_robust_expansion}.

Throughout the rest of the subsection, we focus on the special case where $\{(\Y_i(1), \Y_i(0), \X_i): i \in [n]\}$ are independent. In this case, Assumption \ref{as:regularity_general} holds with $q = 1 > 1 / 2$ and thus the asymptotically linear expansion in Theorem \ref{thm:doubly_robust_expansion} holds. To obtain the asymptotic normality and a consistent variance estimator, we modify Assumption \ref{as:regularity_general} as follows.
\begin{assumption}
  \label{as:independent_general}
  $\{(\Y_i(1), \Y_i(0), \X_i): i = 1,\ldots,n\}$ are independent (but not necessarily identically distributed), and there exists $\omega > 0$ such that
  \begin{equation*}
    \frac{1}{n}\sum_{i=1}^{n}\bigg\{ \E\|\td{\Y}_i(1)\|_{2}^{2 + \omega} + \E\|\td{\Y}_i(0)\|_{2}^{2 + \omega}\bigg\} = O(1).
  \end{equation*}
\end{assumption}
To derive the asymptotic normality of the RIPW estimator, we need the following assumption that prevents the variance from being too small.
\begin{assumption}\label{as:var_low}
  There exists $v_0 > 0$ such that
  \begin{equation*}
    \sigma^2\triangleq \frac{1}{n}\sum_{i=1}^{n}\Var(\V_i)\ge v_0.
  \end{equation*}
\end{assumption}

The following lemma shows the asymptotic normality of the term $\frac{1}{\sqrt{n}}\sum_{i=1}^{n}(\V_i - \E[\V_i])$.
\begin{lemma}\label{lem:CLT}
  Then under Assumptions \ref{as:overlap_general}, \ref{as:independent_general}, and \ref{as:var_low},
  \[d_{K}\lb\mathcal{L}\lb\frac{1}{\sqrt{n}\sigma}\sum_{i=1}^{n}(\V_i - \E[\V_i])\rb, N(0, 1)\rb \rightarrow 0,\]
  where $\mathcal{L}(\cdot)$ denotes the probability law, $d_{K}$ denotes the Kolmogorov-Smirnov distance (i.e., the $\ell_{\infty}$-norm of the difference of CDFs)
\end{lemma}
\begin{proof}
  Since $(\hat{\bfpi}_i, \hat{\bfm}_{i}, \hbftau_i)$ are deterministic, by Assumption \ref{as:independent_general}, $\{\V_i: i\in [n]\}$ are independent. Recalling the definition of $\V_i$, it is easy to see that Assumption \ref{as:independent_general} implies
  \begin{equation}
    \label{eq:Vi_moments}
    \frac{1}{n}\sum_{i=1}^{n}\E|\V_i|^{2+\omega} = O(1).
  \end{equation}
  By Assumption \ref{as:independent_general},
  \[\sum_{i=1}^{n}\E\left|\frac{\V_i}{\sqrt{n}\sigma}\right|^{2+\omega} = O\lb n^{-\omega / 2}\rb = o(1).\]
  The proof is completed by the Berry-Esseen inequality (Proposition \ref{prop:Berry_Esseen}) with $g(x) = x^{\omega}$.
\end{proof}

Let $\hV_i$ denote the plug-in estimate of $\V_i$, i.e.,
  \begin{align}\label{eq:hVi}
    \hV_i &= \Theta_i\Bigg\{\lb \swy - \hat{\tau}\sww\rb - \lb \sy - \hat{\tau}\sw\rb ^\tran\J \W_i  + \stheta\W_i^\tran \J \lb\td{\Y}_i - \hat{\tau}\W_i\rb - \sw^\tran\J\lb \td{\Y}_{i} - \hat{\tau}\W_i\rb \Bigg\}.
  \end{align}
  We first prove that $\hV_i$ is an accurate approximation of $\V_i$ on average, even without the independence assumption.
  \begin{lemma}\label{lem:hVi_Vi}
    Let $\bfPi$ be a solution of the \DATEeq. Under Assumptions \ref{as:latent_ign_app}-\ref{as:regularity_general},
    \[\frac{1}{n}\sum_{i=1}^{n}(\hV_i - \V_i)^2 = o_\P(1), \quad \text{if }\,\, \bar{\Delta}_{\pi}\bar{\Delta}_{y} = o(1).\]
  \end{lemma}
  \begin{proof}
    Let
    \begin{align*}
      \hV'_i &= \Theta_i\Bigg\{\lb \swy - \tau^{*}\sww\rb - \lb \sy - \tau^{*}\sw\rb ^\tran\J \W_i  + \stheta\W_i^\tran \J \lb\td{\Y}_i - \tau^{*}\W_i\rb - \sw^\tran\J\lb \td{\Y}_{i} - \tau^{*}\W_i\rb \Bigg\}.
    \end{align*}
    Then
    \begin{align*}
      \hV_i - \hV'_i &= (\hat{\tau} - \tau^{*})\Theta_i\big\{-\sww + \sw^\tran\J\W_i - \stheta\W_i^\tran\J\W_i + \sw^\tran\J\W_i\big\}.
    \end{align*}
    Under Assumption \ref{as:overlap_general}, there exists a constant $C$ that only depends on $c_{\pi}$ and $T$ such that
    \begin{equation*}
      |\hV_i - \hV'_i|\le C|\hat{\tau} - \tau^{*}|.
    \end{equation*}
    By Theorem \ref{thm:doubly_robust_expansion}, 
    \begin{equation}
      \label{eq:Vi_1}      
      \frac{1}{n}\sum_{i=1}^{n}(\hV_i - \hV'_i)^2 = O((\hat{\tau} - \tau^{*})^2) = o_\P(1)
    \end{equation}
    Next, 
    \begin{align*}
      \hV'_i - \V_i &= \Theta_i\Bigg\{\lb (\swy - \E[\swy])- \tau^{*}(\sww - \E[\sww])\rb - \lb (\sy - \E[\sy]) - \tau^{*}(\sw - \E[\sw])\rb^\tran\J \W_i\\
                    &\qquad\qquad + (\stheta - \E[\stheta])\W_i^\tran \J \lb\td{\Y}_i - \tau^{*}\W_i\rb - (\sw - \E[\sw])^\tran\J\lb \td{\Y}_{i} - \tau^{*}\W_i\rb \Bigg\}.
    \end{align*}
    By Jensen's inequality and Assumption \ref{as:overlap_general}, 
    \begin{align*}
      &\frac{1}{n}\sum_{i=1}^{n}(\hV'_i - \V_i)^2\\
      & \le \frac{5}{nc_{\pi}^2}\sum_{i=1}^{n}\Bigg\{(\swy - \E[\swy])^2 + (\sww - \E[\sww])^2 \cdot\tau^{*2} + \|(\sy - \E[\sy])\|_{2}^2 \cdot \|\J \W_i\|_{2}^2\\
      & \qquad\qquad + \|(\sw - \E[\sw])\|_{2}^2 \cdot \|\J(\td{\Y}_i - 2\tau^{*}\W_i)\|_{2}^2 + (\stheta - \E[\stheta])^2\lb\W_i^\tran \J \lb\td{\Y}_i - \tau^*\W_i\rb\rb^2\Bigg\}\\
      & = \frac{5}{c_{\pi}^2}\Bigg\{(\swy - \E[\swy])^2 + (\sww - \E[\sww])^2 \cdot\tau^{*2} + \|(\sy - \E[\sy])\|_{2}^2 \cdot T\\
      & \qquad\quad + \|(\sw - \E[\sw])\|_{2}^2\cdot \frac{1}{n}\sum_{i=1}^{n}\|(\td{\Y}_i - 2\tau^{*}\W_i)\|_{2}^2\\
      & \qquad\quad + \|(\stheta - \E[\stheta])\|_{2}^2\cdot \frac{T}{n}\sum_{i=1}^{n}\|\td{\Y}_i - \tau^{*}\W_i\|_{2}^2\Bigg\}.
    \end{align*}
    By Lemma \ref{lem:mean_var}, 
    \[\E\left[\frac{1}{n}\sum_{i=1}^{n}(\hV'_i - \V_i)^2\right] = o(1).\]
    By Markov's inequality,
    \begin{equation}
      \label{eq:Vi_2}
      \frac{1}{n}\sum_{i=1}^{n}(\hV'_i - \V_i)^2 = o_\P(1).
    \end{equation}
    Putting \eqref{eq:Vi_1} and \eqref{eq:Vi_2} together, we obtain that
    \begin{equation*}
      \frac{1}{n}\sum_{i=1}^{n}(\hV_i - \V_i)^2\le \frac{2}{n}\sum_{i=1}^{n}\{(\hV_i - \hV'_i)^2 + (\hV'_i - \V_i)^2\} = o_\P(1).
    \end{equation*}
  \end{proof}
  
As in Section \ref{sec:design_based}, we estimate the (conservative) variance of the term $\frac{1}{\sqrt{n}}\sum_{i=1}^{n}(\V_i - \E[\V_i])$ as
  \begin{equation}
    \label{eq:hatsigma2}
    \hat{\sigma}^2 = \frac{1}{n-1}\sum_{i=1}^{n}\lb \hV_i - \frac{1}{n}\sum_{i=1}^{n}\hV_i\rb^2 = \frac{n}{n-1}\left\{\frac{1}{n}\sum_{i=1}^{n}\hV_i^2 - \lb\frac{1}{n}\sum_{i=1}^{n}\hV_i\rb^2\right\}.
  \end{equation}
This yields a Wald-type confidence interval for DATE,
\begin{equation}
  \label{eq:doubly_robust_ci}
  \hat{C}_{1 - \alpha} = [\hat{\tau} - z_{1 - \alpha / 2}\hat{\sigma} / \sqrt{n}\hatdenom, \hat{\tau} + z_{1 - \alpha / 2}\hat{\sigma} / \sqrt{n}\hatdenom],
\end{equation}
where $z_{\eta}$ is the $\eta$-th quantile of the standard normal distribution. 
\begin{theorem}\label{thm:doubly_robust_inference_deterministic}
  Assume that $\bar{\Delta}_{\pi}\bar{\Delta}_{y} = o(1 / \sqrt{n})$. Under Assumptions \ref{as:latent_ign_app}, \ref{as:overlap_general}, \ref{as:independent_general}, and \ref{as:var_low}, 
  \[\liminf_{n\rightarrow \infty}\P\lb \tau^{*}\in \hat{C}_{1 - \alpha}\rb\ge 1 - \alpha.\]
\end{theorem}

\begin{proof}
  By Theorem \ref{thm:asymptotic_expansion_simple}, Theorem \ref{thm:numer_strong_deterministic}, Lemma \ref{lem:CLT}, and Assumption \ref{as:var_low},
  \[\frac{\hatdenom\cdot \sqrt{n}(\hat{\tau} - \tau^{*})}{\sigma} = \frac{1}{\sqrt{n}\sigma}\sum_{i=1}^{n}(\V_i - \E[\V_i]) + o_\P(1)\dcv N(0, 1) \text{ in Kolmogorov-Smirnov distance},\]
  As a result,
  \begin{equation}
    \label{eq:coverage_1}
    \bigg|\P\lb \bigg|\frac{\hatdenom\cdot \sqrt{n}(\hat{\tau} - \tau^{*})}{\sigma}\bigg|\le z_{1 - \alpha / 2}\cdot\frac{\hat{\sigma}}{\sigma}\rb - \left\{2\Phi\lb z_{1 - \alpha / 2}\cdot\frac{\hat{\sigma}}{\sigma}\rb - 1\right\}\bigg| = o(1),
  \end{equation}
  where $\Phi$ is the cumulative distribution function of the standard normal distribution. Let
  \begin{equation}
    \label{eq:sigma2+}
    \sigma^2_{+} = \frac{1}{n}\sum_{i=1}^{n}\E\lb \V_i - \frac{1}{n}\sum_{i=1}^{n}\E[\V_i]\rb^2 = \frac{1}{n}\sum_{i=1}^{n}\E[\V_i^2] - \lb\frac{1}{n}\sum_{i=1}^{n}\E[\V_i]\rb^2.
  \end{equation}
  Clearly, $\sigma^2_{+}$ is deterministic and
  \[\sigma^2_{+} = \sigma^2 + \frac{1}{n}\sum_{i=1}^{n}\lb\E[\V_i] - \frac{1}{n}\sum_{i=1}^{n}\E[\V_i]\rb^2\ge \sigma^2.\]
  It remains to show that
  \begin{equation}
    \label{eq:goal_sigma2}
    \left|\frac{n-1}{n}\hat{\sigma}^2 - \sigma^2_{+}\right| = o_\P(1).
  \end{equation}
  In fact, by Assumption \ref{as:var_low}, \eqref{eq:goal_sigma2} implies that
  \[\sqrt{\frac{n-1}{n}}\frac{\hat{\sigma}}{\sigma} \stackrel{p}{\rightarrow} \frac{\sigma_{+}}{\sigma} \ge 1\Longrightarrow \frac{\hat{\sigma}}{\sigma} \stackrel{p}{\rightarrow} \frac{\sigma_{+}}{\sigma} \ge 1.\]
  By continuous mapping theorem,
  \[2\Phi\lb z_{1 - \alpha / 2}\cdot\frac{\hat{\sigma}}{\sigma}\rb - 1\stackrel{p}{\rightarrow} 2\Phi\lb z_{1 - \alpha / 2}\cdot\frac{\sigma_{+}}{\sigma}\rb - 1 \ge 1 - \alpha,\]
  which completes the proof.

  Now we prove \eqref{eq:goal_sigma2}. By Proposition \ref{prop:vonbahresseen} and Jensen's inequality,
    \begin{align*}
      &\E\left|\frac{1}{n}\sum_{i=1}^{n}(\V_i^2 - \E[\V_i^2])\right|^{1+\omega/2} \le \frac{2}{n^{1 + \omega / 2}}\sum_{i=1}^{n}\E|\V_i^2 - \E[\V_i^2]|^{1 + \omega / 2}\\
      & \le \frac{2^{1 + \omega / 2}}{n^{1 + \omega / 2}}\sum_{i=1}^{n}\lb \E[|\V_i|^{2 + \omega}] + \E[\V_i^2]^{1 + \omega / 2}\rb \le \frac{2^{2 + \omega / 2}}{n^{1 + \omega / 2}}\sum_{i=1}^{n}\E\left[|\V_i|^{2 + \omega}\right].
    \end{align*}
    By \eqref{eq:Vi_moments},
    \[\E\left|\frac{1}{n}\sum_{i=1}^{n}(\V_i^2 - \E[\V_i^2])\right|^{1+\omega/2} = o(1).\]
    By Markov's inequality,
    \begin{equation}
      \label{eq:Vi_4}
      \frac{1}{n}\sum_{i=1}^{n}(\V_i^2 - \E[\V_i^2]) = o_\P(1).
    \end{equation}
    Similarly, we have that
    \begin{equation}
      \label{eq:Vi_5}
      \frac{1}{n}\sum_{i=1}^{n}(\V_i - \E[\V_i]) = o_\P(1).
    \end{equation}
    In addition, \eqref{eq:Vi_moments} and H\"{o}lder's inequality imply that
    \[\frac{1}{n}\sum_{i=1}^{n}\E[\V_i^2] = O(1), \quad \frac{1}{n}\sum_{i=1}^{n}\E[\V_i] = O(1).\]
    As a result,
    \begin{equation}
      \label{eq:Vi_6}
      \frac{1}{n}\sum_{i=1}^{n}\V_i^2 = O_\P(1), \quad \frac{1}{n}\sum_{i=1}^{n}\V_i = O_\P(1).
    \end{equation}
    By Lemma \ref{lem:hVi_Vi}, \eqref{eq:Vi_6}, and Cauchy-Schwarz inequality,
    \begin{align}
      &\left|\frac{1}{n}\sum_{i=1}^{n}\hV_i^2 - \frac{1}{n}\sum_{i=1}^{n}\V_i^2\right| \le \frac{2}{n}\sum_{i=1}^{n}\V_i|\hV_i - \V_i| + \frac{1}{n}\sum_{i=1}^{n}(\hV_i - \V_i)^2\nonumber\\
      & \le 2\sqrt{\frac{1}{n}\sum_{i=1}^{n}\V_i^2}\sqrt{\frac{1}{n}\sum_{i=1}^{n}(\hV_i - \V_i)^2} + \frac{1}{n}\sum_{i=1}^{n}(\hV_i - \V_i)^2 = o_\P(1).\label{eq:Vi_7}
    \end{align}
    Similarly,
    \begin{equation}
      \label{eq:Vi_8}
      \left|\lb\frac{1}{n}\sum_{i=1}^{n}\hV_i\rb^2 - \lb\frac{1}{n}\sum_{i=1}^{n}\V_i\rb^2\right| = o_\P(1).
    \end{equation}
    By \eqref{eq:Vi_4}, \eqref{eq:Vi_5}, and \eqref{eq:Vi_6},
    \begin{align}
      \left|\frac{1}{n}\sum_{i=1}^{n}\V_i^2 - \lb\frac{1}{n}\sum_{i=1}^{n}\V_i\rb^2 - \sigma^2_{+}\right| = o_\P(1).\label{eq:Vi_9}
    \end{align}

  Putting \eqref{eq:Vi_7} - \eqref{eq:Vi_9} together, we complete the proof of \eqref{eq:goal_sigma2}.
\end{proof}

\subsection{Inference with deterministic $(\hat{\bfpi}_i, \hat{\bfm}_{i}, \hbftau_i)$ and dependent assignments across units}\label{subapp:design_based_dependent}
Recall Theorem \ref{thm:doubly_robust_expansion} that
  \[\hatdenom\cdot \sqrt{n}(\hat{\tau} - \tau^{*}) = \frac{1}{\sqrt{n}}\sum_{i=1}^{n}(\V_i - \E[\V_i]) + o_\P(1).\]
  This is true even when $(\Y_i(1), \Y_i(0), \X_i)$ are dependent as long as Assumption \ref{as:regularity_general} holds. If $\V_i$'s are observable, a valid confidence interval for $\tau^{*}$ can be derived if the distribution of $(1/\sqrt{n})\sum_{i=1}^{n}(\V_i - \E[\V_i])$ can be approximated. Specifically, assume that
  \begin{equation}
    \label{eq:CLT_general}
    \frac{(1/\sqrt{n})\sum_{i=1}^{n}(\V_i - \E[\V_i])}{\sqrt{(1/n)\Var[\sum_{i=1}^{n}\V_i]}}\stackrel{d}{\rightarrow} N(0, 1),
  \end{equation}
  and there exists a conservative oracle variance estimator $\hat{\sigma}^{*2}$ based on $(\V_1, \ldots, \V_n)$ in the sense that
  \begin{equation}
    \label{eq:hatsigma2_general}
    \frac{(1/n)\Var[\sum_{i=1}^{n}\V_i]}{\hat{\sigma}^{*2}}\le 1 + o_\P(1).
  \end{equation}
  Then, $[\hat{\tau} - z_{1-\alpha/2}\hat{\sigma}^{*} / \sqrt{n}\hatdenom, \hat{\tau} + z_{1-\alpha/2}\hat{\sigma}^{*} / \sqrt{n}\hatdenom]$ is an asymptotically valid confidence interval for $\tau^{*}$. Of course, this interval cannot be computed in practice because $\V_i$ is unobserved due to the unknown quantities including $\E[\stheta], \E[\sw], \E[\sy], \E[\sww]$, $\E[\swy]$, and $\tau^{*}$. A natural variance estimator can be obtained by replacing $\bV\triangleq (\V_1, \ldots, \V_n)$ with $\hat{\bV}\triangleq (\hV_1, \ldots, \hV_n)$ in $\hat{\sigma}^{*2}$. The following theorem makes this intuition rigorous for generic quadratic oracle variance estimators.
  \begin{theorem}\label{thm:design_based_generic_inference}
    Suppose there exists an oracle variance estimator $\hat{\sigma}^{*2}$ such that
    \begin{enumerate}[(i)]
    \item $\hat{\sigma}^{*2} = \bV^\tran \bA_{n}\bV / n$ for some positive semidefinite (and potentially random) matrix $\bA_{n}$ with $\|\bA_{n}\|_{\op} = O_\P(1)$;
    \item $\hat{\sigma}^{*2}$ is conservative in the sense that, for every $\eta$ in a neighborhood of $\alpha$,
      \[\lim_{n\rightarrow\infty}\P\lb \bigg|\frac{(1/\sqrt{n})\sum_{i=1}^{n}(\V_i - \E[\V_i])}{\hat{\sigma}^{*}}\bigg|\ge z_{1 - \eta / 2}\rb\le \eta;\]
    \item $1 / \hat{\sigma}^{*2} = O_\P(1)$.
  \end{enumerate}
  Let $\hat{\sigma}^2 = \hat{\bV}^\tran \bA_n \hat{\bV} / n$ and
  \[\hat{C}_{1 - \alpha} = [\hat{\tau} - z_{1 - \alpha / 2}\hat{\sigma} / \sqrt{n}\hatdenom, \hat{\tau} + z_{1 - \alpha / 2}\hat{\sigma} / \sqrt{n}\hatdenom].\]
  Under Assumptions \ref{as:latent_ign_app}, \ref{as:overlap_general}, and \ref{as:regularity_general} with $q > 1/2$, if $\bfPi$ be an solution of the \DATEeq ~\eqref{eq:DATE_equation} and $\bar{\Delta}_{\pi}\bar{\Delta}_{y} = o(1 / \sqrt{n})$,
  \[\liminf_{n\rightarrow \infty}\P\lb \tau^{*}\in \hat{C}_{1 - \alpha}\rb\ge 1 - \alpha.\]
  \end{theorem}
  \begin{proof}
    By Lemma \ref{lem:hVi_Vi},
    \[\frac{1}{n}\|\hat{\bV} - \bV\|_{2}^2 = \frac{1}{n}\sum_{i=1}^{n}(\hV_i - \V_i)^2 = o_\P(1).\]
    Since $\bA_{n}$ is positive semidefinite, for any $\eps \in (0, 1)$,
    \[(1 - \eps)\hat{\sigma}^{*2} - \lb \frac{1}{\eps} - 1\rb \frac{1}{n}(\hat{\bV} - \bV)^\tran \bA_n (\hat{\bV} - \bV)\le \hat{\sigma}^{2} \le (1 + \eps)\hat{\sigma}^{*2} + \lb \frac{1}{\eps} + 1\rb \frac{1}{n}(\hat{\bV} - \bV)^\tran \bA_n (\hat{\bV} - \bV)\]
    Thus, for any $\eps \in (0, 1)$,
    \[\P\lb \hat{\sigma}^{2}\not\in [(1 - \eps)\hat{\sigma}^{*2}, (1 + \eps)\hat{\sigma}^{*2}]\rb = o(1).\]
    By condition (iii), the above result implies that
    \begin{equation}
      \label{eq:sigma_ratio}
      \left|\frac{\hat{\sigma}}{\hat{\sigma}^{*}} - 1\right| = o_\P(1).
    \end{equation}
    By Theorem \ref{thm:doubly_robust_expansion},
    \[\hatdenom\cdot \sqrt{n}(\hat{\tau} - \tau^{*}) = \frac{1}{\sqrt{n}}\sum_{i=1}^{n}(\V_i - \E[\V_i]) + o_\P(1).\]
    It remains to show that
    \[\lim_{n\rightarrow\infty}\P\lb \bigg|\frac{(1/\sqrt{n})\sum_{i=1}^{n}(\V_i - \E[\V_i])}{\hat{\sigma}}\bigg|\ge z_{1 - \alpha / 2}\rb\le \alpha.\]
    Let $\eta(\eps)$ be the quantity such that $z_{1 - \eta(\eps) / 2} = z_{1 - \alpha / 2}\cdot(1 - \eps)$. For any sufficiently small $\eps$ such that $\eta(\eps)$ lies in the neighborhood of $\alpha$ in condition (ii), 
    \begin{align*}
      &\P\lb \bigg|\frac{(1/\sqrt{n})\sum_{i=1}^{n}(\V_i - \E[\V_i])}{\hat{\sigma}}\bigg|\ge z_{1 - \alpha / 2}\rb\\
      & = \P\lb \bigg|\frac{(1/\sqrt{n})\sum_{i=1}^{n}(\V_i - \E[\V_i])}{\hat{\sigma}^{*}}\bigg|\ge z_{1 - \alpha / 2}\cdot\frac{\hat{\sigma}}{\hat{\sigma}^{*}}\rb\\
      & \le \P\lb \bigg|\frac{(1/\sqrt{n})\sum_{i=1}^{n}(\V_i - \E[\V_i])}{\hat{\sigma}^{*}}\bigg|\ge z_{1 - \eta(\eps) / 2}\rb + \P\lb \frac{\hat{\sigma}}{\hat{\sigma}^{*}}\le 1 - \eps\rb.
    \end{align*}
    By \eqref{eq:sigma_ratio}, when $n$ tends to infinity,
    \[\lim_{n\rightarrow\infty}\P\lb \bigg|\frac{(1/\sqrt{n})\sum_{i=1}^{n}(\V_i - \E[\V_i])}{\hat{\sigma}}\bigg|\ge z_{1 - \alpha / 2}\rb\le \eta(\eps).\]
    The proof is completed by letting $\eps\rightarrow 0$ and noting that $\lim_{\eps\rightarrow 0}\eta(\eps) = \alpha$.
  \end{proof}

  When $\W_i$'s are independent,
  \[\hat{\sigma}^{*2} = \frac{1}{n-1}\sum_{i=1}^{n}(\V_i - \bar{\V})^2.\]
  Thus, $\bA_{n} = (n/(n-1))(I_{n} - \one_{n}\one_{n}^{T} / n)$. Clearly, the condition (i) is satisfied because $\|\bA_{n}\|_{\op} = n / (n - 1)$. Under the assumptions in Theorem \ref{thm:doubly_robust_inference_deterministic}, the condition (ii) is satisfied. Moreover, we have shown that $\hat{\sigma}^{*2}$ converges to $\sigma_{+}^2\ge \sigma^2 > 0$, and thus the condition (iii) is satisfied. Therefore, Theorem \ref{thm:doubly_robust_inference_deterministic} can be implied by Theorem \ref{thm:design_based_generic_inference}.

  When $\V_i$'s are observed, the variance estimators are quadratic under nearly all types of dependent assignment mechanisms. With fixed potential outcomes, Theorem \ref{thm:design_based_generic_inference} applies to completely randomized experiments \citep{hoeffding1951combinatorial, li2017general}, blocked and matched experiments \citep{pashley2021insights}, two-stage randomized experiments \citep{ohlsson1989asymptotic}, and so on. Below, we prove the results for completely randomized experiments with fixed potential outcomes to illustrate how to apply Theorem \ref{thm:design_based_generic_inference}. The notation is chosen to mimic Theorem 5 and Proposition 3 in \cite{li2017general}. 

  \begin{theorem}\label{thm:design_based_without_replacement}
 Assume that $(Y_{it}(1), Y_{it}(0))$ are fixed, and $\hat{\bfpi}_i = \bfpi_i$ as in Section \ref{sec:design_based} (while $(\hat{m}_{it}, \htau_{it})$ are allowed to be non-zero). Consider a completely randomized experiments where the treatment assignments are sampled without replacement from $Q$ possible assignments $\{\w_{[1]}, \ldots, \w_{[Q]}\}$ with $n_{q}$ units assigned $\w_{[q]}$. Let $\bfPi$ be a solution of the \DATEeq ~\eqref{eq:DATE_equation} with support $\{\w_{[1]}, \ldots, \w_{[Q]}\}$, and $\V_i(q)$ be the ``potential outcome'' for $\V_i$ where $(Y_{it}, W_{it})$ is replaced by $(Y_{it}(\w_{[q], t}), \w_{[q], t})$, i.e.,
    \begin{align*}
      \V_i(q) &= \frac{\bfPi(\w_{[q]})}{\hat{\bfpi}_i(\w_{[q]})}\Bigg\{\lb\E[\swy] - \tau^{*}\E[\sww]\rb - \lb \E[\sy] - \tau^{*}\E[\sw]\rb ^\tran\J \w_{[q]} \\
           & \,\, + \E[\stheta]\w_{[q]}^\tran \J \lb\td{\Y}_i(q) - \tau^{*}\w_{[q]}\rb - \E[\sw]^\tran\J\lb \td{\Y}_{i}(q) - \tau^{*}\w_{[q]}\rb \Bigg\},
    \end{align*}
    and $\td{\Y}_i(q) = \big(Y_{i1}(\w_{[q], 1}) - \hat{m}_{i1} - \w_{[q], 1}\htau_{i1}, \ldots, Y_{iT}(\w_{[q], T}) - \hat{m}_{iT} - \w_{[q], 1}\htau_{iT}\big)$. Further, for any $q, r = 1, \ldots, Q$, let
    \[S_{q}^2 = \frac{1}{n-1}\sum_{i=1}^{n}(\V_i(q) - \bar{\V}(q))^2, \quad S_{qr} = \frac{1}{n-1}\sum_{i=1}^{n}(\V_i(q) - \bar{\V}(q))(\V_i(r) - \bar{\V}(r)),\]
    where $\bar{\V}(q) = (1/n)\sum_{i=1}^{n}\V_i(q)$. Define the variance estimate $\hat{\sigma}^2$ as
    \[\hat{\sigma}^2 = \sum_{q=1}^{Q}\frac{n_{q}}{n}s_{q}^2, \quad \text{where }s_{q}^2 = \frac{1}{n_{q} - 1}\sum_{i: \w_{i} = \w_{[q]}}(\hat{\V}_i - \hat{\bar{\V}}(q))^2, \,\, \hat{\bar{\V}}(q) = \frac{1}{n_{q}}\sum_{i: \w_{i} = \w_{[q]}}\hat{\V}_i.\]
    Further, define the confidence interval as
    \[  \hat{C}_{1 - \alpha} = [\hat{\tau} - z_{1 - \alpha / 2}\hat{\sigma} / \sqrt{n}\hatdenom, \hat{\tau} + z_{1 - \alpha / 2}\hat{\sigma} / \sqrt{n}\hatdenom].\]
    Assume that
    \begin{enumerate}[(a)]
    \item $Q = O(1)$ and $n_{q} / n\rightarrow \pi_{q}$ for some constant $\pi_{q} > 0$;
    \item for any $q, r = 1, \ldots, Q$, $S_{q}^2$ and $S_{qr}$ have limiting values $S_{q}^{*2}, S_{qr}^{*}$;
    \item there exists a constant $c_{\tau} > 0$ such that $\sum_{q=1}^{Q}\pi_{q}S_{q}^{*2} > c_{\tau}$;
    \item there exists a constant $M<\infty$ such that $\max_{i,q}\{\|\td{\Y}_{i}(q)\|_{2}\} <M$.
    \end{enumerate}
    Then,
  \[\liminf_{n\rightarrow \infty}\P\lb \tau^{*}\in \hat{C}_{1 - \alpha}\rb\ge 1 - \alpha.\]
\end{theorem}

\begin{proof}
  By definition, for any $i\not= j \in [n]$ and $q\not = r \in [Q]$,
  \[\P(\W_i = \w_{[q]}) = \frac{n_{q}}{n}, \quad \P(\W_i = \W_{j} = \w_{[q]}) = \frac{n_{q}(n_{q} - 1)}{n(n - 1)}, \quad \P(\W_i = \w_{[q]}, \W_{j} = \w_{[r]}) = \frac{n_{q}n_{r}}{n(n - 1)}.\]
  For any functions $f$ and $g$ on $[0, 1]^{T}$,
  \[\E[f(\W_i)] = \sum_{q=1}^{Q}\frac{n_{q}}{n}f(\w_{[q]}), \quad \E[g(\W_j)] = \sum_{q=1}^{Q}\frac{n_{q}}{n}g(\w_{[q]}),\]
  \[\E[f^2(\W_i)] = \sum_{q=1}^{Q}\frac{n_{q}}{n}f^2(\w_{[q]}), \quad \E[g^2(\W_j)] = \sum_{q=1}^{Q}\frac{n_{q}}{n}g^2(\w_{[q]}),\]
  and
  \[\E[f(\W_i)g(\W_i)] = \sum_{q=1}^{Q}\frac{n_{q}}{n}f(\w_{[q]})g(\w_{[q]}),\]
  \[\E[f(\W_i)g(\W_j)] = \sum_{q=1}^{Q}\frac{n_{q}(n_{q} - 1)}{n(n - 1)}f(\w_{[q]})g(\w_{[q]}) + \sum_{q\not=r}\frac{n_{q}n_{r}}{n(n - 1)}f(\w_{[q]})g(\w_{[r]}).\]
  As a result, for any $i\not= j$
  \begin{align*}
    &\Cov(f(\W_i), g(\W_j)) = \E[f(\W_i)g(\W_j)] - \E[f(\W_i)]\E[g(\W_j)]\\
    & = \sum_{q=1}^{Q}\lb\frac{n_{q}(n_{q} - 1)}{n(n - 1)} - \frac{n_{q}^2}{n^2}\rb f(\w_{[q]})g(\w_{[q]}) + \sum_{q\not=r}\lb\frac{n_{q}n_{r}}{n(n - 1)} - \frac{n_{q}n_{r}}{n^2}\rb f(\w_{[q]})g(\w_{[r]})\\
    & = \sum_{q=1}^{Q}-\frac{n_{q}(n - n_{q})}{n^2(n - 1)}f(\w_{[q]})g(\w_{[q]}) + \sum_{q\not=r}\frac{n_{q}n_{r}}{n^2(n - 1)} f(\w_{[q]})g(\w_{[r]})\\
    & = -\frac{1}{n-1}\sum_{q=1}^{Q}\frac{n_{q}}{n}f(\w_{[q]})g(\w_{[q]}) + \frac{1}{n-1}\lb\sum_{q=1}^{Q}\frac{n_{q}}{n}f(\w_{[q]})\rb \lb\sum_{q=1}^{Q}\frac{n_{q}}{n}g(\w_{[r]})\rb\\
    & = -\frac{1}{n-1}\lb \E[f(\W_i)g(\W_i)] - \E[f(\W_i)]\E[g(\W_i)]\rb\\
    & = -\frac{1}{n-1}\Cov(f(\W_i), g(\W_i))
  \end{align*}
  By Cauchy-Schwarz inequality,
  \begin{align*}
      & |\Cov(f(\W_i), g(\W_j))|\le \frac{1}{n-1}|\Cov(f(\W_i), g(\W_i))|\\
      & \le \frac{1}{n-1}\sqrt{\Var[f(\W_i)]\Var[g(\W_i)]} = \frac{1}{n-1}\sqrt{\Var[f(\W_i)]\Var[g(\W_j)]}.
  \end{align*}
  This implies that
  \[\rho_{ij}\le \frac{1}{n-1}\Longrightarrow \rho_i \le 2.\]
  It is then clear that Assumption \ref{as:regularity_general} holds under the condition (d). Further, since $\hat{\bfpi}_i(\w_{[q]}) = \bfpi_i(\w_{[q]}) = n_q / n$, the condition (a) implies Assumption \ref{as:overlap_general} and that $\bar{\Delta}_{\pi}\bar{\Delta}_{y} = 0$ . On the other hand, Assumption \ref{as:latent_ign_app} holds because $\W_i$ is completely randomized. Therefore, it remains to check the condition (i) - (iii) in Theorem \ref{thm:design_based_generic_inference} with
  \[\hat{\sigma}^{*2} = \sum_{q=1}^{Q}\frac{n_{q}}{n}s_{q}^{*2}, \quad \text{where }s_{q}^{*2} = \frac{1}{n_{q} - 1}\sum_{i: \w_{i} = \w_{[q]}}(\V_i - \bar{\V}(q))^2, \,\, \bar{\V}(q) = \frac{1}{n_{q}}\sum_{i: \w_{i} = \w_{[q]}}\V_i.\]
  In this case, $\bA_{n}$ is a block-diagonal matrix with
  \[\bA_{n, \mathcal{I}_{q}, \mathcal{I}_{q}} = \frac{n_q}{n_q - 1}\lb I_{n_q} -  \frac{\one_{n_q}\one_{n_q}^\tran}{n_q}\rb,\]
  where $\mathcal{I}_{q} = \{i: \W_i = \w_{[q]}\}$. As a result,
  \[\|\bA_{n}\|_{\op} = \max_{q}\frac{n_q}{n_q - 1} = O(1).\]
  Thus, the condition (i) holds. The condition (ii) is implied by Proposition 3 in \cite{li2017general} and the condition (iii) is implied by the condition (c). The theorem is then implied by Theorem \ref{thm:design_based_generic_inference}.
\end{proof}

\subsection{Inference of RIPW estimators when $\bar{\delta}_{\pi}\bar{\delta}_{y}\neq o(1/\sqrt{n})$}\label{subapp:global_misspecification}

In this section, we study the asymptotic inference on the RIPW estimator when one of the models is globally misspecified. Doubly robust inference is hard even for cross-sectional data \citep[e.g.][]{benkeser2017doubly}. Here, we focus on a practically relevant case where the researcher fits parametric models for both assignments and outcomes. Specifically, we consider a parametric family $f_{\bfkappa}(\w, \X_i)$ for $\bfpi_i(\w)$, $g_{t, \bfphi_t}(\X_i)$ for $m_{it}$, and $h_{t, \bfpsi_t}(\X_i)$ for $\nu_{it}$, where $\bfkappa\in \R^d, \bfphi_t\in \R^{d_{\phi, t}}, \bfpsi_t\in \R^{d_{\psi, t}}$ are parameter vectors to estimate. We impose the standard regularity conditions on these parametric families. 
\begin{assumption}\label{ass:parametric_regularity}
\begin{enumerate}[(a)]
\item $1/f_{\bfkappa}(\w; \mathbf{x})$ is bounded away from zero uniformly over $(\bfkappa, \w, \mathbf{x})$;
\item $\|\nabla_{\bfkappa}^2 f_{\bfkappa}(\w; \mathbf{x})\|$ is uniformly bounded over $(\bfkappa, \w, \mathbf{x})$;
\item $\|\nabla_{\bfphi_t}^2 g_{t, \bfphi_t}( \mathbf{x})\| + \|\nabla_{\bfpsi_t}^2 h_{t, \bfpsi_t}(\mathbf{x})\|$ is uniformly bounded over $(\bfphi_t, \bfpsi_t, \mathbf{x}, t)$.
\end{enumerate}
\end{assumption}

To ease notation, we denote by $\bfphi$ (resp. $\bfpsi$) the concatenation of $\bfphi_1, \ldots, \bfphi_T$ (resp. $\bfpsi_1, \ldots, \bfpsi_T$) and by $\bftheta$ the concatenation of $\bfkappa, \bfphi, \bfpsi$. We assume that $\hat{\bftheta}$ has an asymptotically linear expansion:
 \begin{assumption}\label{ass:theta_asym_linear}
 For some pseudo parameter $\bftheta'$,
 \begin{equation*}
\hat{\bftheta} = \bftheta' + \frac{1}{n}\sum_{i=1}^{n}\bfC_i + o_\P(n^{-1/2}),
\end{equation*}
where $\bfC_i$ is a function of $(\Y_i, \W_i, \X_i)$ that has zero mean and bounded second moment. 
\end{assumption}
When $(\Y_i(1), \Y_i(0), \W_i)$ are independent or weakly dependent, Assumption \ref{ass:theta_asym_linear} holds under standard regularity conditions \citep{fan2003nonlinear, wooldridge2010econometric}
 with root-$n$ rate. In particular, $\|\hat{\bftheta} - \bftheta\| = O_\P(1/\sqrt{n})$.

Let $\bfkappa', \bfphi_t', \bfpsi_t'$ denote the corresponding elements of $\bftheta'$. Further let 
\begin{equation}\label{eq:oracle_parametric}
\bfpi_i'(\w) = f_{\bfkappa'}(\w, \X_i), \quad m_{it}' = g_{t, \bfphi_t'}(\X_i), \quad \nu_{it}' = h_{t, \bfpsi_t'}(\X_i).
\end{equation}
We say the assignment model is correctly specified if 
\begin{equation}\label{eq:correct_parametric_assignment}
\bfpi_i(\w) = \bfpi_i'(\w),
\end{equation}
and the outcome model is correctly specified if 
\begin{equation}\label{eq:correct_parametric_outcome}
\bfm_{i} = \bfm_{i}', \quad \nhbftau_{i} = \nhbftau_{i}'.
\end{equation}
When one of these two models is globally misspecified and the other one is correctly specified, 
\[\bar{\delta}_\pi \bar{\delta}_y = O(1/\sqrt{n}).\]
Thus, Theorem \ref{thm:dr_linear_expansion} and \ref{thm:dr_coverage} do not apply. Nevertheless, we prove that the RIPW estimator remains to be asymptotically linear though an additional term is added to $\V_i$ to account for the estimation uncertainty of $\hat{\bftheta}$ under the assumption that units are independent. We leave the general dependent case for future research. 

Since the result is very complicated, we first define several quantities. Let $\stheta', \sww', \swy', \sw', \sy'$, $\Theta_i', \td{Y}_{it}'$ be the counterparts of $\stheta, \sww, \swy, \sw$, $\sy, \Theta_i, \td{Y}_{it}$ with $(\hat{\bfpi}_i, \hat{\bfm}_i, \hat{\nhbftau}_i)$ replaced by $(\bfpi_i', \bfm_i', \nhbftau_i')$. Further let
\[\L_i = \nabla_{\bfkappa} \log f_{\bfkappa'}(\W_i, X_i), \quad \G_i = \begin{bmatrix}
\nabla_{\bfphi_1}g_{1, \bfphi_1'}(\X_i) & 0 & \cdots & 0\\
0 & \nabla_{\bfphi_2}g_{2, \bfphi_2'}(\X_i) & \cdots & 0\\
\vdots & \vdots & \vdots & \vdots\\
0 & 0 & \cdots & \nabla_{\bfphi_T}g_{T, \bfphi_T'}(\X_i)
\end{bmatrix},\]
and
\[\H_i = \begin{bmatrix}
\nabla_{\bfpsi_1}h_{1, \bfpsi_1'}(\X_i) & 0 & \cdots & 0\\
0 & \nabla_{\bfpsi_2}h_{2, \bfpsi_2'}(\X_i) & \cdots & 0\\
\vdots & \vdots & \vdots & \vdots\\
0 & 0 & \cdots & \nabla_{\bfpsi_T}h_{T, \bfpsi_T'}(\X_i)
\end{bmatrix}.\]
For any vector $\mathbf{v}$ that has the same dimension as $\bftheta$, let $\mathcal{P}_{\bfkappa}(\mathbf{v}), \mathcal{P}_{\bfphi}(\mathbf{v}), \mathcal{P}_{\bfpsi}(\mathbf{v})$ denote the subvectors corresponding to the positions of $\bfkappa, \bfphi, \bfpsi$ in $\bftheta$. We define $\bfA$ and $\bfB$ as two vectors such that 
\begin{align*}
\mathcal{P}_{\bfkappa}(\bfA) & = \frac{1}{n}\sum_{i=1}^{n}\Big\{-\E[\Theta_i'\L_i]\E[\swy'] - \E[\Theta_i'\L_i\W_i^\tran \J \td{\Y}_i]\E[\stheta'] +  \E[\Theta_i'\L_i\td{\Y}_i^\tran \J]\E[\sw']\\
& \qquad \qquad \quad+ \E[\Theta_i'\L_i\W_i^\tran \J]\E[\sy']\Big\}\\
\mathcal{P}_{\bfphi}(\bfA) & = \frac{1}{n}\sum_{i=1}^{n} \left\{\E[\Theta_i'\G_i\J]\E[\sw'] - \E[\Theta_i'\G_i\J \W_i]\E[\stheta']\right\}\\
\mathcal{P}_{\bfpsi}(\bfA) & = \frac{1}{n}\sum_{i=1}^{n} \left\{\E[\Theta_i'\H_i\diag(\W_i)\J]\E[\sw'] - \E[\Theta_i'\H_i\diag(\W_i)\J \W_i]\E[\stheta']\right\},
\end{align*}
and 
\begin{align*}
\mathcal{P}_{\bfkappa}(\bfB) & =\frac{1}{n}\sum_{i=1}^{n}\lb 2\E[\Theta_i'\L_i\W_i^\tran \J \W_i]\E[\sw] - \E[\Theta_i'\L_i]\E[\sww'] - \E[\Theta_i'\L_i\W_i^\tran \J \W_i]\E[\stheta']\rb\\
\mathcal{P}_{\bfphi}(\bfB) &= \mathcal{P}_{\bfpsi}(\bfB) = 0.
\end{align*}

\begin{theorem}\label{thm:parametric}
Assume that $(\Y_i(1), \Y_i(0), \W_i, \X_i)$ are independent. In the setting of Theorem \ref{thm:dr_linear_expansion}, under Assumptions \ref{ass:parametric_regularity} and \ref{ass:theta_asym_linear}, 
  \[\hatdenom\cdot \sqrt{n}(\hat{\tau}(\bfPi) - \tau^{*}) = \frac{1}{\sqrt{n}}\sum_{i=1}^{n}(\V'_{i} + \cU_i - \E[\V'_{i}]) + o_\P(1/\sqrt{n}),\]
  if either the assignment model or the outcome model is correctly specified. Above, $\V'_i$ is defined as in Theorem \ref{thm:dr_linear_expansion} with $(\hat{\bfpi}_i, \hat{\bfm}_i, \hbftau_i)$ replaced by $(\bfpi_i', \bfm_i', \nhbftau_i')$ and 
  \begin{align*}
  \cU_i &= \left\langle \bfC_i, \bfA - \bfB \tau^{*}\right\rangle, \quad \E[\cU_i] = 0,
  \end{align*}
  where $\bfA$ and $\bfB$ are defined above.
\end{theorem}

\begin{remark}
To make inference, we can replace $\bfA$, $\bfB$, and $\bfC_i$ by their plug-in estimates. It is straightforward, though tedious, to prove that the plug-in variance estimator is consistent. Importantly, this does not require the researcher to know which model is misspecified apriori.
\end{remark}

\begin{proof}
Write $\hat{\tau}$ for $\hat{\tau}(\bfPi)$. Further let 
\[\hat{\tau}' = \frac{\hatnumer'}{\hatdenom'}, \quad \hatnumer' = \swy'\stheta' - \sw^{'\tran}\sy', \quad \hatdenom' = \sww'\stheta' - \sw^{'\tran}\sw'.\]
We will show that 
\begin{equation}\label{eq:parametric_goal}
\hatnumer = \hatnumer'+ \left\langle\bfA, \hat{\bftheta} - \bftheta\right\rangle + o_\P\lb\frac{1}{\sqrt{n}}\rb, \quad \hatdenom = \hatdenom' + \left\langle \bfB, \hat{\bftheta} - \bftheta\right\rangle + o_\P\lb\frac{1}{\sqrt{n}}\rb.
\end{equation}
By definition and Assumption \eqref{as:regularity_general}, $\|\bfA\| + \|\bfB\| = O_\P(1)$. Under Assumption \ref{ass:theta_asym_linear}, $\|\hat{\bftheta} - \bftheta\| = O_\P(1)$. By Lemma \ref{lem:denom} and \eqref{eq:parametric_goal},
\[\hatdenom \cdot \sqrt{n}(\hat{\tau} - \hat{\tau}') = \left\langle\bfA  - \bfB \hat{\tau}', \sqrt{n}(\hat{\bftheta} - \bftheta)\right\rangle + o_\P\lb 1\rb.\]
It is easy to see that the assumptions of Theorem \ref{thm:doubly_robust_expansion} are implied by independence, Assumption \ref{as:bounded_outcomes_dr}, and Assumption \ref{ass:parametric_regularity}. By Theorem \ref{thm:doubly_robust_expansion}, 
\[\hat{\tau}' = \tau^{*} + o_\P(1).\]
Together with Assumption \ref{ass:theta_asym_linear}, it implies 
\begin{align*}
\hatdenom \cdot \sqrt{n}(\hat{\tau} - \hat{\tau}') &= \left\langle\bfA  - \bfB \tau^{*}, \sqrt{n}(\hat{\bftheta} - \bftheta)\right\rangle + o_\P\lb 1\rb = \frac{1}{\sqrt{n}}\sum_{i=1}^{n}\cU_i + o_\P(1).
\end{align*}
By Theorem \ref{thm:doubly_robust_expansion} again, 
\[\hatdenom\cdot \sqrt{n}(\hat{\tau}' - \tau^{*}) = \frac{1}{\sqrt{n}}\sum_{i=1}^{n}(\V'_i - \E[\V'_i]) + o_\P(1).\]
Combining the two pieces yields the desired result.

Now we turn to proving \eqref{eq:parametric_goal}. By Assumption \ref{ass:parametric_regularity}, 
\begin{equation}\label{eq:parametric_Theta_diff}
\Theta_i - \Theta_i' = \frac{\bfPi(\W_i)}{f_{\hat{\bfkappa}}(\W_i, \X_i)} - \frac{\bfPi(\W_i)}{f_{\bfkappa'}(\W_i, \X_i)} = -\Theta_i' \left\langle \L_i, \hat{\bfkappa} - \bfkappa'\right\rangle + O(1) \cdot \|\hat{\bftheta} - \bftheta\|^2.
\end{equation}
where the $O(1)$ terms are uniformly bounded across all units. By \eqref{eq:parametric_Theta_diff}, 
\begin{equation*}
\stheta - \stheta' = -\left\langle \frac{1}{n}\sum_{i=1}^{n}\Theta_i'\L_i, \hat{\bfkappa} - \bfkappa'\right\rangle + o_\P\lb \frac{1}{\sqrt{n}}\rb
\end{equation*}
Similar to Lemma \ref{lem:mean_var}, we can show 
\[\frac{1}{n}\sum_{i=1}^{n}\Theta_i'\L_i = \frac{1}{n}\sum_{i=1}^{n}\E[\Theta_i'\L_i] + O_\P\lb\frac{1}{\sqrt{n}}\rb.\]
Thus, 
\begin{equation}\label{eq:parametric_stheta_diff}
\stheta - \stheta' = -\left\langle \frac{1}{n}\sum_{i=1}^{n}\E[\Theta_i'\L_i], \hat{\bfkappa} - \bfkappa'\right\rangle + o_\P\lb \frac{1}{\sqrt{n}}\rb.
\end{equation}
Using the same argument, we can prove that 
\begin{equation}\label{eq:parametric_sww_diff}
\sww - \sww' = -\left\langle \frac{1}{n}\sum_{i=1}^{n}\E[\Theta_i'\L_i\W_i^\tran \J \W_i], \hat{\bfkappa} - \bfkappa'\right\rangle + o_\P\lb \frac{1}{\sqrt{n}}\rb,
\end{equation}
and 
\begin{equation}\label{eq:parametric_sw_diff}
\sw - \sw' = -\left\langle \frac{1}{n}\sum_{i=1}^{n}\E[\Theta_i'\L_i\W_i^\tran \J], \hat{\bfkappa} - \bfkappa'\right\rangle + o_\P\lb \frac{1}{\sqrt{n}}\rb,
\end{equation}
where we define $\langle A, b \rangle$ to be $A^\tran b$ for a matrix $A$ and vector $b$. In particular, 
\[|\stheta - \stheta'| + |\sww - \sww'| + \|\sw - \sw'\| = O_\P\lb\frac{1}{\sqrt{n}}\rb.\]
Putting \eqref{eq:parametric_stheta_diff} - \eqref{eq:parametric_sw_diff} and Lemma \ref{lem:mean_var} together, we have
\begin{align*}
&\hatdenom - \hatdenom' \\
&= \sww'(\stheta - \stheta') + \stheta'(\sww - \sww') - 2\sw^{'\tran}(\sw - \sw') + o_\P\lb \frac{1}{\sqrt{n}}\rb\\
& = \E[\sww'](\stheta - \stheta') + \E[\stheta'](\sww - \sww') - 2\E[\sw']^{\tran}(\sw - \sw') + o_\P\lb \frac{1}{\sqrt{n}}\rb\\
& = \left\langle \frac{1}{n}\sum_{i=1}^{n}\lb 2\E[\Theta_i'\L_i\W_i^\tran \J]\E[\sw] - \E[\Theta_i'\L_i]\E[\sww'] - \E[\Theta_i'\L_i\W_i^\tran \J \W_i]\E[\stheta']\rb, \hat{\bfkappa} - \bfkappa'\right\rangle\\
& \qquad + o_\P\lb \frac{1}{\sqrt{n}}\rb
\end{align*}
By definition of $\bfB$, 
\[\hatdenom - \hatdenom'  = \left\langle \bfB, \hat{\bftheta} - \bftheta\right\rangle + o_\P\lb \frac{1}{\sqrt{n}}\rb.\]
This proves the second part of \eqref{eq:parametric_goal}. 

To prove the first part of \eqref{eq:parametric_goal}, we first note that, similar to \eqref{eq:parametric_Theta_diff}, Assumption \ref{ass:parametric_regularity} implies 
\begin{equation}\label{eq:parametric_tdY_diff}
\td{\Y}_{i} - \td{\Y}_{i}' = -\left\langle \G_i, \hat{\bfphi} - \bfphi \right\rangle - \left\langle \H_i \diag(\W_i),  \hat{\bfpsi} - \bfpsi\right\rangle + O(1) \cdot \|\hat{\bftheta} - \bftheta\|^2,
\end{equation}
where the $O(1)$ terms are uniformly bounded across all units. 
Together with \eqref{eq:parametric_Theta_diff}, 
\begin{align*}
&\swy - \swy' \\
&= \frac{1}{n}\sum_{i=1}^{n}\Theta_i' \W_i^\tran \J (\td{\Y}_i - \td{\Y}_i') + \frac{1}{n}\sum_{i=1}^{n}(\Theta_i - \Theta_i') \W_i^\tran \J \td{\Y}_i' + o_\P\lb\frac{1}{\sqrt{n}}\rb\\
& = -\left\langle \frac{1}{n}\sum_{i=1}^{n} \Theta_i'\G_i\J\W_i, \hat{\bfphi} - \bfphi \right\rangle -  \left\langle  \frac{1}{n}\sum_{i=1}^{n} \Theta_i'\H_i \diag(\W_i)\J\W_i,  \hat{\bfpsi} - \bfpsi\right\rangle\\
& \qquad -\left\langle \frac{1}{n}\sum_{i=1}^{n}\Theta_i'\L_i\W_i^\tran \J \td{\Y}_i, \hat{\bfkappa} - \bfkappa\right\rangle + o_\P\lb\frac{1}{\sqrt{n}}\rb.
\end{align*}
Similar to \eqref{eq:parametric_stheta_diff} - \eqref{eq:parametric_sw_diff}, 
\begin{align}\label{eq:parametric_swy_diff}
\swy - \swy' &= -\left\langle \frac{1}{n}\sum_{i=1}^{n} \E[\Theta_i'\G_i\J\W_i], \hat{\bfphi} - \bfphi \right\rangle -  \left\langle  \frac{1}{n}\sum_{i=1}^{n} \E[\Theta_i'\H_i \diag(\W_i)\J\W_i],  \hat{\bfpsi} - \bfpsi\right\rangle\nonumber\\
& \qquad -\left\langle \frac{1}{n}\sum_{i=1}^{n}\E[\Theta_i'\L_i\W_i^\tran \J \td{\Y}_i], \hat{\bfkappa} - \bfkappa\right\rangle + o_\P\lb\frac{1}{\sqrt{n}}\rb.
\end{align}
Using the same argument, we can show 
\begin{align}\label{eq:parametric_sy_diff}
\sy - \sy' &= -\left\langle \frac{1}{n}\sum_{i=1}^{n} \E[\Theta_i'\G_i\J], \hat{\bfphi} - \bfphi \right\rangle -  \left\langle  \frac{1}{n}\sum_{i=1}^{n} \E[\Theta_i' \H_i \diag(\W_i)\J],  \hat{\bfpsi} - \bfpsi\right\rangle\nonumber\\
& \qquad -\left\langle \frac{1}{n}\sum_{i=1}^{n}\E[\Theta_i'\L_i\td{\Y}_i^\tran \J], \hat{\bfkappa} - \bfkappa\right\rangle + o_\P\lb\frac{1}{\sqrt{n}}\rb.
\end{align}
As a result,
\begin{align*}
&\hatnumer - \hatnumer' \\
& = \swy'(\stheta - \stheta') + \stheta'(\swy - \swy') - \sw^{'\tran} (\sy - \sy') - \sy^{'\tran} (\sw - \sw') + o_\P\lb\frac{1}{\sqrt{n}}\rb\\
& = \E[\swy'](\stheta - \stheta') + \E[\stheta'](\swy - \swy') - \E[\sw']^\tran (\sy - \sy') - \E[\sy']^\tran (\sw - \sw') + o_\P\lb\frac{1}{\sqrt{n}}\rb\\
& = \Bigg\langle \frac{1}{n}\sum_{i=1}^{n}\Big\{-\E[\Theta_i'\L_i]\E[\swy'] - \E[\Theta_i'\L_i\W_i^\tran \J \td{\Y}_i]\E[\stheta'] +  \E[\Theta_i'\L_i\td{\Y}_i^\tran \J]\E[\sw']\\
& \qquad \qquad \quad + \E[\Theta_i'\L_i\W_i^\tran \J]\E[\sy']\Big\}, \hat{\bfkappa} - \bfkappa\Bigg\rangle \\
& \qquad + \left\langle \frac{1}{n}\sum_{i=1}^{n} \left\{\E[\Theta_i'\G_i\J]\E[\sw'] - \E[\Theta_i'\G_i\J \W_i]\E[\stheta']\right\}, \hat{\bfphi} - \bfphi\right\rangle\\
& \qquad + \left\langle \frac{1}{n}\sum_{i=1}^{n} \left\{\E[\Theta_i'\H_i\diag(\W_i)\J]\E[\sw'] - \E[\Theta_i'\H_i\diag(\W_i)\J \W_i]\E[\stheta']\right\}, \hat{\bfpsi} - \bfpsi\right\rangle\\
& \qquad + o_\P\lb\frac{1}{\sqrt{n}}\rb.
\end{align*}
By definition of $\bfA$, 
\[\hatnumer - \hatnumer' = \left\langle \bfA, \hat{\bftheta} - \bftheta\right\rangle + o_\P\lb\frac{1}{\sqrt{n}}\rb.\]

\end{proof}

\subsection{Proof of Proposition \ref{prop:effective_xi}: induced weights are non-negative}\label{subapp:proof_effective_xi}
\begin{proof}
  Let $\bfomega = \EPi[\W]$. Then 
  \begin{align*}
      &\EPi[\diag(\W)\J (\W - \EPi[\W])]\\ 
      & = \EPi[\diag(\W)(\W - \EPi[\W])] - \frac{1}{T}\EPi[\diag(\W)\one_{T}\one_{T}^\tran (\W - \EPi[\W])]\\
      & = \EPi[\W] - \EPi[\diag(\W)]\EPi[\W] - \EPi\left[\W \lb\frac{\one_T^\tran \W}{T}\rb\right] + \EPi[\W]\frac{\one_{T}^\tran \EPi[\W]}{T}\\
      & = \bfomega - \diag(\bfomega)\bfomega - \EPi\left[\W \lb\frac{\one_T^\tran \W}{T}\rb\right] + \bfomega \lb\frac{\one_{T}^\tran \bfomega}{T}\rb.
  \end{align*}
  By \eqref{eq:effective_xi}, for any $t$,
  \begin{align}
    &\EPi\left[\left\|\td{\W} - \EPi[\td{\W}]\right\|_{2}^2\right]\xi_{t} = \bfomega_{t} - \bfomega_{t}^2 - \EPi\left[\W_t \lb\frac{\one_{T}^\tran \W}{T}\rb\right] + \bfomega_{t} \lb\frac{\one_{T}^\tran \bfomega}{T}\rb.\label{eq:xit_etat}
  \end{align}
  Now we consider two scenarios. 
  \begin{enumerate}
  \item If $\bfomega_t \le \one_{T}^\tran \bfomega / T$, 
  \[\EPi\left[\W_t \lb\frac{\one_{T}^\tran \W}{T}\rb\right]\le \EPi[\W_t] = \bfomega_{t}.\]
  Then the right-hand side of \eqref{eq:xit_etat} is lower bounded by 
  \[\bfomega_t - \bfomega_t^2 - \bfomega_t + \bfomega_{t} \lb\frac{\one_{T}^\tran \bfomega}{T}\rb = \bfomega_{t}\lb \frac{\one_{T}^\tran \bfomega}{T} - \bfomega_{t}\rb\ge 0.\]
  \item If $\bfomega_t > \one_{T}^\tran \bfomega / T$,
    \[\EPi\left[\W_t \lb\frac{\one_{T}^\tran \W}{T}\rb\right]\le \EPi\left[\frac{\one_{T}^\tran \W}{T}\right] = \frac{\one_{T}^\tran\bfomega}{T}.\]
    Then the right-hand side of \eqref{eq:xit_etat} is lower bounded by 
    \[\bfomega_t - \bfomega_t^2 - \frac{\one_{T}^\tran \bfomega}{T} + \bfomega_{t} \lb\frac{\one_{T}^\tran \bfomega}{T}\rb = \lb \bfomega_{t} - \frac{\one_{T}^\tran \bfomega}{T}\rb (1 - \bfomega_{t})\ge 0.\]
  \end{enumerate}
\end{proof}

\subsection{Extension to generalized DATE in Remark \ref{rem:generalized_DATE}}\label{subapp:generalized_DATE}
We prove that, under Assumptions \ref{as:latent_ign_app} - \ref{as:regularity_general}, $\hat{\tau}(\bfPi; \zeta) = \tau^{*}(\xi; \zeta) + o_\P(1)$ if, further,  $\hat{\bfpi}_i = \bfpi_i$ and $n\|\zeta\|_{\infty} = O(1)$. 

To prove consistency, we just need to modify the proofs in Appendix \ref{subsubapp:asymptotic_linear} and \ref{subsubapp:consistency} by redefining $\Theta_i$ as
\[\Theta_i = \frac{(n\zeta_i) \bfPi(\W_i)}{\bfpi_i(\W_i)},\]
and redefining $\stheta, \sw, \sww, \swy, \sy$ correspondingly. 
First, we can apply the same arguments to show that Lemma \ref{lem:mean_var}, Lemma \ref{lem:denom}, and Theorem \ref{thm:asymptotic_expansion_simple} continue to hold under the above assumptions. We are left to prove that 
\[\numer = 0.\]
As in the proof of Lemma \ref{lem:numer_weak_deterministic},  we assume without loss of generality that $\tau^{*}(\xi; \zeta) = 0$; otherwise, we replace $Y_{it}(1)$ by $Y_{it}(1) - \tau^{*}(\xi; \zeta)$ and the resulting $\hat{\tau}$ becomes $\hat{\tau} - \tau^{*}(\xi; \zeta)$. Note that this reduction relies on the fact that $\sum_{i=1}^{n}\zeta_i = 1$. Using the same argument, we can show that \eqref{eq:numer_bias} continues to hold with the new definition of $\Theta_i$ and other related quantities, i.e., 
\begin{align*}
    \numer
    & = \frac{1}{n}\sum_{i=1}^{n}\left\{\E[\Theta_i\J\W_i]\E[\stheta] - \E[\Theta_i]\E[\sw]\right\}^\tran\E[\td{\Y}_i(0)]\nonumber\\
    & \quad + \frac{1}{n}\sum_{i=1}^{n}\left\{\E[\Theta_i \W_i^\tran\J\diag(\W_i)]\E[\stheta] - \E[\sw]^\tran\E[\Theta_i\diag(\W_i)]\right\}\td{\bftau}_i.
    \end{align*}
    By \eqref{eq:change_of_measure}, 
  \[\E[\Theta_i \J\W_i] = n\zeta_i\EPi[\J\W], \quad \E[\Theta_i] = n\zeta_i,\]
  and 
  \[\E[\Theta_i \W_i^\tran\J\diag(\W_i)] = n\zeta_i\EPi[\W\J\diag(\W)], \quad \E[\Theta_i\diag(\W_i)] = n\zeta_i\EPi[\diag(\W)].\]
  Thus,
  \[\stheta = 1, \quad \sw = \EPi[\J\W].\]
  Then,
  \[\E[\Theta_i\J\W_i]\E[\stheta] - \E[\Theta_i]\E[\sw] = n\zeta_i\EPi[\J \W] - n\zeta_i\EPi[\J \W] = 0,\]
  and by \DATEeq,
  \begin{align*}
    &\E[\Theta_i \W_i^\tran\J\diag(\W_i)]\E[\stheta] - \E[\sw]^\tran\E[\Theta_i\diag(\W_i)]\\
    &= n\zeta_i\EPi[(\W - \EPi[\W])^\tran \J \diag(\W)]\\
    & = n\zeta_i\EPi[(\W - \EPi[\W])^\tran \J \W]\xi^\tran.
  \end{align*}
  Since we assume $\tau^{*} = 0$, 
  \begin{align*}
    \numer
    & = \sum_{i=1}^{n}\zeta_i\EPi[(\W - \EPi[\W])^\tran \J \W]\xi^\tran\td{\bftau}_i\\
    & = \EPi[(\W - \EPi[\W])^\tran \J \W]\lb\sum_{i=1}^{n}\zeta_i\xi^\tran\td{\bftau}_i\rb\\
    & = \EPi[(\W - \EPi[\W])^\tran \J \W]\tau^* = 0.
  \end{align*}

 

\subsection{Proof of Theorem \ref{thm:dynamic_consistency}}\label{subapp:dynamic}
Let $\W_{i, \ex} = (\W_{i, 0}, \W_{i, -1}, \ldots, \W_{i, -p})\in \R^{T\times (p+1)}$, where $\W_{i, 0} = \W_{i}$, and $\hat{\bftau}_{\ex} = (\hat{\tau}_0, \hat{\tau}_{-1}. \ldots, \hat{\tau}_{-p})$ where the entries are defined in \eqref{eq:RIPW_event}. To derive the non-stochastic formula of $\hat{\bftau}_{\ex}$, we can repeat the steps in the proof of Theorem \ref{theorem:estimator} by replacing $\W_i$ with $\W_{i, \ex}$ in $\sww, \sw$, and $\swy$. This results in
\[\hat{\bftau}_{\ex} = \left\{\swwex - \frac{\swex^\tran\swex}{\stheta}\right\}^{-1}\left\{\swyex - \frac{\swex^\tran \sy}{\stheta}\right\},\]
where 
\[\swwex = \frac{1}{n}\sum_{i=1}^{n}\Theta_i \W_{i, \ex}^\tran \J \W_{i, \ex}\in \R^{(p+1)\times (p+1)}, \quad \swex \triangleq \frac{1}{n}\sum_{i=1}^{n}\Theta_{i}\J\W_{i, \ex}\in \R^{T\times (p+1)},\]
and 
\[\swyex = \frac{1}{n}\sum_{i=1}^{n}\Theta_i \W_{i, \ex}^\tran \J \Y_{i}\in \R^{(p+1)\times T}, \quad \sy = \frac{1}{n}\sum_{i=1}^{n}\Theta_i \J \Y_{i}\in \R^{T}.\]
Note that $\Y_i = \td{\Y}_i$ in this case since no regression adjustment is applied. 

Following the same steps as in Appendix \ref{subsubapp:asymptotic_linear}, we can prove that 
\begin{equation}\label{eq:dynamic_limit}
\hat{\bftau}_{\ex}\stackrel{p}{\rightarrow} \left\{\E[\swwex] - \frac{\E[\swex]^\tran\E[\swex]}{\E[\stheta]}\right\}^{-1}\left\{\E[\swyex] - \frac{\E[\swex]^\tran \E[\sy]}{\E[\stheta]}\right\}.
\end{equation}
Since $\hat{\bfpi}_i = \bfpi_i$, by \eqref{eq:change_of_measure}, 
\[\E[\stheta] = 1, \quad \E[\swwex] = \EPi[\W_{ \ex}^\tran \J \W_{\ex}], \quad \E[\swex] = \EPi[\J \W_{\ex}].\]
Thus, 
\begin{equation}\label{eq:dynamic_denom}
\E[\swwex] - \frac{\E[\swex]^\tran\E[\swex]}{\E[\stheta]} = \EPi[(\W_{ \ex} - \EPi[\W_{\ex}])^\tran \J (\W_{ \ex} - \EPi[\W_{\ex}])]
\end{equation}
By assumption \eqref{as:dynamic_effect},
\[\Y_i = \Y_i(\zero_{p+1}) + \sum_{\ell = 0}^{p}\W_{i, -\ell}\tau_{i, -\ell} = \Y_i(\zero_{p+1}) + \W_{i, \ex}\bftau_{i, \ex},\]
where 
\[\bftau_{i, \ex} = (\tau_{i, 0}, \tau_{i, -1}, \ldots, \tau_{i, -p})\in \R^{p+1}.\]
Then
\begin{align*}
\E[\swyex] &= \frac{1}{n}\sum_{i=1}^{n}\E[\Theta_i \W_{i, \ex}^\tran \J \Y_i] = \frac{1}{n}\sum_{i=1}^{n}\left\{\E[\Theta_i \W_{i, \ex}^\tran \J \Y_i(\zero_{p+1})] + \E[\Theta_i \W_{i, \ex}^\tran \J \W_{i, \ex}]\bftau_{i, \ex}\right\}\\
& = \frac{1}{n}\sum_{i=1}^{n}\left\{\E[\Theta_i \W_{i, \ex}^\tran \J] \E[\Y_i(\zero_{p+1})] + \E[\Theta_i \W_{i, \ex}^\tran \J \W_{i, \ex}]\bftau_{i, \ex}\right\}\\
& = \EPi[\W_{\ex}^\tran \J]\left\{\frac{1}{n}\sum_{i=1}^{n}\E[\J \Y_i(\zero_{p+1})]\right\} + \EPi[\W_{\ex}^\tran \J \W_{\ex}]\left\{\frac{1}{n}\sum_{i=1}^{n}\bftau_{i, \ex}\right\},
\end{align*}
where the second last line uses the assumption that $\bfpi_{i}(\w) = \P(\W_i = \w\mid \Y_i(\zero_{p+1}))$ and the last line is a result of \eqref{eq:change_of_measure} and the fact that $\J^2 = \J$. Similarly, 
\[\E[\sy] = \frac{1}{n}\sum_{i=1}^{n}\left\{ \E[\J\Y_i(\zero_{p+1})] + \EPi[\J \W_{\ex}]\bftau_{i, \ex}\right\}.\]
Thus, 
\begin{equation}\label{eq:dynamic_numer}
\E[\swyex] - \frac{\E[\swex]^\tran \E[\sy]}{\E[\stheta]} = \EPi[(\W_{ \ex} - \EPi[\W_{\ex}])^\tran \J (\W_{ \ex} - \EPi[\W_{\ex}])]\left\{\frac{1}{n}\sum_{i=1}^{n}\bftau_{i, \ex}\right\}.
\end{equation}
Combining \eqref{eq:dynamic_limit}, \eqref{eq:dynamic_denom}, and \eqref{eq:dynamic_numer} together, we obtain that
\[\hat{\bftau}_{\ex} \stackrel{p}{\rightarrow} \frac{1}{n}\sum_{i=1}^{n}\bftau_{i, \ex}.\]

\subsection{Miscellaneous}
\begin{proposition}\label{prop:Berry_Esseen}[\cite{petrov1975independent}, p.~112, Theorem 5]
Let $X_1, X_2, \ldots, X_{n}$ be independent random variables such that $\E[X_j] = 0$, for all $j$. Assume also $\E[X_j^2 g(X_j)] < \infty$ for some function $g$ that is non-negative, even, and non-decreasing in the interval $x > 0$, with $x / g(x)$ being non-decreasing for $x > 0$. Write $B_{n} = \sum_{j}\Var[X_j]$. Then,
\[d_{K}\left(\mathcal{L}\left(\frac{1}{\sqrt{B_{n}}}\sum_{j=1}^{n}X_j\right), N(0, 1)\right)\le \frac{A}{B_n g(\sqrt{B_n})}\sum_{j=1}^{n}\E\left[X_j^2 g(X_j)\right],\]
where $A$ is a universal constant, $\mathcal{L}(\cdot)$ denotes the probability law, $d_{K}$ denotes the Kolmogorov-Smirnov distance (i.e., the $\ell_{\infty}$-norm of the difference of CDFs)
\end{proposition}


\begin{proposition}[Theorem 2 of \cite{vonbahr65}]\label{prop:vonbahresseen}
 Let $\{Z_{i}\}_{i = 1, \ldots, n}$ be independent mean-zero random variables. Then for any $a \in [0, 1)$,
\[\E \bigg|\sum_{i=1}^{n}Z_{i}\bigg|^{1 + a}\le 2\sum_{i=1}^{n}\E |Z_{i}|^{1 + a}.\]
\end{proposition}

\section{Inference with cross-fitted model estimates}\label{app:cross-fitting}

\subsection{Cross-fitted RIPW estimator and main result}
We split the data into $K$ almost equal-sized folds with $\I_{k}$ denoting the index sets of the $k$-th fold and $|\I_{k}| \in \{\lfloor n / K\rfloor, \lceil n / K\rceil\}$. For each $i\in \I_k$, we estimate $(\hat{\bfpi}_i, \hat{\bfm}_i, \hbftau_i)$ using $\{(\Y_i(1), \Y_i(0), \W_i): i\not\in \I_k\}$. When $\{(\Y_i(1), \Y_i(0), \W_i): i \in [n]\}$ are independent, it is obvious that
\begin{equation*}
  \{(\hat{\bfpi}_i, \hat{\bfm}_i, \hbftau_i): i\in \I_k\}\indep \{(\Y_i(1), \Y_i(0), \W_i): i\in \I_k\}.
\end{equation*}
We assume that 
\[\frac{1}{T}\sum_{t=1}^{T}\hat{m}_{it} = \frac{1}{|\mathcal{I}_k|}\sum_{i\in \mathcal{I}_k}\hat{m}_{it} = 0, \quad \forall i \in \mathcal{I}_k,\quad t = 1,\ldots, T,\]
and 
\begin{equation}\label{eq:avg_hattau}
\frac{1}{|\mathcal{I}_k|}\sum_{i\in \mathcal{I}_k}\xi'\hbftau_i = 0.
\end{equation}
Otherwise, we apply the transformation \eqref{eq:hatm} and \eqref{eq:hattau} in each fold to enforce the above.

For valid inference, we need an additional assumption on the stability of the estimates. 
\begin{assumption}\label{as:asymptotic_deterministic}
There exist functions $\{\bfpi'_i: i \in [n]\}$ which satisfy Assumption \ref{as:overlap_general}  and vectors $\{(\bfm'_i, \nhbftau'_i): i \in [n]\}$ which satisfy Assumption \ref{as:regularity_general}, such that
\begin{equation}\label{eq:as_asymptotic_deterministic}
  \frac{1}{n}\sum_{i=1}^{n}\left\{\E[(\hat{\bfpi}_i(\W_i) - \bfpi'_i(\W_i))^2] + \E[\|\hat{\bfm}_i - \bfm'_i\|_{2}^2] + \E[\|\hbftau_i - \nhbftau'_i\|_{2}^2]\right\} = O(n^{-\r})
  \end{equation}
  for some $\r > 0$. Furthermore, 
  \begin{equation}\label{eq:as_limit}
  \bfpi'_i = \bfpi_i \text{ for all }i, \quad\text{or}\quad (\bfm'_i, \nhbftau'_i) = (\bfm_i, \nhbftau_i) \text{ for all }i.
  \end{equation}
\end{assumption}
The condition \eqref{eq:as_asymptotic_deterministic} states that the estimates need to be asymptotically deterministic given the covariates. This is a very mild assumption. For example, when $\hat{\bfpi}_i$ is estimated from a parametric model $\{f(\X_i; \bftheta): \bftheta\in \R^{d}\}$ as $f(\X_i; \hat{\bftheta})$, under standard regularity conditions, $\hat{\bftheta}$ converges to a limit $\bftheta_0$ even if the model is misspecified. As a result, $\hat{\bfpi}_i$ converges to $\bfpi'_i = f(\X_i; \bftheta_0)$. Under certain smoothness assumption, the estimates converge in the standard parametric rate and thus \eqref{eq:as_asymptotic_deterministic} holds with $\r = 1$. On the other hand, in the settings of Section \ref{sec:design_based}, \eqref{eq:as_asymptotic_deterministic} is always satisfied with $\bfpi'_i = \bfpi_i$ and $\bfm_i' = \nhbftau_i = \zero_T$. More generally, if $\bar{\delta}_\pi^2 + \bar{\delta}_{y}^2 = O(n^{-\r})$, it is also satisfied with $\bfpi'_i = \bfpi_i$ and $(\bfm_i', \nhbftau'_i) = (\bfm_i, \nhbftau_i)$. A similar assumption was considered for cross-sectional data by \cite{chernozhukov2020adversarial}.

The condition \eqref{eq:as_limit} allows one of the treatment and outcome models to be inconsistently estimated. This covers the settings in Section \ref{sec:design_based} where the outcome model does not need to be consistently estimated. It also covers the classical model-based inference in which case the assignment model can be arbitrarily misspecified. 

\begin{theorem}\label{thm:dr_cf_coverage}
  Assume that $\{(\Y_i(1), \Y_i(0), \W_i): i\in [n]\}$ are independent. Let $\{(\hat{\bfpi}_i, \hat{\bfm}_i, \hbftau_i): i \in [n]\}$ be estimates obtained from $K$-fold cross-fitting where $K = O(1)$. Under Assumptions \ref{as:latent_ign_app}, \ref{as:overlap_general}, \ref{as:independent_general}, and \ref{as:asymptotic_deterministic},
  \begin{enumerate}[(i)]
  \item $\hat{\tau}(\bfPi) - \tau^{*}(\xi) = o_\P(1)$ if $\bar{\delta}_{\pi}\bar{\delta}_{y} = o(1)$;
  \item Let $\hat{C}_{1-\alpha}$ be the same confidence interval as in Theorem \ref{thm:dr_coverage}. Then \[\liminf_{n\rightarrow \infty}\P\lb \tau^{*}(\xi)\in \hat{C}_{1 - \alpha}\rb\ge 1 - \alpha\] 
  if (a) $\bar{\delta}_{\pi}\bar{\delta}_{y} = o(1/\sqrt{n})$, (b) Assumption \ref{as:asymptotic_deterministic} holds with $\r > 1 / 2$, and (c) \eqref{eq:variance_lower} holds if 
  $(\Theta_i, \Y_i)$ are replaced by $(\bfPi(\W_i) / \bfpi_i'(\W_i), \Y_i - \bfm'_{i} - \diag(\W_i)\nhbftau'_i)$
  in the definition of $\V_i$.
  \end{enumerate}  
\end{theorem}

The proof of Theorem \ref{thm:dr_cf_coverage} is quite involved because our cross-fitted estimator is non-standard. The standard cross-fitting \citep{chernozhukov2017double} would compute $\hat{\tau}_{k}(\bfPi)$ on $\I_k$ with $\{(\hat{\bfpi}_i, \hat{\bfm}_i, \hbftau_i): i\in \I_k\}$ and then take $\hat{\tau}(\bfPi)$ as the average of $\{\hat{\tau}_{k}(\bfPi): k\in [K]\}$. Under the assumptions of Theorem \ref{thm:dr_linear_expansion}, it is straightforward to show each $\hat{\tau}_{k}(\bfPi)$ is asymptotically linear and hence their average $\hat{\tau}(\bfPi)$. In contrast, our estimator only cross-fitted the nuisance parameters $\{(\hat{\bfpi}_i, \hat{\bfm}_i, \hbftau_i): i\in [n]\}$ but compute $\hat{\tau}(\bfPi)$ using the whole dataset. 
While it is theoretically convenient to deal with the standard cross-fitting estimator, the standard version would fit weighted TWFE regressions on merely $n/K$ units which would cause instability when $n$ is moderate as in many economic applications. For this reason, we opt for our version to max out the sample size for computing $\hat{\tau}(\bfPi)$, even though the technical proofs are lengthier. 

\subsection{De-randomization}\label{subapp:derandomization}
Cross-fitting involves random data splits which introduce operational variation into the final estimate. We propose a de-randomization procedure that mitigates this source of unnecessary uncertainty by averaging over multiple splits. In particular, we consider $B$ independent splits and add a superscript $(b)$ to denote the quantities involved in the $b$-th split.  

In the proof presented in the next subsection, we will show in \eqref{eq:cf_asymptotic_linear} that, for each given data split, 
\[\mathcal{D}^{(b)}\cdot \sqrt{n}(\hat{\tau}^{(b)} - \tau^{*}) = \frac{1}{\sqrt{n}}\sum_{i=1}^{n}(\V'_i - \E[\V'_i]) + o_\P(1),\]
Note that $\V'_i$ does not depend on $b$. We define the de-randomized cross-fitted RIPW estimate as 
\begin{equation}\label{eq:derandomized_cf}
\hat{\tau} =\frac{\sum_{b=1}^{B}\mathcal{D}^{(b)}\hat{\tau}^{(b)}}{\sum_{b=1}^{B}\mathcal{D}^{(b)}}.
\end{equation}
Then, when $B = O(1)$, 
\[\lb\frac{1}{B}\sum_{b=1}^{B}\mathcal{D}^{(b)}\rb\cdot \sqrt{n}(\hat{\tau} - \tau^{*}) = \frac{1}{\sqrt{n}}\sum_{i=1}^{n}(\V'_i - \E[\V'_i]) + o_\P(1).\]
Furthermore, in Section \ref{subsubapp:variance} we show that 
\[\frac{1}{n}\sum_{i=1}^{n}(\V'_i - \hat{\V}_i^{(b)})^2 = o_\P(1).\]
Denote by $\bar{\hat{\V}}_i$ the average influence function:
\[\bar{\hat{\V}}_i = \frac{1}{B}\sum_{b=1}^{B}\hat{\V}_i^{(b)}.\]
Then, 
\[\frac{1}{n}\sum_{i=1}^{n}(\V'_i - \bar{\hat{\V}}_i)^2 = o_\P(1).\]
Therefore, we can estimate the variance of $\hat{\tau}$ by the sample variance of $\bar{\hat{\V}}_1, \ldots, \bar{\hat{\V}}_n$. This justifies the confidence interval stated in Algorithm \ref{algo:RIPW_derandomized_cf}.

\subsection{Proof of Theorem \ref{thm:dr_cf_coverage}}
For convenience, we assume that $m = n / K$ is an integer. All proofs in this subsection can be easily extended to the general case. Without loss of generality, we can assume that
  \begin{equation}
    \label{eq:piimi'_WLOG}
    \E[\bar{\Delta}_{\pi}^2]= \bar{\delta}_{\pi}^2 = \Omega(n^{-\r}), \quad \E[\bar{\Delta}_{y}^2]=\bar{\delta}_{y}^2 = \Omega(n^{-\r}), 
  \end{equation}
where $a_{n} = \Omega(b_{n})$ iff $b_{n} = O(a_n)$. 
Otherwise, we can replace $(\bfpi'_i, \bfm'_i, \nhbftau_i')$ by $(\bfpi_i, \bfm_i, \nhbftau_i)$ without decreasing $\r$.

We use a superscript $(k)$ to denote the corresponding quantity in fold $k$, i.e.,
\begin{equation*}
  \stheta^{(k)}\triangleq \frac{1}{m}\sum_{i\in \I_k}\Theta_{i}, \quad \sww^{(k)} \triangleq \frac{1}{m}\sum_{i\in \I_k}\Theta_{i}\W_{i}^{\tran}\J\W_{i}, \quad \swy^{(k)} \triangleq \frac{1}{m}\sum_{i\in \I_k}\Theta_{i}\W_{i}^{\tran}\J\td{\Y}_i,
\end{equation*}
\begin{equation*}
\sw^{(k)} \triangleq \frac{1}{m}\sum_{i\in \I_k}\Theta_{i}\J\W_{i}, \quad \sy^{(k)} \triangleq \frac{1}{m}\sum_{i\in \I_k}\Theta_{i}\J\td{\Y}_i.
\end{equation*}
  As in the proof of Theorem \ref{thm:numer_strong_deterministic}, we assume $\tau^{*} = 0$ without loss of generality. Let $\lb \swy', \stheta', \sw', \sy'\rb$ and $(\Theta'_i, \td{\Y}'_i, \td{\bftau}'_i)$ be the counterpart of $(\swy, \stheta, \sw, \sy)$ and $(\Theta_i, \td{\Y}_i, \td{\bftau}_i)$ with $(\hat{\bfpi}_i, \hat{\bfm}_i, \hbftau_i)$ replaced by $(\bfpi'_i, \bfm'_i, \nhbftau'_i)$. We first claim that 
  \begin{equation}
    \label{eq:cf_goal1}
    \swy\stheta - \sw^\tran\sy - \left\{\swy'\stheta' - \sw^{'\tran}\sy'\right\} = O_\P\lb n^{-\min\{\r, (\r' + 1) / 2\}} + \sqrt{\E[\bar{\Delta}_{\pi}^2]}\cdot \sqrt{\E[\bar{\Delta}_{y}^2]}\rb,
  \end{equation}
  where $\r' = \r \omega / (2 + \omega)$. The proof of \eqref{eq:cf_goal1} is relegated to the end. Here we prove the rest of the theorem under \eqref{eq:cf_goal1}.

  ~\\
  Note that $\swy'\stheta' - \sw'^\tran\sy'$ is the numerator of $\hat{\tau}$ when $\{(\bfpi'_i, \bfm'_i, \nhbftau'_i): i = 1, \ldots, n\}$ are used as the estimates. Let
  \begin{equation}
    \label{eq:deltai'}
    \delta'_{\pi i} = \sqrt{\E[(\bfpi_i'(\W_i) - \bfpi_i(\W_i))^2]}, \quad \delta'_{y i} = \sqrt{\E[\|\bfm'_i - \bfm_i\|_{2}^{2}] + \E[\|\nhbftau'_i - \nhbftau_i\|_{2}^2]},
  \end{equation}
  and
  \begin{equation}
    \label{eq:bardelta'}
    \bar{\delta}'_{\pi} = \sqrt{\frac{1}{n}\sum_{i=1}^{n}\delta_{\pi i}^{'2}}, \quad \bar{\delta}'_{y} = \sqrt{\frac{1}{n}\sum_{i=1}^{n}\delta_{y i}^{'2}}.
  \end{equation}
  By Assumption \ref{as:asymptotic_deterministic} and \eqref{eq:piimi'_WLOG},
    \begin{align}
      \bar{\delta}_{\pi}^{'2} &= \frac{1}{n}\sum_{i=1}^{n}\delta_{\pi i}^{'2} \nonumber \\
      & \le \frac{2}{n}\sum_{i=1}^{n}\left\{ \E[\Delta_{\pi i}^2] + \E[(\hat{\bfpi}_i(\W_i) - \bfpi'_i(\W_i))^2]\right\}\nonumber\\
      & = O\lb\E[\bar{\Delta}_{\pi}^2] + n^{-r}\rb = O\lb\E[\bar{\Delta}_{\pi}^2]\rb. \label{eq:E_bardelta_pi'}
    \end{align}
    Similarly,
    \begin{equation}
      \bar{\delta}_{y}^{'2} = \frac{1}{n}\sum_{i=1}^{n}\delta_{y i}^{'2} = O\lb\E[\bar{\Delta}_{y}^2]\rb. \label{eq:E_bardelta_y'}
    \end{equation}
    As a result,
    \[\bar{\delta}'_{\pi}\bar{\delta}'_{y} = O\lb\sqrt{\E[\bar{\Delta}_{\pi}^2]}\cdot \sqrt{\E[\bar{\Delta}_{y}^2]}\rb.\]
    Note that Assumption \ref{as:independent_general} implies Assumption \ref{as:regularity_general} with $q = 1$. By Theorem \ref{thm:asymptotic_expansion_simple} and Theorem \ref{thm:numer_strong_deterministic},
    \begin{align}
      \swy'\stheta' - \sw^{'\tran}\sy' &= \frac{1}{n}\sum_{i=1}^{n}(\V'_i - \E[\V'_i]) + O_\P\lb\sqrt{\E[\bar{\Delta}_{\pi}^2]}\cdot \sqrt{\E[\bar{\Delta}_{y}^2]}\rb + o_\P(1/\sqrt{n})\label{eq:cf_expansion}\\
      & = O_\P\lb\sqrt{\E[\bar{\Delta}_{\pi}^2]}\cdot \sqrt{\E[\bar{\Delta}_{y}^2]}\rb + o_\P(1)\label{eq:cf_consistency},
    \end{align}
    where
    \begin{align*}
      \V'_i &= \Theta'_i\bigg\{\E[\swy'] - \E[\sy']^\tran\J \W_i + \E[\stheta']\W_i^\tran \J \td{\Y}'_i - \E[\sw']^\tran\J\td{\Y}'_i \bigg\}.
    \end{align*}
    On the other hand, by \eqref{eq:cf_goal1},
    \begin{equation}
      \label{eq:cf_approx_expansion}
      \hatdenom(\hat{\tau} - \tau^{*}) = \swy\stheta - \sw^\tran\sy = \swy'\stheta' - \sw^{'\tran}\sy' + O_\P\lb n^{-\min\{\r, (\r' + 1) / 2\}} + \sqrt{\E[\bar{\Delta}_{\pi}^2]}\cdot \sqrt{\E[\bar{\Delta}_{y}^2]}\rb.
    \end{equation}
    When $\sqrt{\E[\bar{\Delta}_{\pi}^2]}\cdot \sqrt{\E[\bar{\Delta}_{y}^2]} = o(1)$, \eqref{eq:cf_consistency} and \eqref{eq:cf_approx_expansion} imply that
    \[\hatdenom(\hat{\tau} - \tau^{*}) = o_\P(1).\]
    The consistency then follows from Lemma \ref{lem:denom}.
    
    ~\\
    When $\sqrt{\E[\bar{\Delta}_{\pi}^2]}\cdot \sqrt{\E[\bar{\Delta}_{y}^2]} = o(1 / \sqrt{n})$ and $\r > 1 / 2$, \eqref{eq:cf_consistency} and \eqref{eq:cf_approx_expansion} imply that
    \begin{equation}\label{eq:cf_asymptotic_linear}
    \hatdenom\cdot \sqrt{n}(\hat{\tau} - \tau^{*}) = \frac{1}{\sqrt{n}}\sum_{i=1}^{n}(\V'_i - \E[\V'_i]) + o_\P(1).
    \end{equation}
    Let $\hV'_i$ denote the plug-in estimate of $\V'_i$ assuming that $(\bfpi'_i, \bfm'_i, \nhbftau'_i)$ is known, i.e.,
  \begin{align}\label{eq:hVi'}
    \hV'_i &= \Theta'_i\bigg\{\swy' - \sy^{'\tran}\J \W_i  + \stheta'\W_i^\tran \J \td{\Y}'_i - \sw^{'\tran}\J\td{\Y}'_i \bigg\}.
  \end{align}
  By Lemma \ref{lem:CLT}, under Assumption \ref{as:var_low} (with $(\hat{\bfpi}_i, \hat{\bfm}_i, \hbftau_i) = (\bfpi'_i, \bfm'_i, \nhbftau'_i)$),
  \[\frac{\hatdenom\cdot \sqrt{n}(\hat{\tau} - \tau^{*})}{\sigma'} \dcv N(0, 1) \text{ in Kolmogorov-Smirnov distance},\]
  where 
  \[\sigma^{'2} = \frac{1}{n}\sum_{i=1}^{n}\Var(\V'_i)\ge v_0.\]
  Similar to \eqref{eq:sigma2+}, define
  \[\sigma_{+}^{'2} = \frac{1}{n}\sum_{i=1}^{n}\E\lb \V'_i - \frac{1}{n}\sum_{i=1}^{n}\E[\V'_i]\rb^2 = \frac{1}{n}\sum_{i=1}^{n}\E[\V_i^{'2}] - \lb\frac{1}{n}\sum_{i=1}^{n}\E[\V'_i]\rb^2.\]
  Obviously, $\sigma_{+}^{'2}\ge \sigma^{'2}$. Furthermore, define an oracle variance estimate $\hat{\sigma}^{'2}$ as
  \[\hat{\sigma}^{'2} = \frac{1}{n-1}\sum_{i=1}^{n}\lb \hV'_i - \frac{1}{n}\sum_{i=1}^{n}\hV'_i\rb^2 = \frac{n}{n-1}\left\{\frac{1}{n}\sum_{i=1}^{n}\hV_i^{'2} - \lb\frac{1}{n}\sum_{i=1}^{n}\hV'_i\rb^2\right\}.\]
  Recalling \eqref{eq:hatsigma2} that
  \[\hat{\sigma}^2 = \frac{1}{n-1}\sum_{i=1}^{n}\lb \hV_i - \frac{1}{n}\sum_{i=1}^{n}\hV_i\rb^2 = \frac{n}{n-1}\left\{\frac{1}{n}\sum_{i=1}^{n}\hV_i^2 - \lb\frac{1}{n}\sum_{i=1}^{n}\hV_i\rb^2\right\}.
\]
  Similar to \eqref{eq:goal_sigma2} in Theorem \ref{thm:doubly_robust_inference_deterministic}, it remains to prove that
  \[|\hat{\sigma}^2 - \sigma_{+}^{'2}| = o_\P(1).\]
  Using the same arguments as in Theorem \ref{thm:doubly_robust_inference_deterministic}, we can prove that
  \[|\hat{\sigma}^{'2} - \sigma_{+}^{'2}| = o_\P(1).\]
  Therefore, the proof will be completed if 
  \begin{equation}
    \label{eq:cf_goal2}
    |\hat{\sigma}^2 - \hat{\sigma}^{'2}| = o_\P(1).
  \end{equation}
  We present the proof of \eqref{eq:cf_goal2} in the end.
    
  \subsubsection{Proof of \eqref{eq:cf_goal1}}
  Let $\lb\swy^{'(k)}, \stheta^{'(k)}, \sw^{'(k)}, \sy^{'(k)}\rb$ be the counterpart of $(\swy^{(k)}, \stheta^{(k)}, \sw^{(k)}, \sy^{(k)})$ with $(\hat{\bfpi}_i, \hat{\bfm}_i, \hbftau_i)$ replaced by $(\bfpi'_i, \bfm'_i, \nhbftau'_i)$. Since the proof is lengthy, we decompose it into seven steps. 
  \begin{enumerate}[\textbf{Step} 1]
  \item By triangle inequality and Cauchy-Schwarz inequality,
    \begin{align*}
      &|\swy - \swy'|\\
      &\le \frac{1}{n}\sum_{i=1}^{n}|\Theta_i \W_i^\tran\J \td{\Y}_i - \Theta'_i \W_i^\tran\J \td{\Y}'_i|\\
      & \le \frac{1}{n}\sum_{i=1}^{n}\left|\Theta_i\W_i^\tran\J(\hat{\bfm}_i - \bfm'_i)\right| + \frac{1}{n}\sum_{i=1}^{n}\left|\Theta_i\W_i^\tran\J\diag(\W_i)(\hbftau_i - \nhbftau'_i)\right|\\
      & \quad + \frac{1}{n}\sum_{i=1}^{n}\left|(\Theta_i - \Theta'_i)\W_i^\tran\J \td{\Y}'_i\right|\\
      & \le \sqrt{\lb\frac{1}{n}\sum_{i=1}^{n}\|\Theta_i\W_i^\tran\J\|_{2}^2\rb\lb\frac{1}{n}\sum_{i=1}^{n}\|\hat{\bfm}_i - \bfm'_i\|_{2}^2\rb}\\
      & \quad + \sqrt{\lb\frac{1}{n}\sum_{i=1}^{n}\|\Theta_i\W_i^\tran\J\diag(\W_i)\|_{2}^2\rb\lb\frac{1}{n}\sum_{i=1}^{n}\|\hbftau_i - \nhbftau'_i\|_{2}^2\rb}\\
      & \quad + \sqrt{\lb\frac{1}{n}\sum_{i=1}^{n}\|(\Theta_i - \Theta'_i)\W_i^\tran\J\|_{2}^2\rb\lb\frac{1}{n}\sum_{i=1}^{n}\|\td{\Y}'_i\|_{2}^2\rb}.
    \end{align*}
    By Assumption \ref{as:var_low} and H\"{o}lder's inequality,
    \[\frac{1}{n}\sum_{i=1}^{n}\E[\|\td{\Y}'_i\|_{2}^2]\le \lb\frac{1}{n}\sum_{i=1}^{n}\E[\|\td{\Y}'_i\|_{2}^{2 + \omega}]\rb^{2 / (2 + \omega)} = O(1).\]
    By Markov's inequality,
    \begin{equation}
      \label{eq:Yi_2norm}
      \frac{1}{n}\sum_{i=1}^{n}\|\td{\Y}'_i\|_{2}^2 = O_\P(1).
    \end{equation}
    By Assumption \ref{as:overlap_general} and the boundedness of $\|\W_i\J\|_{2}$ and $\|\W_i \J \diag(\W_i)\|_{2}$,
    \[\frac{1}{n}\sum_{i=1}^{n}\|\Theta_i\W_i^\tran\J\|_{2}^2 = O(1), \quad \frac{1}{n}\sum_{i=1}^{n}\|\Theta_i\W_i^\tran\J\diag(\W_i)\|_{2}^2 = O(1),\]
    and, further, by Markov's inequality,
    \[\frac{1}{n}\sum_{i=1}^{n}\|(\Theta_i - \Theta'_i)\W_i^\tran\J\|_{2}^2 = O_\P\lb\frac{1}{n}\sum_{i=1}^{n}\E[(\hat{\bfpi}_i(\W_i) - \bfpi'_i(\W_i))^2]\rb.\]
    Putting pieces together and using Assumption \ref{as:asymptotic_deterministic}, we arrive at
    \[|\swy - \swy'| = O_\P(n^{-\r / 2}).\]
    Similarly, we can prove that
    \begin{equation}
      \label{eq:cf_goal1_step1_1}
      |\swy - \swy'| + |\stheta - \stheta'| + \|\sw - \sw'\|_{2} + \|\sy - \sy'\|_{2} = O_\P(n^{-\r / 2}).
    \end{equation}
    As a consequence,
    \begin{equation}
      \label{eq:cf_goal1_1}
      \left|(\swy - \swy')(\stheta - \stheta') - (\sw - \sw')^\tran (\sy - \sy')\right| = O_\P(n^{-\r}).
    \end{equation}
  \item Note that Assumption \ref{as:independent_general} implies Assumption \ref{as:regularity_general} with $\q = 1$. By Lemma \ref{lem:mean_var},
    \[\big|\stheta' - \E[\stheta']\big| + \big|\swy' - \E[\swy']\big| + \big\|\sw' - \E[\sw']\big\|_{2} + \big\|\sy' - \E[\sy']\big\|_{2} = O_{\P}\lb n^{-1/2}\rb.\]
    By \eqref{eq:cf_goal1_step1_1}, we have
    \begin{align}
      &\bigg|(\swy - \swy')(\stheta' - \E[\stheta']) + (\swy' - \E[\swy'])(\stheta - \stheta')\nonumber\\
      & - (\sw - \sw')^\tran (\sy' - \E[\sy']) - (\sw' - \E[\sw'])^\tran (\sy - \sy')\bigg| = O_\P(n^{-(\r + 1) / 2}).       \label{eq:cf_goal1_2}
    \end{align}
  \item Note that
    \[\swy - \swy' = \frac{1}{K}\sum_{k=1}^{K}\lb\swy^{(k)} -  \swy^{'(k)}\rb.\]
    For each $k$,
    \[\swy^{(k)} -  \swy^{'(k)} = \frac{1}{m}\sum_{i\in \I_k}(\Theta_i \W_i^\tran \J \td{\Y}_i - \Theta'_i \W_i^\tran \J \td{\Y}'_i).\]
    Under Assumption \ref{as:independent_general}, the summands are independent conditional on $\cD_{-[k]}\triangleq \{(\Y_i(1), \Y_i(0), \X_i): i\not\in \I_k\}$. Let $\E^{(k)}$ and $\Var^{(k)}$ denote the expectation and variance conditional on $\cD_{-[k]}$. By Chebyshev's inequality,
    \begin{align}
      &\lb\swy^{(k)} -  \swy^{'(k)} - \E^{(k)}[\swy^{(k)} -  \swy^{'(k)}]\rb^2\nonumber\\
      &= O_\P\lb\frac{1}{m^2}\sum_{i\in \I_k}\Var^{(k)}\lb\Theta_i \W_i^\tran \J \td{\Y}_i - \Theta'_i \W_i^\tran \J \td{\Y}'_i\rb\rb\nonumber\\
      & \stackrel{(i)}{=} O_\P\lb\frac{1}{n^2}\sum_{i=1}^{n}\E^{(k)}\lb\Theta_i \W_i^\tran \J \td{\Y}_i - \Theta'_i \W_i^\tran \J \td{\Y}'_i\rb^2\rb\nonumber\\
      & \stackrel{(ii)}{=} O_\P\lb\frac{1}{n^2}\sum_{i=1}^{n}\E\lb\Theta_i \W_i^\tran \J \td{\Y}_i - \Theta'_i \W_i^\tran \J \td{\Y}'_i\rb^2\rb,\label{eq:cf_goal1_step3_1}
    \end{align}
    where (i) follows from $K = O(1)$ and (ii) applies Markov's inequality.
    By Jensen's inequality and Cauchy-Schwarz inequality,
    \begin{align}
      &\E\lb\Theta_i \W_i^\tran \J \td{\Y}_i - \Theta'_i \W_i^\tran \J \td{\Y}'_i\rb^2\nonumber\\
      & \le 3\Big\{\E\lb\Theta_i\W_i^\tran\J(\hat{\bfm}_i - \bfm'_i)\rb^2 + \E\lb\Theta_i\W_i^\tran\J\diag(\W_i)(\hbftau_i - \nhbftau'_i)\rb^2 \nonumber\\
      & \qquad + \E\lb(\Theta_i - \Theta'_i)\W_i^\tran\J \td{\Y}'_i\rb^2\Big\}\nonumber\\
      & \le 3\Big\{\E\left[\|\Theta_i\W_i^\tran\J\|_{2}^2\cdot \|\hat{\bfm}_i - \bfm'_i\|_2^2\right] + \E\left[\|\Theta_i\W_i^\tran\J\diag(\W_i)\|_2^2\cdot \|\hbftau_i - \nhbftau'_i\|_2^2\right] \nonumber\\
      &\quad + \E\left[(\Theta_i - \Theta'_i)^2(\W_i^\tran\J \td{\Y}'_i)^2\right]\Big\}\nonumber\\
      & \le C\left\{\E\left[(\hat{\bfm}_i - \bfm'_i)^2\right] + \E\left[(\hbftau_i - \nhbftau'_i)^2\right] + \E\left[(\hat{\bfpi}_i(\W_i) - \bfpi'_i(\W_i))^2(\td{\Y}'_i)^2\right]\right\},\label{eq:cf_goal1_step3_2}
    \end{align}
    where $C$ is a constant that only depends on $c_{\pi}$ and $T$. 
    The second term can be bounded by 
    \begin{align*}
      &\frac{1}{n}\sum_{i=1}^{n}\E\left[(\hat{\bfpi}_i(\W_i) - \bfpi'_i(\W_i))^2(\td{\Y}'_i)^2\right]\\
      & = \E\left[\frac{1}{n}\sum_{i=1}^{n}(\hat{\bfpi}_i(\W_i) - \bfpi'_i(\W_i))^2(\td{\Y}'_i)^2\right]\\
      & \stackrel{(i)}{\le} \E\left[\lb\frac{1}{n}\sum_{i=1}^{n}(\hat{\bfpi}_i(\W_i) - \bfpi'_i(\W_i))^{2(1 + 2 / \omega)}\rb^{\omega / (2 + \omega)}\lb\frac{1}{n}\sum_{i=1}^{n}(\td{\Y}'_i)^{2 + \omega}\rb^{2 / (2 + \omega)}\right]\\
      & \stackrel{(ii)}{\le}\lb\frac{1}{n}\sum_{i=1}^{n}\E\left[(\hat{\bfpi}_i(\W_i) - \bfpi'_i(\W_i))^{2(1 + 2 / \omega)}\right]\rb^{\omega / (2 + \omega)}\lb\frac{1}{n}\sum_{i=1}^{n}\E\left[(\td{\Y}'_i)^{2 + \omega}\right]\rb^{2 / (2 + \omega)}\\
      & \stackrel{(iii)}{\le} \lb\frac{1}{n}\sum_{i=1}^{n}\E(\hat{\bfpi}_i(\W_i) - \bfpi'_i(\W_i))^2\rb^{\omega / (2 + \omega)}\lb\frac{1}{n}\sum_{i=1}^{n}\E\left[(\td{\Y}'_i)^{2 + \omega}\right]\rb^{2 / (2 + \omega)},
    \end{align*}
    where (i) applies the H\"{o}lder's inequality for sums, (ii) applies the H\"{o}lder's inequality that $\E[XY]\le \E[X^{(2 + \omega) / \omega}]^{\omega / (2 + \omega)}\E[Y^{(2 + \omega) / 2}]^{2 / (2 + \omega)}$, and (iii) uses the fact that $|\hat{\bfpi}_i(\W_i) - \bfpi'_i(\W_i)|\le 1$. By Assumptions \ref{as:var_low} and \ref{as:asymptotic_deterministic},
    \begin{equation}
      \label{eq:cf_goal1_step3_3}
      \frac{1}{n}\sum_{i=1}^{n}\E\left[(\hat{\bfpi}_i(\W_i) - \bfpi'_i(\W_i))^2(\td{\Y}'_i)^2\right] = O\lb n^{-\r\omega / (2 + \omega)}\rb = O\lb n^{-\r'}\rb
    \end{equation}
    \eqref{eq:cf_goal1_step3_2} and \eqref{eq:cf_goal1_step3_3} together imply that
    \begin{equation}
      \label{eq:sum_square_swy}
      \frac{1}{n}\sum_{i=1}^{n}\E\lb\Theta_i \W_i^\tran \J \td{\Y}_i - \Theta'_i \W_i^\tran \J \td{\Y}'_i\rb^2 = O\lb n^{-\r'}\rb.
    \end{equation}
    By \eqref{eq:cf_goal1_step3_1}, for each $k$,
    \[\swy^{(k)} -  \swy^{'(k)} - \E^{(k)}[\swy^{(k)} -  \swy^{'(k)}] = O_\P\lb n^{-(\r' + 1) / 2}\rb.\]
    Since $K = O(1)$, it implies that
    \[\left|\swy -  \swy^{'} - \frac{1}{K}\sum_{k=1}^{K}\E^{(k)}[\swy^{(k)} -  \swy^{'(k)}]\right| = O_\P\lb n^{-(\r' + 1) / 2}\rb.\]
    Similarly, we have
    \begin{align*}
      &\left|\stheta -  \stheta^{'} - \frac{1}{K}\sum_{k=1}^{K}\E^{(k)}[\stheta^{(k)} -  \stheta^{'(k)}]\right| + \left\|\sw -  \sw^{'} - \frac{1}{K}\sum_{k=1}^{K}\E^{(k)}[\sw^{(k)} -  \sw^{'(k)}]\right\|_{2}\\
      & + \left\|\sy -  \sy^{'} - \frac{1}{K}\sum_{k=1}^{K}\E^{(k)}[\sy^{(k)} -  \sy^{'(k)}]\right\|_{2} = O_\P\lb n^{-(\r' + 1) / 2}\rb.
    \end{align*}
    By Lemma \ref{lem:mean_var},
    \[|\E[\stheta']|  +|\E[\swy']| + \|\E[\sw']\|_{2} + \|\E[\sy']\|_{2} = O(1).\]
    Therefore,
    \begin{align}
      &\bigg|\lb\swy - \swy' - \frac{1}{K}\sum_{k=1}^{K}\E^{(k)}[\swy^{(k)} - \swy^{'(k)}]\rb\E[\stheta']\nonumber\\
      & + \E[\swy']\lb\stheta - \stheta' - \frac{1}{K}\sum_{k=1}^{K}\E^{(k)}[\stheta^{(k)} - \stheta^{'(k)}]\rb\nonumber\\
      & - \lb\sw - \sw' - \frac{1}{K}\sum_{k=1}^{K}\E^{(k)}[\sw^{(k)} - \sw^{'(k)}]\rb^\tran \E[\sy']\nonumber\\
      & - \E[\sw']^\tran \lb\sy - \sy' - \frac{1}{K}\sum_{k=1}^{K}\E^{(k)}[\sy^{(k)} - \sy^{'(k)}]\rb\bigg|\nonumber\\
      & = O_\P(n^{-(\r' + 1) / 2}).\label{eq:cf_goal1_3}
    \end{align}
  \item Note that $\E[\swy']\E[\stheta'] - \E[\sw']^\tran\E[\sy']$ is the limit of $\hatdenom\cdot \sqrt{n}(\hat{\tau} - \tau^{*}) $ when $\{(\bfpi'_i, \bfm'_i, \nhbftau'_i): i = 1, \ldots, n\}$ are plugged in as the estimates. Under Assumption \ref{as:asymptotic_deterministic}, either $\bfpi'_i = \bfpi_i$ for all $i\in [n]$ or $(\bfm'_i, \nhbftau'_i) = (\bfm_i, \nhbftau_i)$ for all $i\in [n]$. Then, by Lemma \ref{lem:numer_weak_deterministic}, 
    \begin{equation}
      \label{eq:cf_goal1_4}
      \E[\swy']\E[\stheta'] - \E[\sw']^\tran\E[\sy'] = 0.
    \end{equation}
  \item We shall prove that
    \begin{align}
      \label{eq:cf_goal1_5}
      &\left|\frac{1}{K}\sum_{k=1}^{K}\E^{(k)}[\swy^{(k)}]\E[\stheta'] - \E[\sw']^\tran\E^{(k)}[\sy^{(k)}]\right|= O\lb\sqrt{\E[\bar{\Delta}_{\pi}^2]}\cdot\sqrt{\E[\bar{\Delta}_{y}^2]}\rb.
    \end{align}
By definition, we can write
    \[\Delta_{\pi i} = \sqrt{\E^{(k)}[(\hat{\bfpi}_i(\W_i) - \bfpi_i(\W_i))^2]}, \quad \Delta_{y i} = \sqrt{\E^{(k)}[\|\hat{\bfm}_i - \bfm_i\|_{2}^{2}] + \E^{(k)}[\|\hbftau_i - \nhbftau_i\|_{2}^{2}]}, \quad \forall i\in \I_k.\]
    By Assumption \ref{as:latent_ign_app} and \ref{as:independent_general},
    \begin{align*}
      \E^{(k)}[\swy^{(k)}]& = \frac{1}{m}\sum_{i\in \I_{k}}\E^{(k)}[\Theta_i\W_i^\tran \J \td{\Y}_i]\\
                          & = \frac{1}{m}\sum_{i\in \I_{k}}\E^{(k)}[\Theta_i\W_i^\tran \J \td{\Y}_i(0)] + \frac{1}{m}\sum_{i\in \I_{k}}\E^{(k)}[\Theta_i\W_i^\tran \J \diag(\W_i)\td{\bftau}_i]\\
                          & = \frac{1}{m}\sum_{i\in \I_{k}}\E^{(k)}[\Theta_i\W_i^\tran \J]\E^{(k)}[\td{\Y}_i(0)] + \frac{1}{m}\sum_{i\in \I_{k}}\E^{(k)}[\Theta_i\W_i^\tran \J \diag(\W_i)]\E^{(k)}[\td{\bftau}_i].
    \end{align*}
    Similarly,
    \begin{align*}
      \E^{(k)}[\sy^{(k)}] = \frac{1}{m}\sum_{i\in \I_{k}}\E^{(k)}[\Theta_i\J]\E^{(k)}[\td{\Y}_i(0)] + \frac{1}{m}\sum_{i\in \I_{k}}\E^{(k)}[\Theta_i\J\diag(\W_i)]\E^{(k)}[\td{\bftau}_i].
    \end{align*}
    Putting the pieces together and using the fact that $\E[\sw']^\tran \J = \E[\sw']^\tran, \E^{(k)}[\sw]^\tran \J = \E^{(k)}[\sw]^\tran$, 
    \begin{align}
      &\E^{(k)}[\swy^{(k)}]\E[\stheta'] - \E[\sw^{'}]^\tran\E^{(k)}[\sy^{(k)}]\nonumber\\
      & = \frac{1}{m}\sum_{i\in \I_k}\bigg\{\E^{(k)}[\Theta_i\W_i^\tran \J\diag(\W_i)]\E[\stheta'] - \E[\sw']^\tran\E^{(k)}[\Theta_i\J\diag(\W_i)]\bigg\}\E^{(k)}[\td{\bftau}_i].\nonumber\\
      & \,\, + \frac{1}{m}\sum_{i\in \I_k}\bigg\{\E^{(k)}[\Theta_i\W_i^\tran \J]\E[\stheta'] - \E[\sw']^\tran \E^{(k)}[\Theta_i]\bigg\}\E^{(k)}[\td{\Y}_i(0)]\nonumber\\
      & \triangleq \frac{1}{m}\sum_{i\in \I_k}\bfa_{i1}^\tran\E^{(k)}[\td{\bftau}_i] + \frac{1}{m}\sum_{i\in \I_k}\bfa_{i2}^\tran\E^{(k)}[\td{\Y}_i(0)] \label{eq:cf_goal1_step5_1}
    \end{align}
    As in the proof of Theorem \ref{thm:numer_strong_deterministic}. Let
    \[\Theta_i^* = \frac{\bfPi(\W_i)}{\bfpi_i(\W_i)}, \quad \td{\Y}_i^{*} = \Y_i - \bfm_i - \diag(\W_i)\nhbftau_i,\]
    and $(\stheta^*, \sw^*)$ be the counterpart of $(\stheta, \sw)$ with $(\Theta_i, \td{\Y}_i)$ replaced by $(\Theta_i^{*}, \td{\Y}_i^{*})$. Recalling \eqref{eq:Thetai-Thetai*} on page \pageref{eq:Thetai-Thetai*}, there exists a constant $C_1$ that only depends on $c_{\pi}$ and $T$ such that
    \begin{align}
      \label{eq:cf_goal1_step5_2}
      &\left|\E^{(k)}[(\Theta_i - \Theta_i^{*})\W_i^\tran \J\diag(\W_i)]\right| + \left\|\E^{(k)}[(\Theta_i - \Theta_i^{*})\J\W_i]\right\|_{2} + \left|\E^{(k)}[\Theta_i - \Theta_i^{*}]\right|\nonumber\\
      & \,\, + \left\|\E^{(k)}[(\Theta_i - \Theta_i^{*})\J\diag(\W_i)]\right\|_{\op} \le C_1\Delta_{\pi i}.
    \end{align}
    where $\delta'_{\pi i}$ and $\bar{\delta}'_{\pi}$ are defined in \eqref{eq:deltai'} and \eqref{eq:bardelta'}, respectively. Then,
    \begin{align*}
      &\bigg|\E^{(k)}[\Theta_i\W_i^\tran \J\diag(\W_i)]\E[\stheta'] - \lb \E[\Theta_i^{*}\W_i^\tran \J\diag(\W_i)]\E[\stheta^{*}]\rb\bigg|\\
      &\le\left|\E^{(k)}[\Theta_i\W_i^\tran \J\diag(\W_i)] - \E[\Theta_i^{*}\W_i^\tran \J\diag(\W_i)] \right|\cdot\E[\stheta']\\
      & \quad + \E[\Theta_i^{*}\W_i^\tran \J\diag(\W_i)] \cdot |(\E[\stheta'] - \E[\stheta^{*}])|\\
      & \le C_1(\E[\stheta']\cdot \Delta_{\pi i} + \E[\Theta_i^{*}\W_i^\tran \J\diag(\W_i)]\cdot \bar{\delta}_{\pi}').
    \end{align*}
    Note that $\E[\Theta_i^{*}\W_i^\tran \J\diag(\W_i)] = \EPi[\W^\tran \J\diag(\W)]\le T$ is a constant. By Assumption \ref{as:overlap_general}, $\E[\stheta]\le 1 / c_{\pi}$. Thus,
    \begin{align}
      &\left\|\bfa_{i1} - \lb\E[\Theta_i^*\W_i^\tran \J\diag(\W_i)]\E[\stheta^*] - \E[\sw^*]^\tran\E[\Theta_i^*\J\diag(\W_i)]\rb\right\|_{2}\nonumber\\
      & \le C_2(\Delta_{\pi i} + \bar{\delta}'_{\pi}), \label{eq:ai1}
    \end{align}
    for some constant $C_2$ that only depends on $c_{\pi}$ and $T$. Let
    \[\bfa_{i1}^{*} = \E[\Theta_i^*\W_i^\tran \J\diag(\W_i)]\E[\stheta^*] - \E[\sw^*]^\tran\E[\Theta_i^*\J\diag(\W_i)].\]
    Since we assume $\tau^{*} = 0$, $\|\E^{(k)}[\td{\bftau}_i]\|_{2} = \|\E^{(k)}[\hbftau_i - \nhbftau_i]\|_{2}  \le \Delta_{y i}$,
    \begin{align}
        \label{eq:cf_goal1_step5_6}
      \left|\frac{1}{m}\sum_{i\in \I_k}\bfa_{i1}^\tran\E^{(k)}[\td{\bftau}_i] - \frac{1}{m}\sum_{i\in \I_k}\bfa_{i1}^{*\tran}\E^{(k)}[\td{\bftau}_i]\right|\le \frac{C_3}{m}\sum_{i\in \I_k}(\Delta_{\pi i} + \bar{\Delta}'_{\pi})\Delta_{y i}.
    \end{align}
    On the other hand, by definition of $\Theta_i^{*}$,
    \[\E[\Theta_i^*\W_i^\tran \J\diag(\W_i)] = \EPi[\W^\tran \J\diag(\W)], \quad \E[\Theta_i^*\J\diag(\W_i)] = \EPi[\J\diag(\W)],\]
    and
    \[\E[\stheta^{*}] = \frac{1}{n}\sum_{i=1}^{n}\E[\Theta_i^{*}] = 1, \quad \E[\sw^*] = \frac{1}{n}\sum_{i=1}^{n}\E[\Theta_i^{*}\J \W_i] = \EPi[\J \W].\]
    Thus,
    \begin{align*}
      \bfa_{i1}^{*\tran} &= \EPi[\W^\tran \J\diag(\W)] - \EPi[\J \W]^\tran\EPi[\J\diag(\W_i)]\\
      & = \EPi[(\W - \EPi[\W])^\tran \J\diag(\W)].
    \end{align*}
    By the \DATEeq,
    \[\bfa_{i1}^{*\tran} = \EPi[(\W - \EPi[\W])^\tran \J\W]\xi^\tran.\]
    As a consequence,
    \begin{equation}
      \label{eq:cf_goal1_step5_7}
      \frac{1}{m}\sum_{i\in \I_k}\bfa_{i1}^{*\tran}\E^{(k)}[\td{\bftau}_i] = \EPi[(\W - \EPi[\W])^\tran \J\W]\cdot \lb\frac{1}{m}\sum_{i\in \I_k}\xi^\tran\E^{(k)}[\td{\bftau}_i]\rb.
    \end{equation}
    Now we turn to the second and third terms of \eqref{eq:cf_goal1_step5_1}. Similar to \eqref{eq:ai1}, we can show that 
    \[\|\bfa_{i2}\|_{2} = \|\bfa_{i2} - (\E[\sw^{*}]^\tran\E[\stheta^{*}] - \E[\sw^{*}]^\tran\E[\stheta^{*}])\|_{2}\le C_3(\Delta_{\pi i} + \bar{\delta}'_{\pi}),\]
    for some constant $C_3$ that only depends on $c_{\pi}$ and $T$. By \eqref{eq:Yi0_WLOG}, $\E^{(k)}[\|\td{\Y}_i(0)\|_{2}]\le \Delta_{yi}$. Therefore,
    \begin{equation}
      \label{eq:cf_goal1_step5_8}
      \left|\frac{1}{m}\sum_{i\in \I_k}\bfa_{i2}^\tran\E^{(k)}[\td{\Y}_i(0)]\right|\le \frac{C_3}{m}\sum_{i\in \I_k}(\Delta_{\pi i}  + \bar{\Delta}'_{\pi})\Delta_{y i}.
    \end{equation}
    Putting \eqref{eq:cf_goal1_step5_1}, \eqref{eq:cf_goal1_step5_6}, \eqref{eq:cf_goal1_step5_7}, and \eqref{eq:cf_goal1_step5_8} together, we arrive at
    \begin{align*}
      &\left|\E^{(k)}[\swy^{(k)}]\E[\stheta'] - \E[\sw^{'}]^\tran\E[\sy^{(k)}] - \EPi[(\W - \EPi[\W])^\tran \J\W]\cdot \lb\frac{1}{m}\sum_{i\in \I_k}\xi^\tran\E^{(k)}[\td{\bftau}_i]\rb\right|\\
      & \le \frac{C_4}{m}\sum_{i\in \I_k}(\Delta_{\pi i} + \bar{\delta}'_{\pi})\Delta_{y i},
    \end{align*}
    for some constant $C_4$ that only depends on $c_{\pi}$ and $T$. Since $\tau^{*} = 0$,
    \begin{align*}
    &\frac{1}{K}\sum_{k=1}^{K}\lb\frac{1}{m}\sum_{i\in \I_k}\xi^\tran\E^{(k)}[\td{\bftau}_i]\rb = \frac{1}{K}\sum_{k=1}^{K}\lb\frac{1}{m}\sum_{i\in \I_k}\xi^\tran (\bftau_i - \E^{(k)}[\hbftau_i])\rb\\
    & = \tau^{*} - \frac{1}{K}\sum_{k=1}^{K}\lb\frac{1}{m}\sum_{i\in \I_k}\xi^\tran \E^{(k)}[\hbftau_i]\rb = 0,
    \end{align*}
    where the last step uses \eqref{eq:avg_hattau}.
    Therefore, averaging over $k$ and marginalizing over $\cD_{-k}$ yields that
    \begin{align*}
      &\left|\frac{1}{K}\sum_{k=1}^{K}\E^{(k)}[\swy]\E[\stheta'] - \E[\sw']^\tran\E^{(k)}[\sy^{(k)}]\right| = O_\P\lb\frac{1}{n}\sum_{i=1}^{n}\E\left[(\Delta_{\pi i} + \bar{\delta}'_{\pi})\Delta_{y i}\right]\rb.
    \end{align*}
    By Cauchy-Schwarz inequality,
    \begin{align*}
      &\left|\frac{1}{K}\sum_{k=1}^{K}\E^{(k)}[\swy]\E[\stheta'] - \E[\sw']^\tran\E^{(k)}[\sy^{(k)}]\right|\\
      & = O_\P\lb\frac{1}{n}\sum_{i=1}^{n}\sqrt{\E\left[(\Delta_{\pi i} + \bar{\delta}'_{\pi})^2\right]}\sqrt{\E\left[\Delta_{y i}^2\right]}\rb\\
      & = O_\P\lb\sqrt{\frac{1}{n}\sum_{i=1}^{n}\E\left[(\Delta_{\pi i} + \bar{\delta}'_{\pi})^2\right]}\sqrt{\frac{1}{n}\sum_{i=1}^{n}\E\left[\Delta_{y i}^2\right]}\rb\\
      & = O_\P\lb\sqrt{\frac{1}{n}\sum_{i=1}^{n}\lb\E[\Delta_{\pi i}^2] + \bar{\delta}_{\pi}^{'2}\rb}\sqrt{\frac{1}{n}\sum_{i=1}^{n}\E[\Delta_{y i}^2]}\rb\\
      & = O_\P\lb\sqrt{\E[\bar{\Delta}_{\pi}^2] + \bar{\delta}_{\pi}^{'2}}\cdot\sqrt{\E[\bar{\Delta}_{y}^2]}\rb.
    \end{align*}
    Therefore, \eqref{eq:cf_goal1_5} is proved by \eqref{eq:E_bardelta_pi'} and \eqref{eq:E_bardelta_y'} on page \pageref{eq:E_bardelta_pi'}.
  \item Next, we shall prove that
    \begin{align}
      \label{eq:cf_goal1_6}
      \left|\frac{1}{K}\sum_{k=1}^{K}\E[\swy']\E^{(k)}[\stheta^{(k)} - \stheta^{'(k)}] - \E^{(k)}[\sw^{(k)} - \sw^{'(k)}]^\tran\E[\sy']\right| = O\lb\sqrt{\E[\bar{\Delta}_{\pi}^2]}\cdot\sqrt{\E[\bar{\Delta}_{y}^2]}\rb.
    \end{align}
    Using the same argument as \eqref{eq:cf_goal1_step5_2}, we can show that
    \begin{align*}
      &\left\|\E[(\Theta'_i - \Theta_i)\J\W_i]\right\|_{2} + \left|\E[\Theta'_i - \Theta_i]\right|\le C_1(\delta'_{\pi i} + \Delta_{\pi i}).
    \end{align*}
    Averaging over $i\in \I_{k}$, we obtain that
    \begin{equation}
      \label{eq:cf_goal1_step6_1}
      |\E^{(k)}[\stheta^{(k)}] - \E^{(k)}[\stheta^{'(k)}]| + \|\E^{(k)}[\sw^{(k)}] - \E^{(k)}[\sw^{'(k)}]\|_{2}\le C_1(\bar{\delta}'_{\pi} + \bar{\Delta}_{\pi}) = O\lb\bar{\Delta}_{\pi}\rb,
    \end{equation}
    where the last step uses \eqref{eq:piimi'_WLOG}. On the other hand,
    \begin{align*}
      \E[\swy']& = \frac{1}{n}\sum_{i=1}^{n}\E[\Theta_i'\W_i^\tran \J \td{\Y}'_i]\\
                          & = \frac{1}{n}\sum_{i=1}^{n}\E[\Theta_i'\W_i^\tran \J \td{\Y}'_i(0)] + \frac{1}{n}\sum_{i=1 }^{n}\E[\Theta'_i\W_i^\tran \J \diag(\W_i)\td{\bftau}'_i]\\
                          & = \frac{1}{n}\sum_{i=1}^{n}\E[\Theta'_i\W_i^\tran \J]\E[\td{\Y}'_i(0)] + \frac{1}{n}\sum_{i=1}^{n}\E[\Theta'_i\W_i^\tran \J \diag(\W_i)]\td{\bftau}'_i.
    \end{align*}
    Note that
    \[\|\E[\Theta'_i\W_i^\tran \J]\|_{2}\le \frac{\sqrt{T}}{c_{\pi}}, \quad \|\E[\Theta'_i\W_i^\tran \J \diag(\W_i)]\|_{2}\le \frac{T}{c_{\pi}}, \quad \|\E[\td{\Y}'_i(0)]\|_{2} + \|\td{\bftau}'_i\|_{2}\le \delta'_{y i}.\]
    As a result, there exists a constant $C_5$ that only depends on $c_{\pi}$ and $T$ such that
    \[|\E[\swy']|\le \frac{C_5}{n}\sum_{i=1}^{n}\delta'_{y i} = O\lb\E[\bar{\Delta}_{y}^2]\rb.\]
    Similarly, 
    \[\|\E[\sy']\|_{2} = O\lb\E[\bar{\Delta}_{y}^2]\rb.\]
    Together with \eqref{eq:cf_goal1_step6_1}, we prove \eqref{eq:cf_goal1_6}.
  \item Consider the following decompositions:
  \begin{align*}
    &\swy\stheta - \swy'\stheta'\\
    & = (\swy - \swy')(\stheta - \stheta') \\
    & + (\swy - \swy')(\stheta' - \E[\stheta']) + (\swy' - \E[\swy'])(\stheta - \stheta')\\
    & + \lb\swy - \swy' - \frac{1}{K}\sum_{k=1}^{K}\E^{(k)}[\swy^{(k)} - \swy^{'(k)}]\rb\E[\stheta']\\
    & + \E[\swy']\lb\stheta - \stheta' - \frac{1}{K}\sum_{k=1}^{K}\E^{(k)}[\stheta^{(k)} - \stheta^{'(k)}]\rb\\
    & - \frac{1}{K}\lb\sum_{k=1}^{K}\E^{(k)}[\swy^{'(k)}]\rb\cdot \E[\stheta']\\
    & + \frac{1}{K}\lb\sum_{k=1}^{K}\E^{(k)}[\swy^{(k)}]\rb\cdot \E[\stheta']\\
    & + \E[\swy']\cdot  \frac{1}{K}\lb\sum_{k=1}^{K}\E^{(k)}[\stheta^{(k)} - \stheta^{'(k)}]\rb,
  \end{align*}
  and
  \begin{align*}
    &\sw^\tran\sy - \sw^{'\tran}\sy'\\
    & = (\sw - \sw')^\tran (\sy - \sy')\\
    & + (\sw - \sw')^\tran (\sy' - \E[\sy']) + (\sw' - \E[\sw'])^\tran (\sy - \sy')\\
    & + \lb\sw - \sw' - \frac{1}{K}\sum_{k=1}^{K}\E^{(k)}[\sw^{(k)} - \sw^{'(k)}]\rb^\tran \E[\sy'] \\
    & + \E[\sw']^\tran \lb \sy - \sy' - \frac{1}{K}\sum_{k=1}^{K}\E^{(k)}[\sy^{(k)} - \sy^{'(k)}]\rb\\
    & - \E[\sw']^\tran \frac{1}{K}\lb\sum_{k=1}^{K}\E^{(k)}[\sy^{'(k)}]\rb\\
    & + \E[\sw']^\tran  \frac{1}{K}\lb\sum_{k=1}^{K}\E^{(k)}[\sy^{(k)}]\rb\\
    & + \frac{1}{K}\lb\sum_{k=1}^{K}\E^{(k)}[\sw^{(k)} - \sw^{'(k)}]\rb^\tran \E[\sy'].
  \end{align*}
  Since $(\bfpi'_i, \bfm'_i, \bftau'_i)$ are deterministic,
  \begin{align*}
    & \frac{1}{K}\lb\sum_{k=1}^{K}\E^{(k)}[\swy^{'(k)}]\rb\cdot \E[\stheta'] = \frac{1}{K}\lb\sum_{k=1}^{K}\E[\swy^{'(k)}]\rb\cdot \E[\stheta'] =\E[\swy']\E[\stheta'],
  \end{align*}
  and 
  \begin{align*}
    & \E[\sw']^\tran \frac{1}{K}\lb\sum_{k=1}^{K}\E^{(k)}[\sy^{'(k)}]\rb = \E[\sw']^\tran \frac{1}{K}\lb\sum_{k=1}^{K}\E[\sy^{'(k)}]\rb= \E[\sw']^\tran \E[\sy'].
  \end{align*}
  By \eqref{eq:cf_goal1_1}, \eqref{eq:cf_goal1_2}, \eqref{eq:cf_goal1_3}, \eqref{eq:cf_goal1_4}, \eqref{eq:cf_goal1_5}, \eqref{eq:cf_goal1_6}, and triangle inequality,
  \[\swy\stheta - \sw^\tran\sy - \left\{\swy'\stheta' - \sw^{'\tran}\sy'\right\} = O_\P\lb n^{-\r} + n^{-(\r + 1) / 2} + n^{-(\r' + 1) / 2} + \sqrt{\E[\bar{\Delta}_{\pi}^2]}\cdot\sqrt{\E[\bar{\Delta}_{y}^2]}\rb.\]
  The proof of \eqref{eq:cf_goal1} is then completed.
  \end{enumerate}

  \subsubsection{Proof of \eqref{eq:cf_goal2}} \label{subsubapp:variance}
  Let
  \begin{align}\label{eq:hVi''}
    \hV''_i &= \Theta'_i\bigg\{\swy - \sy^{\tran}\J \W_i  + \stheta\W_i^\tran \J \td{\Y}'_i - \sw^{\tran}\J\td{\Y}'_i \bigg\}.
  \end{align}
  Recalling the definition of $\hV'_i$ in \eqref{eq:hVi'} on page \pageref{eq:hVi'},
  \begin{align*}
    |\hV'_i - \hV''_i|
    &\le |\swy - \swy'|\cdot \Theta'_i + \|\sy - \sy'\|_{2}\cdot \|\Theta_i' \J \W_i\|_{2}\\
    & \quad + |\stheta - \stheta'|\cdot |\Theta'_i\W_i^\tran \J \td{\Y}'_i| + \|\sw - \sw'\|_{2}\cdot \|\Theta_i' \J\td{\Y}'_i \|_{2}\\
    & \le \left\{|\swy - \swy'| + \|\sy - \sy'\|_{2} + |\stheta - \stheta'| + \|\sw - \sw'\|_{2}\right\}\\
    &\quad \cdot \left\{\Theta'_i + \|\Theta_i' \J \W_i\|_{2} + |\Theta'_i\W_i^\tran \J \td{\Y}'_i| + \|\Theta_i' \J\td{\Y}'_i\|_{2}\right\}
  \end{align*}
  By Jensen's inequality and Cauchy-Schwarz inequality,
  \begin{align*}
    &\frac{1}{n}\sum_{i=1}^{n}(\hV'_i - \hV''_i)^2\\
    & \le 4\left\{|\swy - \swy'| + \|\sy - \sy'\|_{2} + |\stheta - \stheta'| + \|\sw - \sw'\|_{2}\right\}^2\\
    & \quad \cdot \frac{1}{n}\sum_{i=1}^{n}\left\{\Theta_i^{'2} + \|\Theta_i' \J \W_i\|_{2}^2 + |\Theta'_i\W_i^\tran \J \td{\Y}'_i|^2 + \|\Theta_i' \J\td{\Y}'_i\|_{2}^2\right\}\\
    & \le \frac{8T}{c_{\pi}^2}\left\{|\swy - \swy'|^2 + \|\sy - \sy'\|_{2}^2 + |\stheta - \stheta'|^2 + \|\sw - \sw'\|_{2}^2\right\}\cdot \frac{1}{n}\sum_{i=1}^{n}\left\{1 + \|\td{\Y}'_i\|_{2}^2\right\},
  \end{align*}
  where the last inequality uses Assumption \ref{as:overlap_general}. By \eqref{eq:Yi_2norm} and \eqref{eq:cf_goal1_step1_1} on page \pageref{eq:Yi_2norm}, 
  \begin{equation}
    \label{eq:hV'hV''}
    \frac{1}{n}\sum_{i=1}^{n}(\hV'_i - \hV''_i)^2= O_\P(n^{-\r}) = o_\P(1).
  \end{equation}
  On the other hand, recalling the definition of $\hV_i$ in \eqref{eq:hVi} on page \pageref{eq:hVi},
  \begin{align*}
    |\hV''_i - \hV_i|
    & \le |\swy|\cdot |\Theta'_i - \Theta_i| + \|\sy\|_{2} \cdot \|(\Theta'_i - \Theta_i)\J \W_i\|_{2} \\
    & \quad + |\stheta|\cdot |\Theta'_i \W_i^\tran\J \td{\Y}'_i - \Theta_i \W_i^\tran\J \td{\Y}_i| + \|\sw\|_{2} \cdot \|\Theta_i' \J\td{\Y}'_i - \Theta_i \J\td{\Y}_{i}\|_{2}\\
    & \quad + \Theta_i \hat{\tau}\left\{|\sww| + |\sw^\tran \J\W_i| + |\stheta \W_i^\tran \J \W_i| + |\sw^\tran \J \W_i|\right\}\\
    & \le \left\{ |\swy| + \|\sy\|_{2} + |\stheta| + \|\sw\|_{2}\right\} \cdot \bigg\{|\Theta'_i - \Theta_i| + \|(\Theta'_i - \Theta_i)\J \W_i\|_{2}\\
    & \quad + |\Theta_i' \W_i^\tran\J\td{\Y}'_i - \Theta_i \W_i^\tran\J\td{\Y}_{i}| + \|\Theta_i' \J\td{\Y}'_i - \Theta_i \J\td{\Y}_{i}\|_{2}\bigg\}\\
    & \quad + \Theta_i |\hat{\tau}|\left\{|\sww| + |\sw^\tran \J\W_i| + |\stheta \W_i^\tran \J \W_i| + |\sw^\tran \J \W_i|\right\}.
  \end{align*}
  Since $\|\J \W_i\|_{2}\le \sqrt{T}$, 
  \[\|(\Theta'_i - \Theta_i)\J \W_i\|_{2}\le \sqrt{T}|\Theta'_i - \Theta_i|.\]
  By triangle inequality,
  \begin{align*}
    &|\Theta_i' \W_i^\tran\J\td{\Y}'_i - \Theta_i \W_i^\tran\J\td{\Y}_{i}|\\
    & \le |\Theta_i \W_i^\tran\J\td{\Y}'_i - \Theta_i \W_i^\tran\J\td{\Y}_{i}| + |\Theta_i' \W_i^\tran\J\td{\Y}'_i - \Theta_i \W_i^\tran\J\td{\Y}'_i|\\
    & \le \frac{\sqrt{T}}{c_{\pi}}\|\td{\Y}‘_i - \td{\Y}_i\|_{2} + \sqrt{T}\|\td{\Y}_i’\|_{2}\cdot |\Theta'_i - \Theta_i|\\
    & = \frac{\sqrt{T}}{c_{\pi}}\|\hat{\bfm}_{i} - \bfm'_{i}\|_{2} + \sqrt{T}\|\td{\Y}'_i\|_{2}\cdot |\Theta'_i - \Theta_i|.
  \end{align*}
  Similarly,
  \begin{align*}
    &\|\Theta_i' \J\td{\Y}'_i - \Theta_i \J\td{\Y}_{i}\|_{2}\\
    & \le \|\Theta_i \J\td{\Y}'_i - \Theta_i \J\td{\Y}_{i}\|_{2} + \|\Theta_i' \J\td{\Y}'_i - \Theta_i \J\td{\Y}'_i\|_{2}\\
    & \le \frac{1}{c_{\pi}}\|\td{\Y}'_i - \td{\Y}_{i}\|_{2} + \|\td{\Y}'_i\|_{2}\cdot |\Theta'_i - \Theta_i|\\
    & = \frac{1}{c_{\pi}}\|\hat{\bfm}_{i} - \bfm'_{i}\|_{2} + \|\td{\Y}'_i\|_{2}\cdot |\Theta'_i - \Theta_i|
  \end{align*}
Putting pieces together, we have that
  \begin{align}
    &(\hV''_i - \hV_i)^2\nonumber\\
    & \le C\left\{ |\swy| + \|\sy\|_{2} + |\stheta| + \|\sw\|_{2}\right\}^2\bigg\{|\Theta'_i - \Theta_i|^2 \cdot (1 + \|\td{\Y}'_i\|_{2}^2) + \|\hat{\bfm}_{i} - \bfm'_{i}\|_{2}^2\bigg\}\nonumber\\
    & \quad + C|\hat{\tau}|\left\{|\sww| + |\sw^\tran \J\W_i| + |\stheta \W_i^\tran \J \W_i| + |\sw^\tran \J \W_i|\right\}\nonumber,
  \end{align}
  for some constant $C$ that only depends on $c_{\pi}$ and $T$. By Lemma \ref{lem:mean_var} and Markov's inequality,
  \[|\swy| + \|\sy\|_{2} + |\stheta| + \|\sw\|_{2} = O_\P(1), \quad |\sww| + |\sw^\tran \J\W_i| + |\stheta \W_i^\tran \J \W_i| + |\sw^\tran \J \W_i| = O(1).\]
  By the first part of the theorem,
  \[|\hat{\tau}| = o_\P(1).\]
  Therefore,
  \begin{equation}
    \label{eq:hV''hV_1}
    \frac{1}{n}\sum_{i=1}^{n}(\hV''_i - \hV_i)^2 = O_\P\lb\frac{1}{n}\sum_{i=1}^{n}\bigg\{|\Theta'_i - \Theta_i|^2 \cdot (1 + \|\td{\Y}'_i\|_{2}^2) + \|\hat{\bfm}_{i} - \bfm'_{i}\|_{2}^2\bigg\} \rb + o_\P(1).
  \end{equation}
  By Assumption \ref{as:overlap_general} and \ref{as:asymptotic_deterministic},
  \begin{align*}
    \frac{1}{n}\sum_{i=1}^{n}\E[|\Theta'_i - \Theta_i|^2] &= \frac{1}{n}\sum_{i=1}^{n}\E\left[\frac{\bfPi(\W_i)^2}{\hat{\bfpi}_i(\W_i)^2\bfpi'_i(\W_i)^2} |\hat{\bfpi}_i(\W_i) - \bfpi'_i(\W_i)|^2\right]\\
    &\le\frac{1}{c_{\pi}^2}\sum_{i=1}^{n}\E[(\hat{\bfpi}_i(\W_i) - \bfpi'_i(\W_i))^2] = O(n^{-\r}) = o(1).
  \end{align*}
  By Assumption \ref{as:asymptotic_deterministic},
  \[\frac{1}{n}\sum_{i=1}^{n}\E\left[\|\hat{\bfm}_i - \bfm'_i\|_{2}^{2}\right] = O(n^{-\r}) = o(1).\]
  By Markov's inequality, we obtain that
  \begin{equation}
    \label{eq:hV''hV_2}
    \frac{1}{n}\sum_{i=1}^{n}|\Theta'_i - \Theta_i|^2 + \frac{1}{n}\sum_{i=1}^{n}\|\hat{\bfm}_i - \bfm'_i\|_{2}^{2} = o_\P(1).
  \end{equation}
  By H\"{o}lder's inequality,
  \begin{align*}
    \frac{1}{n}\sum_{i=1}^{n}|\Theta'_i - \Theta_i|^2 \cdot \|\td{\Y}'_i\|_{2}^2 &\le \lb\frac{1}{n}\sum_{i=1}^{n}|\Theta'_i - \Theta_i|^{2(1 + 2 / \omega)}\rb^{\omega / (2 + \omega)}\lb\frac{1}{n}\sum_{i=1}^{n}\|\td{\Y}'_i\|_{2}^{2 + \omega}\rb^{2 / (2 + \omega)}.
  \end{align*}
  By Markov's inequality and Assumption \ref{as:independent_general},
  \[\frac{1}{n}\sum_{i=1}^{n}\|\td{\Y}'_i\|_{2}^{2 + \omega} = O_\P\lb \frac{1}{n}\sum_{i=1}^{n}\E[\|\td{\Y}'_i\|_{2}^{2 + \omega}]\rb = O_\P(1).\]
  By Assumption \ref{as:overlap_general},
  \begin{align*}
    &\frac{1}{n}\sum_{i=1}^{n}\E\left[|\Theta'_i - \Theta_i|^{2(1 + 2 / \omega)}\right]\\
    &= \frac{1}{n}\sum_{i=1}^{n}\E\left[\frac{\bfPi(\W_i)^{2(1 + 2 / \omega)}}{\hat{\bfpi}_i(\W_i)^{2(1 + 2/\omega)}\bfpi'_i(\W_i)^{2(1 + 2/\omega)}} |\hat{\bfpi}_i(\W_i) - \bfpi'_i(\W_i)|^{2(1 + 2/\omega)}\right]\\
    & \le \frac{1}{c_{\pi}^{4(1 + 2/\omega)}}\frac{1}{n}\sum_{i=1}^{n}\E\left[(\hat{\bfpi}_i(\W_i) - \bfpi'_i(\W_i))^{2(1 + 2/\omega)}\right]\\
    & \stackrel{(i)}{\le} \frac{1}{c_{\pi}^{4(1 + 2/\omega)}}\frac{1}{n}\sum_{i=1}^{n}\E\left[(\hat{\bfpi}_i(\W_i) - \bfpi'_i(\W_i))^{2}\right]\\
    & = O\lb n^{-\r}\rb = o(1),
  \end{align*}
  where (i) uses the fact that $|\hat{\bfpi}_i(\W_i) - \bfpi'_i(\W_i)|\le 1$. Thus, by Markov's inequality,
  \begin{equation}
    \label{eq:hV''hV_3}
    \frac{1}{n}\sum_{i=1}^{n}|\Theta'_i - \Theta_i|^2 \cdot \|\td{\Y}'_i\|_{2}^2 = o_\P(1).
  \end{equation}
  Putting \eqref{eq:hV''hV_1}, \eqref{eq:hV''hV_2}, and \eqref{eq:hV''hV_3} together, we conclude that
  \begin{equation}
    \label{eq:hV''hV}
    \frac{1}{n}\sum_{i=1}^{n}(\hV''_i - \hV_i)^2 = o_\P(1).
  \end{equation}
  By Jensen's inequality, \eqref{eq:hV'hV''}, and \eqref{eq:hV''hV},
  \begin{equation}
    \label{eq:hV'hV}
    \frac{1}{n}\sum_{i=1}^{n}(\hV'_i - \hV_i)^2 \le \frac{2}{n}\sum_{i=1}^{n}(\hV'_i - \hV''_i)^2 + \frac{2}{n}\sum_{i=1}^{n}(\hV''_i - \hV_i)^2 = o_\P(1).
  \end{equation}
  By Lemma \ref{lem:mean_var}, it is easy to see that
  \begin{equation}
    \label{eq:hVi^2}
    \left|\frac{1}{n}\sum_{i=1}^{n}\hV_i\right|^2\le \frac{1}{n}\sum_{i=1}^{n}\hV_i^2 = O_\P(1).
  \end{equation}
  As a result,
  \[\left|\frac{1}{n}\sum_{i=1}^{n}\hV'_i - \frac{1}{n}\sum_{i=1}^{n}\hV_i\right|\le \frac{1}{n}\sum_{i=1}^{n}|\hV'_i - \hV_i| \le \sqrt{\frac{1}{n}\sum_{i=1}^{n}(\hV'_i - \hV_i)^2} = o_\P(1).\]
  Together with \eqref{eq:hVi^2}, it implies that
  \[\left|\lb\frac{1}{n}\sum_{i=1}^{n}\hV'_i\rb^2 - \lb\frac{1}{n}\sum_{i=1}^{n}\hV_i\rb^2\right| = o_\P(1).\]
  On the other hand, by triangle inequality, Cauchy-Schwarz inequality, and \eqref{eq:hVi^2},
  \begin{align*}
    &\left|\frac{1}{n}\sum_{i=1}^{n}\hV_i^{'2} - \frac{1}{n}\sum_{i=1}^{n}\hV_i^2\right|\le \frac{2}{n}\sum_{i=1}^{n}\hV_i(\hV'_i - \hV_i) + \frac{1}{n}\sum_{i=1}^{n}(\hV'_i - \hV_i)^2\\
    & \le 2\sqrt{\frac{1}{n}\sum_{i=1}^{n}\hV_i^2}\sqrt{\frac{1}{n}\sum_{i=1}^{n}(\hV'_i - \hV_i)^2} + \frac{1}{n}\sum_{i=1}^{n}(\hV'_i - \hV_i)^2 = o_\P(1).
  \end{align*}
  Therefore,
  \[|\hat{\sigma}^2 - \hat{\sigma}^{'2}|\le \left|\frac{1}{n}\sum_{i=1}^{n}\hV_i^2 - \frac{1}{n}\sum_{i=1}^{n}\hV_i^{'2}\right| + \left|\lb\frac{1}{n}\sum_{i=1}^{n}\hV_i\rb^2 - \lb\frac{1}{n}\sum_{i=1}^{n}\hV'_i\rb^2\right| = o_\P(1).\]

\section{Solutions of the DATE equation }\label{app:DATE_equation}
\subsection{The case of two periods}
When there are two periods, the \DATEeq ~only involves four  variables $\bfPi(0, 0), \bfPi(0, 1)$, $\bfPi(1, 0), \bfPi(1, 1)$. Through some tedious algebra presented in Appendix \ref{subapp:proof_T=2}, we can show that the \DATEeq ~can be simplified into the following equation:
\begin{equation}
  \label{eq:DATE_equation_T=2}
  \{\bfPi(1, 1) - \bfPi(0, 0)\}\{\bfPi(1, 0) - \bfPi(0, 1)\} = (\xi_{1} - \xi_{2})\left\{(\bfPi(1, 0) - \bfPi(0, 1))^{2} - (\bfPi(1, 0) + \bfPi(0, 1))\right\}.
\end{equation}

\subsubsection{Difference-in-difference designs}
In the setting of difference-in-difference (DiD), $(0, 0)$ and $(0, 1)$ are the only two possible treatment assignments. As a result, we should set the support of the reshaped distribution to be $\S^{*}= \{(0, 0), (0, 1)\}$. Then \eqref{eq:DATE_equation_T=2} reduces to
\[\bfPi(0, 0)\bfPi(0, 1) = (\xi_1 - \xi_2)(\bfPi(0, 1)^2 - \bfPi(0, 1)) = (\xi_2 - \xi_1)\bfPi(0, 0)\bfPi(0, 1).\]
It has a solution only when $\xi_2 - \xi_1 = 1$, i.e. $(\xi_1, \xi_2) = (0, 1)$ and hence $\tau^{*}(\xi) = \tau_2$, in which case any reshaped distribution $\bfPi$ with $\bfPi(0, 0), \bfPi(0, 1) > 0$ is a solution. This is not surprising because for DiD, no unit is treated in the first period and thus $\tau_1$ is unidentifiable. Nonetheless, $\tau_2$ is an informative causal estimand in the literature of DiD. This implies that the RIPW estimator with any $\bfPi$ with $\bfPi(0, 0), \bfPi(0, 1) > 0$ and $\bfPi(0, 0) + \bfPi(0, 1) = 1$ yields a doubly robust DiD estimator.

\subsubsection{Cross-over designs}
For a two-period cross-over design, $(0, 1)$ and $(1, 0)$ are the only two possible treatment assignments. Since the support of $\bfPi$ must contain at least two elements, it has to be $\S^{*} = \{(1, 0), (0, 1)\}$. Then \DATEeq ~reduces to
\[0 = (\xi_{1} - \xi_{2})\left\{(\bfPi(1, 0) - \bfPi(0, 1))^{2} - (\bfPi(1, 0) + \bfPi(0, 1))\right\}.\]
When $\xi_1 \not = \xi_2$, it implies that
\[0 = (\bfPi(1, 0) - \bfPi(0, 1))^{2} - (\bfPi(1, 0) + \bfPi(0, 1)) = (\bfPi(1, 0) - \bfPi(0, 1))^{2} - 1.\]
It never holds since $\bfPi(1, 0), \bfPi(0, 1) > 0$. By contrast, when $\xi_1 = \xi_2 = 1/2$, any $\bfPi$ with support $(1, 0)$ and $(0, 1)$ is a solution. 

\subsubsection{Estimating equally-weighted DATE for general designs}
When $\xi_1 = \xi_2 = 1/2$, the \DATEeq ~reduces to
\[\{\bfPi(1, 1) - \bfPi(0, 0)\}\{\bfPi(1, 0) - \bfPi(0, 1)\} = 0\Longleftrightarrow \bfPi(1, 1) = \bfPi(0, 0)\text{ or }\bfPi(1, 0) = \bfPi(0, 1).\]
If $\S^{*} = \{(1, 1), (0, 0), (1, 0), (0, 1)\}$ in Assumption \ref{as:overlap_dr}, that is, when all combinations of treatments are possible, the solutions are 
\begin{align*}
  &(\bfPi(1, 1), \bfPi(0, 0), \bfPi(0, 1), \bfPi(1, 0)) = (a, a, b, 1-2a-b), \quad a > 0, 2a + b < 1\\
  \mbox{or }
  &(\bfPi(1, 1), \bfPi(0, 0), \bfPi(0, 1), \bfPi(1, 0)) = (a, 1-a-2b, b, b), \quad b > 0, a + 2b < 1.
\end{align*}
The uniform distribution on $\S^{*}$ is a solution, implying that the IPW weights deliver the average effect in this case. If $\S^{*} = \{(1, 1), (0, 0), (0, 1)\}$ (staggered adoption), we cannot make $\bfPi(1, 0)$ and $\bfPi(0, 1)$ equal since the former must be zero while the latter must be positive. Therefore, the solutions can be characterized as
\begin{equation}
  \label{eq:two_periods_staggered_adoption}
  (\bfPi(1, 1), \bfPi(0, 0), \bfPi(0, 1)) = (a, a, 1 - 2a), \quad a\in (0, 1/2).
\end{equation}
Again, the uniform distribution on $\S^{*}$ is a solution. However, we will show in the next section that the uniform distribution is not a solution for staggered adoption designs with $T \ge 3$.

\subsection{Staggered adoption with multiple periods}
For staggered adoption designs, $\bfpi_{i}$ is supported on
\[\cW_{T}^{\mathrm{sta}} \triangleq \{\w: \w_{1} = \ldots = \w_{i} = 0, \w_{i+1} = \ldots = \w_{T} = 1 \mbox{ for some }i = 0, 1, \ldots, T\}.\]
For notational convenience, we denote by $\w_{(j)}$ the vector in $\cW_{T}^{\mathrm{sta}}$ with $j$ entries equal to $1$ for $j = 0, 1, \ldots, T$. Thus, the support $\S^{*}$ of $\bfPi$ must be a subset of $\cW_{T}^{\mathrm{sta}}$. For general weights, the \DATEeq ~is a quadratic system with complicated structures. Nonetheless, when $\xi_{1} = \ldots = \xi_{T} = 1 / T$, the solution set is an union of segments on the $T$-dimensional simplex with closed-form expressions. We focus on the equally-weighted DATE in this section.

\begin{theorem}\label{thm:staggered_adoption_DATEeq}
Let $\S^{*}=\{\w_{(0)}, \w_{(j_1)}, \ldots, \w_{(j_r)}, \w_{(T)}\}$ with $1\le j_1 < \ldots < j_r\le T-1$. Then the set of solutions of the \DATEeq ~with support $\S^{*}$ is characterized by the following linear system:
    \begin{align}
      \label{eq:solution_staggered_adoption}
      \left\{
      \begin{array}{ll}
        \bfPi(\w_{(T)}) = \frac{T - j_{r}}{T} - \bfPi(\w_{(j_{r})}) + \frac{1}{T}\sum_{k=1}^{r}j_{k}\bfPi(\w_{(j_{k})})\\
        \bfPi(\w_{(j_{k+1})}) + \bfPi(\w_{(j_{k})}) = \frac{j_{k+1} - j_{k}}{T}, \quad k = 1, \ldots, r-1\\
        \bfPi(\w_{(0)}) = 1 - \bfPi(\w_{(T)}) - \sum_{k=1}^{r}\bfPi(\w_{(j_k)})\\
        \bfPi(\w) > 0 \text{ iff } \w\in \S^{*}
      \end{array}
\right.
    \end{align}
Furthermore, the solution set of \eqref{eq:solution_staggered_adoption} is either an empty set or a $1$-dimensional segment in the form of $\{\lambda \bfPi^{(1)} + (1 - \lambda)\bfPi^{(2)}: \lambda\in (0, 1)\}$ for some distributions $\bfPi^{(1)}$ and $\bfPi^{(2)}$.
\end{theorem}

The proof of Theorem \ref{thm:staggered_adoption_DATEeq} is presented in Appendix \ref{subapp:proof_staggered}. In the following corollary, we show that the solution set with $\S^{*} = \cW_{T}^{\mathrm{sta}}$ is always non-empty with nice explicit expressions.
\begin{corollary}\label{cor:staggered_adoption_DATEeq}
  When $\S^{*} = \cW_{T}^{\mathrm{sta}}$, the solution set of \eqref{eq:solution_staggered_adoption} is $\{\lambda \bfPi^{(1)} + (1 - \lambda)\bfPi^{(2)}: \lambda\in (0, 1)\}$ where 
  \begin{itemize}
  \item if $T$ is odd,
    \[\bfPi^{(1)}(\w_{(T)}) = \frac{(T + 1)^{2}}{4T^{2}}, \,\,\, \bfPi^{(1)}(\w_{(0)}) = \frac{T^{2} - 1}{4T^{2}}, \,\,\, \bfPi^{(1)}(\w_{j}) = \frac{I(j\text{ is odd})}{T}, \,\, j = 1,\ldots, T-1,\]
    \[\text{and }\,\,\bfPi^{(2)}(\w_{(j)}) = \bfPi^{(1)}(\w_{(T - j)}), \,\, j = 0, \ldots, T;\]
  \item if $T$ is even,
    \[\bfPi^{(1)}(\w_{(T)}) = \bfPi^{(1)}(\w_{(0)}) = \frac{1}{4}, \,\,\, \bfPi^{(1)}(\w_{j}) = \frac{I(j\text{ is odd})}{T}, \,\, j = 1,\ldots, T-1,\]
    \[\text{and }\,\,\bfPi^{(2)}(\w_{(T)}) = \bfPi^{(2)}(\w_{(0)}) = \frac{T+2}{4T}, \,\,\, \bfPi^{(2)}(\w_{j}) = \frac{I(j\text{ is even})}{T}, \,\, j = 1,\ldots, T-1.\]
  \end{itemize}
\end{corollary}
In particular, when $T = 3$ and $\S^{*} = \cW_{T}^{\mathrm{sta}}$, the solution set is
\begin{equation}
  \label{eq:staggered_adoption_T=3}
  \left\{(\bfPi(\w_{(0)}), \bfPi(\w_{(1)}), \bfPi(\w_{(2)}), \bfPi(\w_{(3)}) = \lambda \lb\frac{2}{9}, \frac{1}{3}, 0, \frac{4}{9}\rb + (1 - \lambda)\lb\frac{4}{9}, 0, \frac{1}{3}, \frac{2}{9}\rb: \lambda \in (0, 1) \right\}.
\end{equation}
Clearly, the uniform distribution on $\S^{*}$ is excluded. Thus, although the RIPW estimator with a uniform reshaped distribution is inconsistent, the non-uniform distribution $(1 / 3, 1 / 6, 1 / 6, 1 / 3)$, namely the midpoint of the solution set, induces a consistent RIPW estimator. For general $T$, it is easy to see that the midpoint is
\begin{equation}
  \label{eq:midpoint}
  \bfPi(\w_{(T)}) = \bfPi(\w_{(0)}) = \frac{T + 1}{4T}, \,\,\, \bfPi(\w_{(j)}) = 
 \frac{1}{2T},\,\, j = 1, \ldots, T - 1.
\end{equation}
This distribution uniformly assigns probabilities on the subset $\{\w_{(1)}, \ldots, \w_{(T-1)}\}$ while puts a large mass on $\{\w_{(0)}, \w_{(T)}\}$. Intuitively, the asymmetry is driven by the special roles of $\w_{(0)}$ and $\w_{(T)}$: the former provides the only control group for period $T$ while the latter provides the only treated group for period $1$.

Corollary \ref{cor:staggered_adoption_DATEeq} offers a unified recipe for the reshaped distribution when the positivity Assumption \ref{as:overlap_dr} holds for all possible assignments. In some applications, certain assignment never or rarely occurs and we are forced to restrict the support of $\bfPi$ into a smaller subset $\S^{*}$. To start with, we provide a detailed account of the case $T = 3$. When $j_1 = 1, j_2 = 2$, \eqref{eq:staggered_adoption_T=3} shows that $\bfPi(\w_{(0)}), \bfPi(\w_{(3)}) > 0$, and thus $\S^{*}$ must be $\cW_3^{\mathrm{sta}}$ and cannot be $\{\w_{(1)}, \w_{(2)}\}$, $\{\w_{(0)}, \w_{(1)}, \w_{(2)}\}$, or $\{\w_{(1)}, \w_{(2)}, \w_{(3)}\}$. When $j_1 = 1, r = 1$, via some tedious algebra, the solution set of \eqref{eq:solution_staggered_adoption} is
\begin{equation}
  \label{eq:staggered_adoption_T=3_1}
  \left\{(\bfPi(\w_{(0)}), \bfPi(\w_{(1)}), \bfPi(\w_{(2)}), \bfPi(\w_{(3)}) = \lambda \lb 0, 1, 0, 0\rb + (1 - \lambda)\lb\frac{1}{3}, 0, 0, \frac{2}{3}\rb: \lambda \in (0, 1) \right\}.
\end{equation}
Thus, $\{\w_{(0)}, \w_{(1)}, \w_{(3)}\}$ is the only support with $j_1 = 1, r = 1$ that induces a non-empty solution set of \eqref{eq:solution_staggered_adoption}. Similarly, we can show that the only support with $j_2 = 1, r = 1$ that induces a non-empty solution set as
\begin{equation}
  \label{eq:staggered_adoption_T=3_2}
  \left\{(\bfPi(\w_{(0)}), \bfPi(\w_{(1)}), \bfPi(\w_{(2)}), \bfPi(\w_{(3)}) = \lambda \lb 0, 0, 1, 0\rb + (1 - \lambda)\lb\frac{2}{3}, 0, 0, \frac{1}{3}\rb: \lambda \in (0, 1) \right\}.
\end{equation}
In sum, $\cW_{T}^{\mathrm{sta}}, \cW_{T}^{\mathrm{sta}} \setminus \{\w_{(1)}\}, \cW_{T}^{\mathrm{sta}}\setminus \{\w_{(2)}\}$ are the only three supports with non-empty solution sets, characterized by \eqref{eq:staggered_adoption_T=3}, \eqref{eq:staggered_adoption_T=3_1}, and \eqref{eq:staggered_adoption_T=3_2}, respectively.

For $T = 3$, $\{j_1, \ldots, j_r\}$ can be any non-empty subset of $\{1, 2\}$. Via some tedious algebra, we can show that this continues to be true for $T = 4$. However, this no longer holds for $T \ge 5$. For instance, if $\{j_1, \ldots, j_r\} = \{1, 2, 4, 5\}$, the second equation of \eqref{eq:solution_staggered_adoption} implies that
\[\bfPi(\w_{(1)}) + \bfPi(\w_{(2)}) = \bfPi(\w_{(4)}) + \bfPi(\w_{(5)}) = \frac{1}{T}, \quad \bfPi(\w_{(2)}) + \bfPi(\w_{(4)}) = \frac{2}{T}.\]
Under the support constraint, the first two equations imply that $\bfPi(\w_{(2)}), \bfPi(\w_{(4)}) < 1 / T$, contradicting with the third equation. Nonetheless, the contradiction can be resolved if any of these four elements is discarded. If this is the case in practice, we can discard the element that is believed to be the least likely assignment.
 
\subsection{Other designs}
In many applications, the treatment can be switched on and off at different periods for a single unit. 
In general, a design is characterized by a collection of possible assignments $\S_{\mathrm{design}}$. If any subset $\S^{*}\subset \S_{\mathrm{design}}$ yields a non-empty solution set of the \DATEeq, we can derive a doubly robust estimator of the DATE. In this section, we consider several designs with more than two periods which are not staggered adoption designs.

First we consider transient designs with zero or one period being treated and with each period being treated with a non-zero chance, i.e., 
\begin{equation*}
 \cW_{T,1}^{\mathrm{tra}} = \left\{\w\in \{0, 1\}^T: \sum_{t=1}^{T}\w_t \le 1\right\}.
\end{equation*}
For notational convenience, we denote by $\td{\w}_{(0)}$ the never-treated assignment and $\td{\w}_{(j)}$ the assignment with only $j$-th period treated. The above design can be encountered, for example, when the treatment is a natural disaster. The following theorem characterizes all solutions of the \DATEeq ~for any $\xi$. 

\begin{theorem}\label{thm:design_tra}
When $\S^{*} = \cW_{T, 1}^{\mathrm{tra}}$, $\bfPi$ is a solution of the \DATEeq ~iff there exists $b > 0$ such that
\[\bfPi(\td{\w}_{(t)})\left\{1 - \bfPi(\td{\w}_{(t)}) - \frac{\bfPi(\td{\w}_{(0)})}{T}\right\} = \xi_t b, \quad \forall t\in [T].\]
\end{theorem}
In particular, when $\xi_t = 1 / T$ for every $t$, Theorem \ref{thm:design_tra} implies that $\bfPi\sim \Unif(\cW_{T, 1}^{\mathrm{tra}})$ is a solution.
In fact, for any given $\bfPi(\td{\w}_0)\in (0, 1)$, $\bfPi$ is a solution if 
\[\bfPi(\cdot \mid \W\neq \td{\w}_{(0)})\sim \Unif(\{\td{\w}_{(1)}, \ldots, \td{\w}_{(T)}\}),\]
where $\W$ denotes a generic random vector drawn from $\bfPi$. The above decomposition can be used to construct solutions for more general transient designs:
\begin{equation*}
 \cW_{T,k}^{\mathrm{tra}} = \left\{\w\in \{0, 1\}^T: \sum_{t=1}^{T}\w_t \le k\right\}.
\end{equation*}
This design is common in marketing experiments where, for example, $k$ is the maximal number of coupons given to a user and each user can receive coupons in any combination of up to $k$ time periods. 
\begin{theorem}\label{thm:design_tra_general}
  When $\S^{*} = \cW_{T}^{\mathrm{tra}}$, $\bfPi$ is a solution of the \DATEeq ~with $\xi_t  = 1/T\,\, (t = 1, \ldots, T)$, if
  \[\bfPi\left(\cdot \mid \sum_{t=1}^{T}\W_t = k'\right)\sim \Unif(\cW_{T, k'}^{\mathrm{tra}}\setminus \cW_{T, k'-1}^{\mathrm{tra}}), \quad k' = 1, \ldots, k,\]
\end{theorem}

\subsection{Proofs}
  For notational convenience, denote by $h(\bfPi) = (h_{1}(\bfPi), \ldots, h_{T}(\bfPi))$ the left-hand side of the \DATEeq. We start by a simple but useful observation that, for any $\bfPi$,
\begin{align}
  \one_{T}^{\tran}h(\bfPi) &= \EPi\left[(\one_{T}^{\tran}\diag(\W) - \one_{T}^{\tran}\xi \W^{\tran})\J(\W - \EPi[\W])\right]\nonumber\\
  & = \EPi\left[(\W^{\tran} - \W^{\tran})\J(\W - \EPi[\W])\right] = 0.  \label{eq:oneT_h}
\end{align}
Thus, there is at least one redundant equation and for any matrix $V\in \R^{T\times (T-1)}$ with $V^{\tran}\one_{T} = 0$, 
\begin{equation}
  \label{eq:h_diff}
  h(\bfPi) = 0\Longleftrightarrow V^{\tran}h(\bfPi) = 0.
\end{equation}

\subsubsection{Proof of equation \eqref{eq:DATE_equation_T=2}}\label{subapp:proof_T=2}
Set $V = (1, -1)^\tran$ in \eqref{eq:h_diff}. Then
\[V^{\tran}h(\bfPi) = 0 \Longleftrightarrow h_{1}(\bfPi) - h_{2}(\bfPi) = 0.\]
As a result,
\begin{align*}
  0 &= \EPi \left[\lb (W_{1}, -W_{2}) - (\xi_{1} - \xi_{2})(W_{1}, W_{2})\rb
    \begin{bmatrix}
      1 & -1\\
      -1 & 1
    \end{bmatrix}
    \begin{bmatrix}
      W_{1} - \EPi[W_{1}]\\
      W_{2} - \EPi[W_{2}]
    \end{bmatrix}
  \right]\\
    & = \EPi [(W_{1} + W_{2} - (\xi_{1} - \xi_{2})(W_{1} - W_{2}))(W_{1} - W_{2} - \EPi(W_{1} - W_{2}))]\\
    & = \EPi [W_{1}^{2} - W_{2}^{2} - (\xi_{1} - \xi_{2})(W_{1} - W_{2})^{2}]\\
    & \quad - \EPi [W_{1} + W_{2} - (\xi_{1} - \xi_{2})(W_{1} - W_{2})]\EPi(W_{1} - W_{2})\\
    & = \EPi [W_{1} - W_{2} - (\xi_{1} - \xi_{2})(W_{1} - W_{2})^{2}]\\
    & \quad - \EPi [W_{1} + W_{2} - (\xi_{1} - \xi_{2})(W_{1} - W_{2})]\EPi(W_{1} - W_{2})\\
    & = (\bfPi(1, 0) - \bfPi(0, 1)) - (\xi_{1} - \xi_{2})(\bfPi(1, 0) + \bfPi(0, 1))\\
    & \quad - \left\{\bfPi(1, 0) + \bfPi(0, 1) + 2\bfPi(1, 1) - (\xi_{1} - \xi_{2})(\bfPi(1, 0) - \bfPi(0, 1))\right\}\left\{\bfPi(1, 0) - \bfPi(0, 1)\right\}\\
   & = (\bfPi(1, 0) - \bfPi(0, 1)) - (\xi_{1} - \xi_{2})(\bfPi(1, 0) + \bfPi(0, 1))\\
    & \quad - \left\{1 + \bfPi(1, 1) - \bfPi(0, 0) - (\xi_{1} - \xi_{2})(\bfPi(1, 0) - \bfPi(0, 1))\right\}\left\{\bfPi(1, 0) - \bfPi(0, 1)\right\}.
\end{align*}
Rearranging the terms yields
\begin{equation}
  \label{eq:Pi_cond_T=2}
  \{\bfPi(1, 1) - \bfPi(0, 0)\}\{\bfPi(1, 0) - \bfPi(0, 1)\} = (\xi_{1} - \xi_{2})\left\{(\bfPi(1, 0) - \bfPi(0, 1))^{2} - (\bfPi(1, 0) + \bfPi(0, 1))\right\}.
\end{equation}

\subsubsection{Proof of Theorem \ref{thm:staggered_adoption_DATEeq}}\label{subapp:proof_staggered}
Let $\e_{j}$ denote the $j$-th canonical basis in $\R^{T}$. Then
\[h_{j}(\bfPi) = \e_{j}^\tran \EPi\left[(\diag(\W) - \xi\W^{\tran})\J(\W - \EPi[\W])\right].\]
We can decompose $h_{j}(\bfPi)$ into $h_{j1}(\bfPi)- \xi_{j}h_{2}(\bfPi)$ where
\[h_{j1}(\bfPi) = \e_{j}^\tran \EPi\left[\diag(\W)\J(\W - \EPi[\W])\right], \quad h_{2}(\bfPi) = \EPi\left[\W^{\tran}\J(\W - \EPi[\W])\right].\]
Then
\begin{align*}
  h_{j1}(\bfPi)
  &= \EPi\left[W_{j}\e_{j}^\tran\J (\W - \EPi[\W])\right]\\
  & = \EPi\left[W_{j}\e_{j}^\tran\J\W \right] - \EPi\left[W_{j}\e_{j}^\tran\J\right]\EPi[\W]\\
  & = \EPi\left[W_{j}\lb W_{j} - \frac{\one_{T}^\tran \W}{T}\rb \right] - \EPi\left[W_{j}\right]\e_{j}^\tran\J\EPi[\W]\\
  & = \EPi\left[W_{j}\lb W_{j} - \frac{\one_{T}^\tran \W}{T}\rb \right] - \EPi\left[W_{j}\right]\EPi\left[W_{j} - \frac{\one_{T}^\tran \W}{T}\right]\\
  & = \EPi[W_{j}] - (\EPi[W_{j}])^{2} + \frac{\EPi [W_{j}]\EPi[\one_{T}^{\tran}\W]}{T} - \frac{\EPi[W_{j}(\one_{T}^{\tran}\W)]}{T},
\end{align*}
where the last equality follows from the fact that $W_{j}^{2} = W_{j}$. 
By \eqref{eq:h_diff}, it is equivalent to find $\bfPi$ satisfying
\begin{equation*}
  \Delta h_{j}(\bfPi) = h_{j+1}(\bfPi) - h_{j}(\bfPi) = 0, \quad j = 1, 2, \ldots, T-1.
\end{equation*}
In this case, $\xi_{j+1} = \xi_{j}$ for any $j$, and thus,
\begin{equation}
  \label{eq:staggered_adoption_Deltah}
  h_{(j+1)1}(\bfPi) - h_{j1}(\bfPi) = 0, \quad j = 1, 2, \ldots, T-1.
\end{equation}
By definition,
\begin{equation}
  \label{eq:Wj+1-Wj}
  W_{j+1} - W_{j} = I(\W = \w_{(T - j)}).
\end{equation}
As a consequence, we have
\[\EPi[W_{j+1}] - \EPi[W_{j}] = \bfPi(\w_{(T - j)}), \]
\[(\EPi[W_{j+1}])^{2} - (\EPi[W_{j}])^{2} = \bfPi(\w_{(T - j)})^{2} + 2\bfPi(\w_{(T - j)})\EPi[W_{j}],\]
and
\begin{align*}
  &\EPi[W_{j+1}(\one_{T}^{\tran}\W)] - \EPi[W_{j}(\one_{T}^{\tran}\W)]\\
  & = \EPi[I(\W = \w_{(T - j)})(\one_{T}^{\tran}\w_{(T-j)})] = (T - j)\bfPi(\w_{(T - j)}).
\end{align*}
As a result,
\begin{align}
  &h_{(j+1)1}(\bfPi) - h_{j1}(\bfPi)\nonumber\\
  &= \bfPi(\w_{(T - j)})\left\{1 - \bfPi(\w_{(T - j)}) - 2\EPi[W_{j}] + \frac{\EPi[\one_{T}^{\tran}\W]}{T} - \frac{T - j}{T}\right\}\nonumber\\
  &= \bfPi(\w_{(T - j)})\left\{\frac{j}{T} - \bfPi(\w_{(T - j)}) - 2\EPi[W_{j}] + \frac{\EPi[\one_{T}^{\tran}\W]}{T}\right\}.\label{eq:Deltah}
\end{align}
Let
\begin{equation}
  \label{eq:gj}
  g_{j}(\bfPi) = \frac{T - j}{T} - \bfPi(\w_{(j)}) - 2\EPi[W_{T-j}] + \frac{\EPi[\one_{T}^{\tran}\W]}{T}.
\end{equation}
Thus, \eqref{eq:staggered_adoption_Deltah} can be reformulated as
\begin{equation}
  \label{eq:staggered_adoption_g}
  \bfPi(\w_{(j)}) = 0 \,\, \mbox{ or }\,\, g_{j}(\bfPi) = 0, \quad j = 1, 2, \ldots, T - 1.
\end{equation}
Since $\S^{*} = \{\w_{(0)}, \w_{(j_1)}, \ldots, \w_{(j_r)}, \w_{(T)}\}$, $\bfPi(\w_{(j_k)}) > 0$ for each $k = 1, \ldots, r$. As a result,
\eqref{eq:staggered_adoption_g} is equivalent to
\begin{equation}
  \label{eq:staggered_adoption_g_simplified}
  g_{j_{r}}(\bfPi) = 0, \quad g_{j_{k}}(\bfPi) - g_{j_{k+1}}(\bfPi) = 0, \quad k = 1, \ldots, r-1.
\end{equation}
Note that 
\[W_{T-j_k} = 1\Longleftrightarrow \W\in \{\w_{(j_k + 1)}, \ldots, \w_{(T)}\}.\] The first equation is equivalent to
\begin{align}
  &\frac{T - j_{r}}{T} - \bfPi(\w_{(j_{r})}) - 2\bfPi(\w_{(T)}) + \frac{1}{T}\lb\sum_{k=1}^{r}j_{k}\bfPi(\w_{(j_{k})}) + T\bfPi(\w_{(T)})\rb = 0\nonumber\\
  &\Longleftrightarrow \bfPi(\w_{(T)}) = \frac{T - j_{r}}{T} - \bfPi(\w_{(j_{r})}) + \frac{1}{T}\sum_{k=1}^{r}j_{k}\bfPi(\w_{(j_{k})}).   \label{eq:staggered_adoption_g1}
\end{align}
By \eqref{eq:Wj+1-Wj},
\begin{align*}
  \EPi[W_{T-j_{k}}] - \EPi[W_{T - j_{k+1}}] &= \P_{\W\sim \bfPi}\lb \W\in \{\w_{(j_{k}+1)}, \w_{(j_{k}+2)}, \ldots, \w_{(j_{k+1})}\}\rb\\
  & = \P_{\W\sim \bfPi}\lb\W = \w_{(j_{k+1})}\rb = \bfPi(\w_{(j_{k+1})}).
\end{align*}
Therefore, the second equation of \eqref{eq:staggered_adoption_g} can be simplified to
\begin{equation}
  \label{eq:staggered_adoption_Deltag}
  \bfPi(\w_{(j_{k+1})}) + \bfPi(\w_{(j_{k})}) = \frac{j_{k+1} - j_{k}}{T}, \quad k = 1, \ldots, r-1.
\end{equation}
Finally the simplex constraint determines $\bfPi(\td{\w}_{(0)})$ as
\begin{equation}
  \label{eq:staggered_adoption_gPi0}
  \bfPi(\w_{(0)}) = 1 - \bfPi(\w_{(T)}) - \sum_{k=1}^{r}\bfPi(\w_{(j_{k})}).
\end{equation}
Clearly, $\bfPi(\w_{(j_{1})})$ determines all other $\bfPi(\W_{(j_{k})})$'s. Therefore, the solution set of \eqref{eq:staggered_adoption_g1} - \eqref{eq:staggered_adoption_gPi0} is a one-dimensional linear subspace. The solution set of the \DATEeq ~is empty if it has no intersection with the set $\{\bfPi: \bfPi(\w_{(j_{k})}) > 0, r = 1, \ldots, r\}$; otherwise, it must be a segment which can be characterized as $\{\lambda \bfPi^{(1)} + (1 - \lambda)\bfPi^{(2)}: \lambda\in (0, 1)\}$.

\subsubsection{Proof of Theorem \ref{thm:design_tra}}
Let $\bfeta = (\bfPi(\td{\w}_{(1)}), \ldots, \bfPi(\td{\w}_{(T)}))\in \R^{T}$. Then the \DATEeq ~can be equivalently formulated as
\[\sum_{j=1}^{T}\lb \diag(\td{\w}_{(j)}) - \xi \td{\w}_{(j)}^\tran\rb \J(\td{\w}_{(j)} - \bfeta)\eta_{j} = 0.\]
Since $\td{\w}_{(j)} = \e_{j}$, $\diag(\td{\w}_{(j)}) = \e_{j}\e_{j}^\tran$ and we can reformulate the above equation as
\[\sum_{j=1}^{T}\lb \e_{j} - \xi \rb \e_{j}^\tran \J(\e_{j} - \bfeta)\eta_{j} = 0\Longleftrightarrow \sum_{j=1}^{T}f_{j}(\bfeta)\e_{j} = \left\{\sum_{j=1}^{T}f_{j}(\bfeta)\right\}\xi.\]
where $f_{j}(\bfeta) = \e_{j}^\tran \J(\e_{j} - \bfeta)\eta_{j}$. It can be equivalently formulated as an equation on $\bfeta$ and a scalar $b$: 
\begin{equation}
  \label{eq:design_tra_1}
  \sum_{j=1}^{T}f_{j}(\bfeta)\e_{j} = b\xi.
\end{equation}
This is because for any $\bfeta$ that satisfies \eqref{eq:design_tra_1}, multiplying $\one_{T}^\tran$ on both sides implies that
\[b = b(\xi^\tran \one_{T}) = \sum_{j=1}^{T}f_{j}(\bfeta).\]
Taking the $j$-th entry of both sides, \eqref{eq:design_tra_1} yields that
\begin{equation}
  \label{eq:fj}
  f_{j}(\bfeta) = \xi_{j}b.
\end{equation}
By definition,
\[f_{j}(\bfeta) = \eta_{j}\lb \e_{j}^\tran \J \e_{j} - \e_{j}^\tran \J \bfeta\rb = \eta_{j}\lb 1 - \frac{1}{T} - \eta_{j} + \frac{1}{T}\sum_{j=1}^{T}\eta_{j}\rb.\]
Since $\bfPi$ should be supported on $\{\td{\w}_{(0)}, \td{\w}_{(1)}, \ldots, \td{\w}_{(T)}\}$, 
\[\sum_{j=1}^{T}\eta_{j} = \sum_{j=1}^{T}\bfPi(\td{\w}_{(j)}) = 1 - \bfPi(\td{\w}_{(0)}).\]
Therefore, \eqref{eq:fj} is equivalent to
\[\bfPi(\td{\w}_{(j)})\lb 1 - \bfPi(\td{\w}_{(j)}) - \frac{\bfPi(\td{\w}_{(0)})}{T}\rb = \xi_{j}b.\]

\subsubsection{Proof of Theorem \ref{thm:design_tra_general}}
Let $\|\w\|_1$ be the $L_1$ norm of $\w$, i.e., $\|\w\|_{1} = \sum_{i=1}^{n}w_{i}$. For given $\bfPi$ such that
\[\bfPi\left(\cdot \mid \|\w\|_1 = k'\right)\sim \Unif(\cW_{T, k'}^{\mathrm{tra}}\setminus \cW_{T, k'-1}^{\mathrm{tra}}), \quad k' = 1, \ldots, k,\]
By symmetry, 
\[\EPi[\W \mid \|\W\|_1] = \frac{\|\W\|_1}{T}\one_{T}.\]
By the iterated law of expectation, 
\[\EPi[\W] = \E_{\|\W\|_1}[\EPi[\W \mid \|\W\|_1]] = \frac{\EPi[\|\W\|_1]}{T}\one_{T}.\]
Since $\J \one_{T} = 0$, the \DATEeq ~with $\xi = \one_{T} / T$ reduces to
\[\EPi\left[\lb\diag(\W) - \frac{\one_{T}}{T}\W^{\tran}\rb\J\W \right] = 0.\]
We will prove the following stronger claim:
\[\EPi\left[\lb\diag(\W) - \frac{\one_{T}}{T}\W^{\tran}\rb\J\W \mid \|\W\|_1 = k'\right] = 0, \quad \forall k' = 1, \ldots, k.\]
Conditional on $\|\W\|_1 = k'$,
\[\J\W = \W - \frac{k'}{T}\one_{T}, \quad \diag(\W)\W = \W, \quad \W^\tran\W = \W^\tran \one_{T} = k'\]
Thus,
\begin{align*}
  &\EPi\left[\lb\diag(\W) - \frac{\one_{T}}{T}\W^{\tran}\rb\J\W \mid \|\W\|_1 = k'\right]\\
  & = \EPi\left[\lb\diag(\W) - \frac{\one_{T}}{T}\W^{\tran}\rb\lb \W - \frac{k'}{T}\one_{T}\rb \mid \|\W\|_1 = k'\right]\\
  & = \EPi\left[\W - \frac{k'}{T}\W - \frac{k'\one_{T}}{T} + \frac{k^{'2}\one_{T}}{T^2} \mid \|\W\|_1 = k'\right]\\
  & = 0.
\end{align*}

\subsection{A general solver via nonlinear programming}\label{subapp:solver}
For a general design $\S_{\mathrm{design}} = \{\check{\w}_{(1)}, \ldots, \check{\w}_{(K)}\}$, the \DATEeq ~can be formulated as a quadratic system. The $j$-th equation of \DATEeq ~is 
\begin{equation}\label{eq:DATEeq_j}
\EPi\left[(\\e_{j}\W_{j} - \W\xi_{j})^{T}\J(\W - \EPi[\W])\right] = 0,
\end{equation}
Let $\bfp = (\bfPi(\check{\w}_{(1)}), \ldots, \bfPi(\check{\w}_{(K)}))\in \R^{T}, A = (\check{\w}_{(1)}, \ldots, \check{\w}_{(K)})\in \R^{T\times K}$, $B^{(j)} = (B_{1}^{(j)}, \ldots, B_{K}^{(j)})\in \R^{T\times K}$, and $\bfb^{(j)} = (b_1^{(j)}, \ldots, b_K^{(j)})^\tran\in \R^{K}$, where
\[B_{k}^{(j)} = J(\e_j\check{\w}_{(k), j} - \check{\w}_{(k)}\xi_j)\in \R^{T}, \quad b_{k}^{(j)} = \check{\w}_{(k)}^\tran B_{k}^{(j)}\in \R.\]
It is easy to see that $B^{(j)} = \J(\e_j\e_j^\tran - \xi_j I)A$ and $\bfb^{(j)} = \diag(A^\tran B^{(j)})$. Then \eqref{eq:DATEeq_j} can be reformulated as 
\[\bfp^\tran \bfb^{(j)} - \bfp^\tran (A^\tran B^{(j)})\bfp = 0.\]
As a result, the \DATEeq ~has a solution iff the minimal value of the following optimization problem is $0$: 
\begin{equation}\label{eq:DATE_optim}
\min \sum_{j=1}^{T}\{\bfp^\tran \bfb^{(j)} - \bfp^\tran (A^\tran B^{(j)})\bfp\}^2, \quad \text{s.t., }\bfp^\tran \one = 1, \bfp\ge 0.
\end{equation}
We can optimize \eqref{eq:DATE_optim} via the standard BFGS algorithm, with the uniform distribution being the initial value. When the minimal value with a given initial value is bounded away from zero, we will try other randomly generated initial values to ensure a thorough search. If none of the initial values yields a zero objective,  we claim that the \DATEeq ~has no solution. Note that \eqref{eq:DATE_optim} is a nonconvex problem, the BFGS algorithm is not guaranteed to find the global minimum. Therefore, it should be viewed as an attempt to find a solution of the \DATEeq ~instead of a trustable solver. 

On the other hand, when the \DATEeq ~has multiple solutions, it is unclear which solution can be found. In principle, we can add different constraints or regularizers to \eqref{eq:DATE_optim} in order to obtain a "well-behaved" solution. For instance, it is reasonable to find the most dispersed reshaped function to maximize the sample efficiency. For this purpose, we can find the solution that maximizes $\min_{k}\bfPi(\check{\w}_{(k)})$. This can be achieved by replacing the constraint $\bfp \ge 0$ in \eqref{eq:DATE_optim} by $\bfp\ge c\one$ and find the largest $c$ for which the minimal value is zero by a binary search. 

\section{Aggregated AIPW estimator is not doubly robust\\ in the presence of fixed effects}\label{app:AIPW}
We are not aware of other doubly robust estimators for DATE when the treatment and outcome models are defined as in our paper. In the absence of dynamic treatment effects, it is tempting to treat each period as a cross-sectional data, estimate the time-specific ATE $\tau_t$ by an aggregated AIPW estimator, and aggregate these estimates. To the best of our knowledge, this estimator has not been proposed in the literature. However, perhaps surprisingly, we show in this section that the aggregated AIPW estimator is not doubly robust because of the fixed effect terms in the outcome model. 

Specifically, for time period $t$, the AIPW estimator for $\tau_t$ is defined as
\[\hat{\tau}_t = \frac{1}{n}\sum_{i=1}^{n}\lb \frac{(Y_{it} - \hat{\E}[Y_{it}(1)\mid \X_i])W_{it}}{\hat{\P}(W_{it} = 1\mid \X_i)} - \frac{(Y_{it} - \hat{\E}[Y_{it}(0)\mid \X_i])(1 - W_{it})}{\hat{\P}(W_{it} = 0\mid \X_i)} + \hat{\E}[Y_{it}(1)\mid \X_i] - \hat{\E}[Y_{it}(0)\mid \X_i]\rb.\]
Then the aggregated AIPW estimator is defined as
\[\hat{\tau}_{\mathrm{AIPW}} = \frac{1}{T}\sum_{t=1}^{T}\hat{\tau}_{t}.\]
It is known that $\hat{\tau}_t$ is doubly robust in the sense that $\hat{\tau}_t$ is consistent if either $\hat{\P}(W_{it} = 1)$ or $(\hat{\E}[Y_{it}(1)\mid X_i], \hat{\E}[Y_{it}(0)\mid X_i])$ is consistent for all $i$ and $t$. Importantly, the requirement on the outcome model for the AIPW estimator is strictly stronger than that for the RIPW estimator; the former requires both $m_{it}$ and the fixed effects to be consistently estimated while the latter only requires $m_{it}$ to be consistent. It turns out that the extra requirement leads to tricky problems of the AIPW estimator.

To demonstrate the failure of the AIPW estimator, we only consider the case with sample size $n = 1000$ and a constant treatment effect to highlight that the failure is not driven by small samples or effect heterogeneity. In particular, we consider a standard TWFE model
\[Y_{it}(0) = \alpha_i + \lambda_t + m_{it} + \epsilon_{it}, \quad m_{it} = X_i\beta_{t}, \quad \tau_{it} = \tau,\]
where $\sum_{i=1}^{n}\alpha_i = \sum_{t=1}^{T}\lambda_t = 0$.
The other details are the same as Section \ref{subsec:synthetic}.

Both the RIPW and the aggregated AIPW estimators require estimates of the treatment and outcome models. First, we consider a wrong and a correct treatment model:
\begin{itemize}
\item (Wrong treatment model): set $\hat{\bfpi}_i(\w) = |\{j: \W_j = \w\}| / n$, i.e., the empirical distribution of $\W_i$'s that ignores the covariate;
\item (Correct treatment model): set $\hat{\bfpi}_i(\w) = |\{j: \W_j = \w, X_j = X_i\}| / |\{j: X_j = X_i\}|$, i.e., the empirical distribution of $\W_i$'s stratified by the covariate.
\end{itemize}
With a large sample, $\hat{\bfpi}_i$ in the second setting is a consistent estimator of $\bfpi_i$. For the aggregated AIPW estimator, we use the marginal distributions of $\hat{\bfpi}_i$ as the estimates of marginal propensity scores. Similarly, we consider a wrong and a correct outcome model:
\begin{itemize}
\item (Wrong outcome model): $\hat{m}_{it} = 0$ for every $i$ and $t$;
\item (Correct outcome model): run unweighted TWFE regression adjusting for interaction between $X_i$ and time fixed effects, i.e., $X_{i}I(t = t')$ for each $t' = 1, \ldots, t$, and set $\hat{m}_{it} = X_{i}\hat{\beta}_t$. 
\end{itemize}
With a large sample, the standard theory implies the consistency of $\hat{\beta}_t$, and hence $\hat{m}_{it} \approx m_{it}$. Unlike the RIPW estimator, the aggregated AIPW estimator requires the estimate of full conditional expectations of potential outcomes, instead of merely $\hat{m}_{it}$. In this case, a reasonble estimate of the outcome model can be formulated as
\[\hat{\E}[Y_{it}(0)\mid X_i] = \hat{\alpha}_i + \hat{\lambda}_t + X_i\hat{\beta}_t, \quad \hat{\E}[Y_{it}(1)\mid X_i] = \hat{\E}[Y_{it}(0)\mid X_i] + \hat{\tau}.\]
For short panels with $T = O(1)$, the time fixed effects $\lambda_t$'s can be estimated via the standard TWFE regression, which are known to be consistent. However, there is no way to consistently estimate the unit fixed effect $\alpha_i$ since only $T$ samples $Y_{i1}, \ldots, Y_{iT}$ can be used for estimation. The central question is how to estimate $\alpha_i$ for the aggregated AIPW estimator. Here we consider three strategies:
\begin{enumerate}[(1)]
\item using the plug-in estimate of $\alpha_i$'s, even if they are inconsistent;
\item pretending that $\alpha_i$ does not exist and setting $\hat{\alpha}_i = 0$;
\end{enumerate}
Note that the first strategy cannot be used with cross-fitting because it is impossible to estimate $\alpha_i$ without using the $i$-th sample.

We then consider all four combinations of outcome and treatment modelling. Figure \ref{fig:simul_aipw} presents the boxplots of $\hat{\tau} - \tau$ for the three versions of AIPW, RIPW, and unweighted TWFE estimator.

First, we can see that all estimators are unbiased when both models are correct and biased when both models are wrong. As expected, the RIPW estimator is also unbiased when one of the model is correct, and the unweighted estimator is unbiased when the outcome model is correct. However, none of AIPW estimators are doubly robust: the AIPW estimator with estimated fixed effects is biased when the treatment model is correct, and the AIPW estimator that zeros out fixed effects with or without cross-fitting are biased when the outcome model is correct. 

\begin{figure}
  \centering
  \includegraphics[width = 1\textwidth]{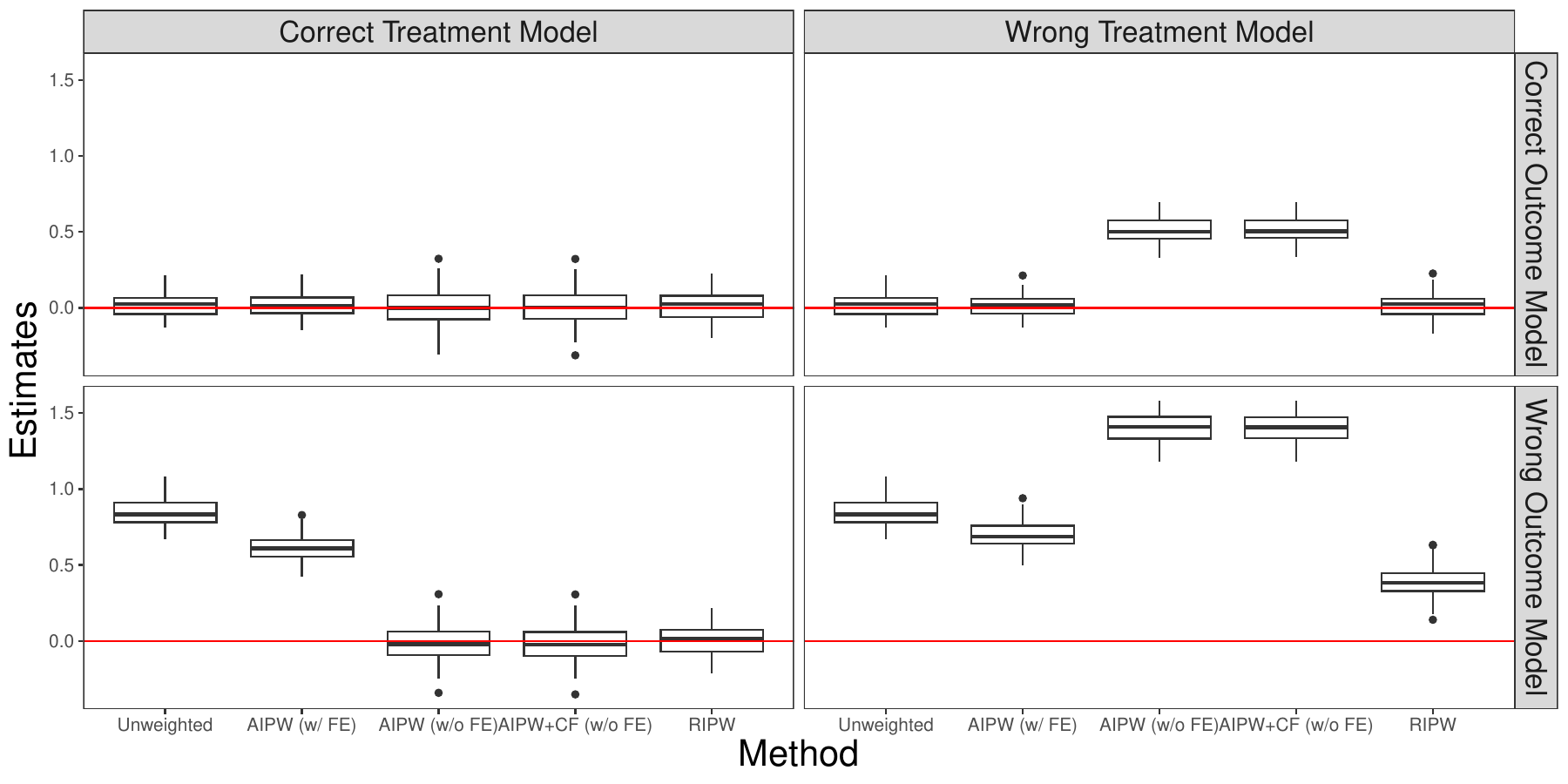}
  \caption{Boxplots of $\hat{\tau} - \tau$ for the RIPW estimator, unweighted TWFE estimator, and the three versions of AIPW: "AIPW (w/ FE)" for the one fit on the entire data with estimated fixed effects, "AIPW (w/o FE)" for the one fit on the entire data with fixed effects zeroed out, and "AIPW+CF (w/o FE)" for the cross-fitted one with fixed effects zeroed out.}\label{fig:simul_aipw}
\end{figure}

The bias of AIPW estimator that zeros out the fixed effects can be attributed to biased estimates of the outcome model despite including the covariates. The bias of the in-sample AIPW estimator can be attributed to the dependence between the outcome model estimates and the treatment assignment. In fact, when $T$ is small, this dependence is nonvanishing no matter how fixed effects are estimated. On the other hand, the AIPW estimator is valid under a correct treatment model but a wrong outcome model only when the outcome model estimate is asymptotically independent of the assignments. In sum, there is no simple way to estimate fixed effects to make the resulting aggregated AIPW estimator doubly robust.

\section{More details of the OpenTable dataset}\label{app:experiment}

We collect the variables from different sources. 
\begin{itemize}
\item Daily state-level year-over-year percentage change in seated diners provided by OpenTable (outcome variable): \url{https://www.opentable.com/state-of-industry}. \item Indicator of whether the state of emergency has been declared (treatment variable):  \url{https://www.businessinsider.com/}

\url{california-washington-state-of-emergency-coronavirus-what-it-means-2020-3}. 
\item Daily state-level accumulated confirmed cases 
(covariate): \url{https://coronavirus.jhu.edu/}.
\item Vote share of Democrats based on the 2016 presidential election data (covariate): \url{https://dataverse.harvard.edu/dataset.xhtml?persistentId=doi:10.7910/DVN/VOQCHQ}.
\item Number of hospital beds: \url{https://github.com/rbracco/covidcompare}.
\end{itemize}

The summary statistics are reported in Table \ref{tab:summary}. We also plot the daily average of the reservation difference and confirmed cases as well as the histograms of the other two variables in Figure \ref{fig:summary_statistics}.

\begin{table}
\centering
\begin{tabular}{lcccc}
\toprule
\toprule
Variable & Mean & SD & 1st quartile & 3rd quartile \\
\midrule
Reservation Diff. (in \%) & -8.89 & 14.94 & -17.00 & -1.00\\
Confirmed Cases & 14.55 & 50.10 & 0.00 & 7.00\\
Vote Share & 49.40 & 9.72 & 40.11 & 57.18\\
log(\#Hospital Beds) & 9.83 & 0.81 & 9.37 & 10.24\\
\bottomrule
\bottomrule
\end{tabular}
\caption{Summary statistics.}\label{tab:summary}
\end{table}

\begin{figure}
\includegraphics[width = 0.47\textwidth]
{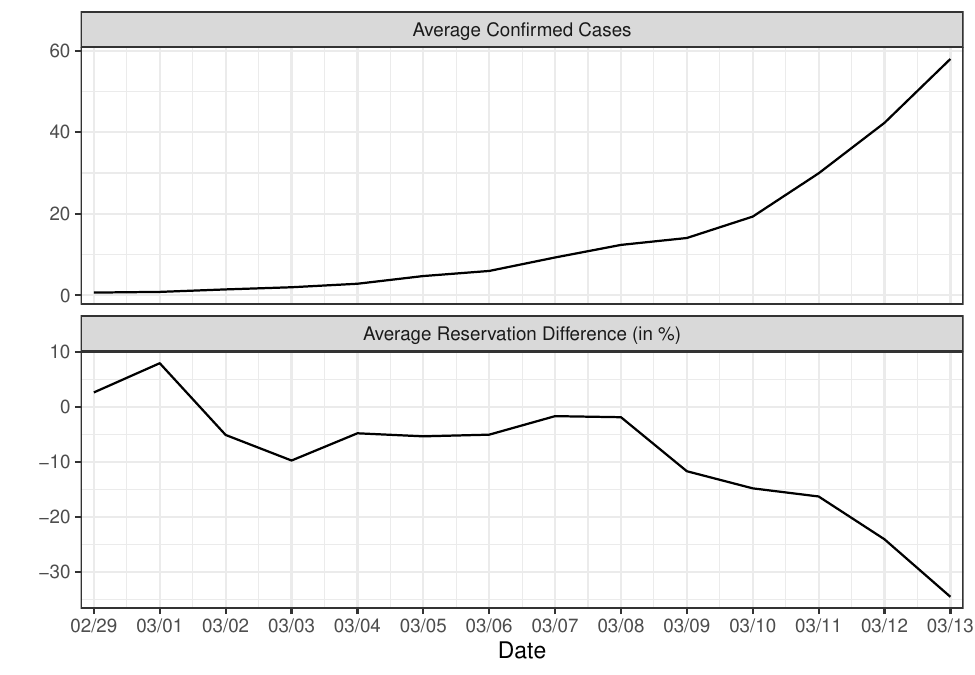}
\hfill
\includegraphics[width = 0.47\textwidth]{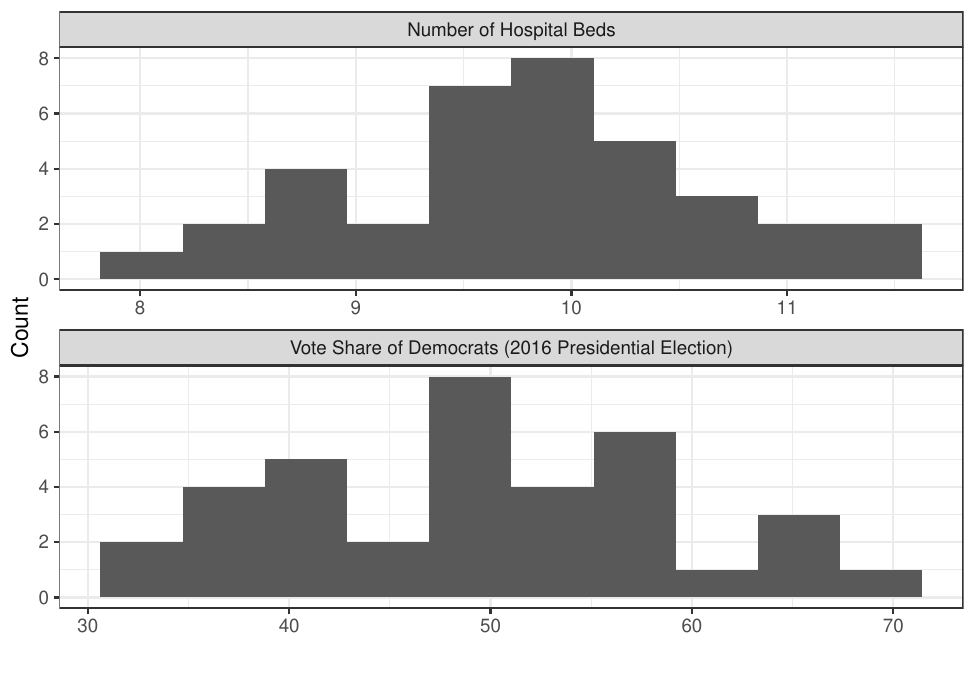}
\caption{(Left) daily average of the reservation difference and confirmed cases. (Right) histograms of number of hospital beds and vote share.}\label{fig:summary_statistics}
\end{figure}

\end{document}